\documentclass[sigconf, nonacm]{acmart}
\usepackage{amsmath,amsfonts,dsfont,multirow,multicol,epsfig,url,array,makecell,balance,color,epstopdf}
\usepackage{algorithmic}
\usepackage[ruled, vlined, linesnumbered]{algorithm2e}
\usepackage{hhline}
\usepackage{array}
\usepackage{enumerate}
\usepackage{enumitem}
\usepackage{subfigure}
\usepackage{booktabs}
\usepackage{xcolor,colortbl}
\usepackage{bm}
\usepackage{lipsum}
\usepackage{tabularx}
\usepackage{diagbox}
\setlist[itemize]{leftmargin=*}

\newtheorem{definition}{Definition}
\newtheorem{lemma}{Lemma}
\newtheorem{theorem}{Theorem}
\newtheorem{fact}{Fact}
\newtheorem{observation}{Observation}
\def\header{\vspace{1mm} \noindent}

\def\la{\langle}
\def\ra{\rangle}

\def\epi{\bm{\hat{\pi}}}
\def\r{\bm{r}}

\def\vpi{\bm{\pi}}

\def\p{\bm{p}}
\def\q{\bm{q}}
\def\L{\mathbf{L}}
\def\I{\mathbf{I}}
\def\P{\mathbf{P}}
\def\A{\mathbf{A}}
\def\D{\mathbf{D}}
\def\R{\mathbf{R}}
\def\Q{\mathbf{Q}}

\def\Dp{\mathbf{D}^{\frac{1}{2}}}
\def\Dn{\mathbf{D}^{-\frac{1}{2}}}

\def\C{\mathcal{C}}
\def\e{\varepsilon}
\def\p{\varphi}
\def\z{\bm{\zeta}}
\def\x{\bm{\chi}}
\def\k{\bm{\kappa}}

\def\hz{\color{black}}

\def\crc{\color{black}}

\def\localpush{{\em LocalPush}\xspace}
\def\edgepush{{\em EdgePush}\xspace}

\def\lpush{{\em push}\xspace}
\def\epush{{\em edge-based push}\xspace}

\def\FindMin{{\em FindMin}\xspace}
\def\ExtractMin{{\em ExtractMin}\xspace}
\def\DecreaseKey{{\em DecreaseKey}\xspace}

\def\nadd{{\em normalized MaxAddErr}\xspace}
\def\npre{{\em normalized precision@50}\xspace}
\def\add{{\em MaxAddErr}\xspace}
\def\pre{{\em precision@50}\xspace}
\def\lerr{{\em actual $\ell_1$-error}\xspace}

\newcommand\vldbdoi{10.14778/3523210.3523216}
\newcommand\vldbpages{1376-1389}
\newcommand\vldbvolume{15}
\newcommand\vldbissue{7}
\newcommand\vldbyear{2022}
\newcommand\vldbauthors{\authors}
\newcommand\vldbtitle{\shorttitle} 
\newcommand\vldbavailabilityurl{https://github.com/wanghzccls/EdgePush}
\newcommand\vldbpagestyle{plain}

\begin{document}
\title{Edge-based Local Push for Personalized PageRank}
\subtitle{[Technical Report]}

\author{Hanzhi Wang}
\affiliation{%
  \institution{Renmin University of China}
  \city{Beijing}
  \state{China}
}
\email{hanzhi\_wang@ruc.edu.cn}

\author{Zhewei Wei}
\authornote{Zhewei Wei is the corresponding author. \vspace*{1mm}}
\affiliation{%
  \institution{Renmin University of China}
  \streetaddress{1 Th{\o}rv{\"a}ld Circle}
  \city{Beijing}
  \country{China}
}
\email{zhewei@ruc.edu.cn}

\author{Junhao Gan}
\affiliation{%
  \institution{University of Melbourne}
  \city{Melbourne}
  \country{Australia}
}
\email{junhao.gan@unimelb.edu.au}

\author{Ye Yuan}
\affiliation{%
  \institution{Beijing Institute of Technology}
  \city{Beijing}
  \country{China}
}
\email{yuan-ye@bit.edu.cn}

\author{Xiaoyong Du}
\affiliation{%
  \institution{Renmin University of China}
  \city{Beijing}
  \country{China}
}
\email{duyong@ruc.edu.cn}

\author{Ji-Rong Wen}
\affiliation{%
  \institution{Renmin University of China}
  \city{Beijing}
  \country{China}
}
\email{jrwen@ruc.edu.cn}

\begin{abstract}
\label{abstract}
{\crc Personalized PageRank (PPR) is a popular node proximity metric in graph mining and network research. A single-source PPR (SSPPR) query asks for the PPR value of each node on the graph. Due to its importance and wide applications, decades of efforts have been devoted to the efficient processing of SSPPR queries. Among existing algorithms, \localpush is a fundamental method for SSPPR queries and serves as a cornerstone for subsequent algorithms.} In \localpush, a \lpush operation is a crucial primitive operation, which distributes the probability at a node $u$ to ALL $u$'s  neighbors via the corresponding edges. Although this \lpush operation works well on {\em unweighted} graphs, unfortunately, it can be rather inefficient on {\em weighted} graphs. In particular, on {\em unbalanced} weighted graphs where only a few of these edges take the majority of the total weight among them, the \lpush operation would have to distribute {``insignificant''} probabilities along those edges which just take the minor weights, resulting in expensive overhead. 

To resolve this issue, in this paper, we propose the \edgepush algorithm, a novel method for computing SSPPR queries on weighted graphs. \edgepush decomposes the aforementioned \lpush operations in \epush, allowing the algorithm to operate at the edge level granularity. As a result, it can flexibly distribute the probabilities according to edge weights. Furthermore, our \edgepush allows a fine-grained termination threshold for each individual edge, leading to a superior complexity over \localpush. Notably, we prove that \edgepush improves the theoretical query cost of  \localpush by an order of up to $O(n)$ when the graph's weights are {\em  unbalanced}. 
Our experimental results demonstrate that \edgepush significantly outperforms state-of-the-art baselines in terms of query efficiency {\hz on large motif-based and real-world weighted graphs.}

\end{abstract}


\maketitle

\vspace{-2mm}
\pagestyle{\vldbpagestyle}
\begingroup\small\noindent\raggedright\textbf{PVLDB Reference Format:}\\
\vldbauthors. \vldbtitle. PVLDB, \vldbvolume(\vldbissue): \vldbpages, \vldbyear. 
\href{https://doi.org/\vldbdoi}{doi:\vldbdoi}
\endgroup

\vspace{-1mm}
\begingroup
\renewcommand\thefootnote{}\footnote{
\noindent This work is licensed under the Creative Commons BY-NC-ND 4.0 International License. Visit \url{https://creativecommons.org/licenses/by-nc-nd/4.0/} to view a copy of this license. For any use beyond those covered by this license, obtain permission by emailing \href{mailto:info@vldb.org}{info@vldb.org}. Copyright is held by the owner/author(s). Publication rights licensed to the VLDB Endowment. \\
\raggedright Proceedings of the VLDB Endowment, Vol. \vldbvolume, No. \vldbissue\ %
ISSN 2150-8097. \\
}\addtocounter{footnote}{-1}
\endgroup

\ifdefempty{\vldbavailabilityurl}{}{
\begingroup\small\noindent\raggedright\textbf{PVLDB Artifact Availability:}\\
The source code, data, and/or other artifacts have been made available at \url{\vldbavailabilityurl}.
\endgroup
}

\vspace{-2mm}
\section{Introduction} 
\label{sec:intro}
Personalized PageRank (PPR), as a variant of PageRank~\cite{page1999pagerank}, has become a classic node proximity measure. It effectively captures the relative importance of all the nodes with respect to a source node in a graph. One particular interest is the single-source PPR (SSPPR) query. Given a source node $s$ in a graph $G=(V,E)$ with $n$ nodes and $m$ edges, the SSPPR query aims to return an SSPPR vector $\vpi \in \mathbb{R}^n$, where $\vpi(u)$ denotes the PPR value of node $u\in V$ with respect to the source node $s$. We can consider the SSPPR vector $\vpi$ as a probability distribution, with $\vpi(u)$ defined as the probability that an $\alpha$-random walk starting from the source node $s$ stops at node $u$. Specifically, the $\alpha$-random walk~\cite{page1999pagerank} represents a random walk process that at each step, the walk either moves to a random neighbor with probability $1-\alpha$, or stops at the current node with probability $\alpha$. The teleport probability $\alpha$ is a constant in $(0,1)$. 

SSPPR queries has been widely adopted in various graph mining and network analysis tasks. For example, the seminal local clustering paper~\cite{FOCS06_FS} and its variant~\cite{yin2017MAPPR, fountoulakis2019variational} identify clusters based on the SSPPR queries with the seed node as the source node. 
Additionally, the recommendations in social networks employ SSPPR values to evaluate the relative importance of other users regarding the target user, such as the Point-of-Interest recommendation~\cite{guo2017POI}, the connection prediction~\cite{backstrom2011supervised}, the topical experts finding application~\cite{lahoti2017expertfinding} and the Who-To-Follow recommendation in Twitter~\cite{gupta2013wtf}. 
Recently, several graph representation learning tasks~\cite{Bojchevski2020PPRGo,klicpera2019APPNP,chen2020GBP,zhou2004consistency} compute SSPPR queries to propagate initial node features in the graph. 

In this paper, we focus on efficient SSPPR queries on {\em weighted} graphs. Weighted graphs are extremely common in real life, where the weight of each edge indicates the distance, similarity or other strength measures of the relationships between two nodes. Various real applications are in dire need of the SSPPR results on weighted graphs. For instance, the personalized ranking results incorporating user preference or feedback embedded in the edge weight are highly valued in social network~\cite{feng2012incorporating,gao2011semi, xie2015edge}. Additionally, to rank web pages by SSPPR queries, taking into account the importance of pages' links shows increasingly significance for the performance of page ranking~\cite{xing2004weightedPageRank}. In the local clustering application, computing SSPPR queries on motif-based weighted graphs\footnote{A motif is defined as a small subgraph (e.g. a triangle). } can effectively capture the high-order information of network structure which is crucial to the clustering quality~\cite{yin2017MAPPR}. 

Despite the large-scale applications of SSPPR queries on weighted graphs, 
this topic are less studied in literature due to its hardness. 
The state-of-the-art algorithm is MAPPR~\cite{yin2017MAPPR}, which is a version of \localpush on weighted graphs. 
\localpush~\cite{FOCS06_FS} is a crucial and fundamental method for SSPPR queries, which has been regarded as a cornerstone method for advanced developments~\cite{wei2018topppr, wang2017fora, fountoulakis2019variational, wu2021SpeedPPR, hou2021MPC-PPR}. The main idea of \localpush is to approximate SSPPR results by deterministically {\em pushing} the probabilities on the graph. The {\lpush} operation in \localpush restricts the computation in a local manner, which achieves remarkable scalability on {\em unweighted graphs}. 
Unfortunately, although \localpush works well on unweighted graphs, it can be rather ineffective on weighted graphs, leading to excessive time consumption.

\begin{figure}[t]
\begin{center}
\centerline{\includegraphics[width=.85\linewidth,trim=10 105 10 100,clip]{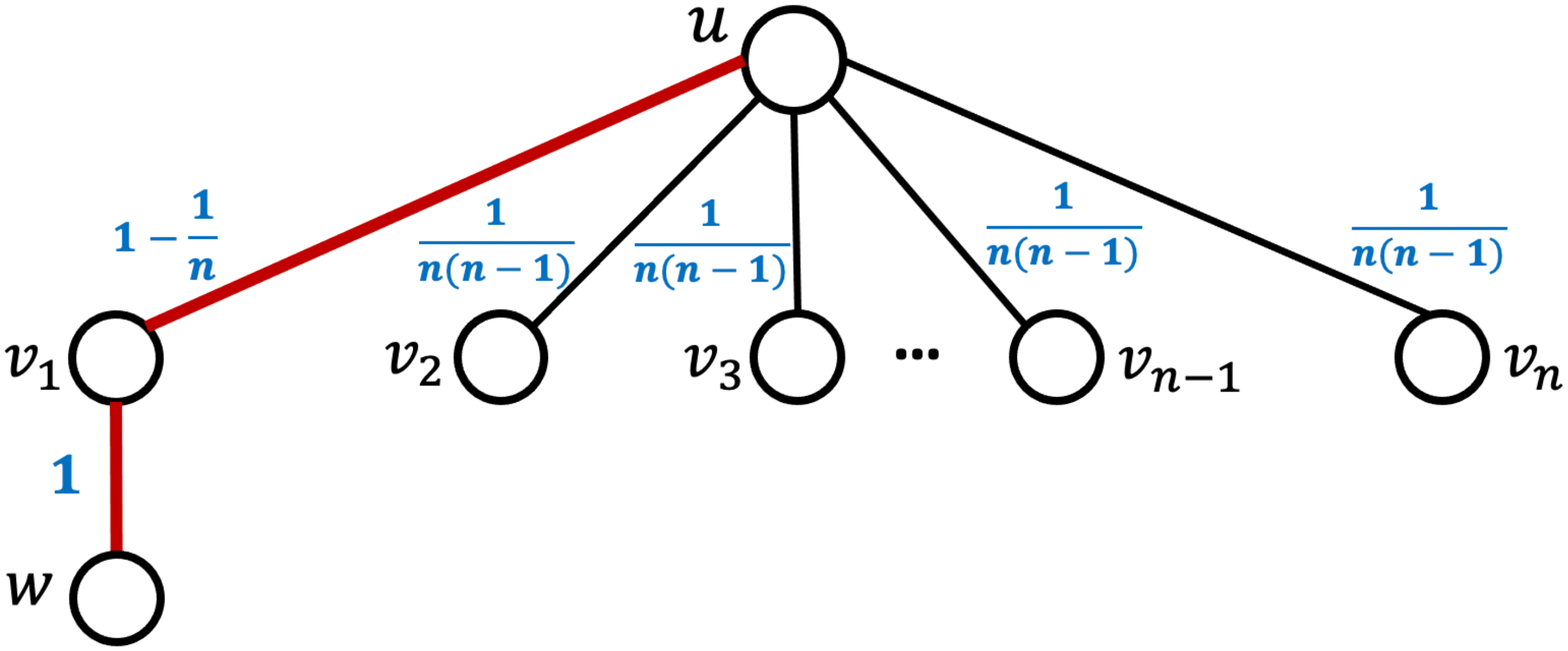}}
\vspace{-4.5mm}
\caption{A bad case for the \localpush. The number on each edge is the edge weight.}
\label{fig:weighted_graph}
\vspace{-4mm}
\end{center}
\vspace{-5mm}
\end{figure}

\header{\bf \localpush's limitation on weighted graphs.} 
As a crucial primitive operation in \localpush, the \lpush operation 
pushes probability mass from the current node to {\em all} its neighbors. Whenever the push operation on a node is invoked, it has to touch all the edges incident on the node. While this push strategy works fine on unweighted graphs, unfortunately, it has evident drawbacks on weighted graphs. 
When the weights of a node's edges are {\em unbalanced} where only a small number of edges taking a majority portion of the total weight among them, the push operation has to spend a significant cost on just pushing a tiny probability mass, resulting in severe overhead.

Figure~\ref{fig:weighted_graph} shows a toy example of the bad case for \localpush. Consider node $u$ whose total weight of edges is $1$. There is an edge $(u, v_1)$ taking a weight $1 - \frac{1}{n}$, merely the total weight, and all the others just share $\frac{1}{n}$ together. When a \lpush operation on $u$ is performed, it requires a cost of $n-1$ just on pushing an extremely tiny probability mass for those ``insignificant'' edges. {\hz As a result, the push operation is extremely inefficient on such severely unbalanced weighted graphs. }

{\crc
It's worth to mention that weighted graphs with severely unbalanced edge weights are common in many real-world applications. Let's take the affinity graph as an example. 
Affinity graphs are frequently used in a variety of practical tasks~\cite{zhou2004consistency, ng2002spectral,ye2020affinity,zhu2014constructing,wang2020affinity,yang2019learning,yadav2021cross,yang2019relationship} to model the affinities between pairwise data points. 
Nodes in affinity graphs represent high dimensional data points, i.e. $V=\{x_1, ..., x_n\}$, where $x_i\in \mathbb{R}^{\k}$. Edges are fully connected and weighted, and the weight of edge $(x_i, x_j)$ indicates the affinity between data points $x_i$ and $x_j$, defined as 
$\A_{ij}=\exp\left(-\|x_i-x_j\|^2/2\sigma^2\right)$. 
Here $\|x_i-x_j\|$ denotes the Euclidean distance between data points $x_i$ and $x_j$, and $\sigma^2$ denotes the variance of all data points in $V$. We note that the value of distance $\|x_i-x_j\|^2$ is exponential to the edge weight $\A_{ij}$. Thus, small differences among pairwise distances can lead to significantly-skewed edge weights distribution. On the other hand, computing PPR values on affinity graphs is a commonly adopted technique in various tasks, such as label propagation~\cite{zhou2004consistency}, spectral clustering~\cite{ye2020affinity}, image segmentation~\cite{yang2019learning} and relationship profiling~\cite{yang2019relationship}. Therefore, to apply \localpush for PPR computation on such heavily unbalanced weighted graphs can invoke expensive but unnecessary time cost. } 

\header{\bf Our contributions. } 
To remedy the above issue of \localpush on weighted graphs, we make the following contributions: 
\begin{itemize}[leftmargin = *]
\vspace{-3mm}
\item \underline{Edge-based Push Method. } We propose \edgepush, a novel edge-based push method for SSPPR queries. Our \edgepush further decomposes the aforementioned atomic \lpush operation of \localpush into {\em separate} \epush operations. As a result, \edgepush can flexibly select edges to push probability mass based on the edge weights. 
\item \underline{Theoretical Analysis.} 
	\edgepush admits a fine-grained {\em individual} termination threshold $\theta(u,v)$ for each edge. 
	With careful choices of $\theta(u,v)$, \edgepush achieves superior query efficiency in terms of the trade-offs between the approximation error and the expected overall running time. 
	In this paper, we analyze the time complexity of \edgepush and present the suggested choice of $\theta(u,v)$ with two specific error measurements:  the $\ell_1$-error and the normalized additive error. 
	In particular, when the edge weights are {unbalanced} (as shown in Figure~\ref{fig:weighted_graph}), with the optimal setting of $\theta(u,v)$, \edgepush can approximate PPR values in time $o(m)$, {\em sub-linear} to the number of edges, with specified $\ell_1$-error.
	In other words, in this case, we can solve the approximate PPR with $\ell_1$-error even without touching every edge in the graph.

\item \underline{Superiority Illustration.} 
We demonstrate that \edgepush achieves a superior expected time complexity over \localpush on arbitrary graphs as shown in Table~\ref{tbl:intro-compare}. 
For the ease of illustration,
here we present superior results for a relatively restricted case, where all the nodes in the graph are {\em $(a, b)$-unbalanced} (which is defined next). However, it should be noted that as proved in Section~\ref{sec:analysis}, the conditions for \edgepush strictly outperforming \localpush 
are actually more general and less restrictive.
Specifically, the notion of $(a,b)$-unbalanced node is used to quantify the unbalancedness of the weighted graph. 
A node is said to be $(a,b)$-unbalanced if $a$ fraction of its adjacency edges take $b$ fraction of its edge weights, where $0\le a \le b \le 1$. Based on the $(a,b)$-unbalanced definition, we summarize three theoretical implications in the following. Here we assume the source node is chosen according to the node degree distribution.
    \begin{itemize}
        \item {\hz The overall running time bound of \edgepush is no worse than that of \localpush even on unweighted graphs, regardless of whether the $\ell_1$-error or the normalized additive error. }
        \item {\hz When the edge weights are unbalanced, \edgepush achieves superior query efficiency over \localpush. And the superiority of \edgepush over \localpush can be quantified by the unbalancedness of the weighted graphs.}
	    \item When the graph $G$ is a complete graph with $n$ nodes, $a = 1 / n$ and $b = 1 - 1/n$, \edgepush outperforms \localpush by a $O(n)$ factor, both for the $\ell_1$ and normalized additive error. 
	    \item When $a = o(1)$ and $b = 1 - o(1)$, \edgepush compute SSPPR queries in time {\em sub-linear} to the number of edges in the graph with any specified $\ell_1$-error. 
    \end{itemize}
\item \underline{Extensive Experiments.} {\hz We conduct comprehensive experiments to show the effectiveness of our \edgepush  on both motif-based and real-world weighted graphs}. The experimental results demonstrate that when achieving the same approximation error (e.g. normalized additive error or $\ell_1$-error), \edgepush outperforms \localpush on large real-world graphs by orders of magnitude in terms of query efficiency. 
	{\hz Notably, even on the graphs with less unbalanced edge weights, \edgepush still shows superior performances over the state-of-the-art baselines. } 
\end{itemize}

\begin{table*} [t]
	\centering
	\renewcommand{\arraystretch}{1.5}
	\begin{small}
		\caption{The comparison between the complexities of \localpush and \edgepush. The ``Improvements" column quantifies the superiority of \edgepush over \localpush in terms of the expected time complexities when the source node is chosen according to the degree distribution. $\p$ and $\p_v$ denote specific angles on weighted graphs, which are formally illustrated in Section~\ref{sec:analysis}. } \label{tbl:intro-compare}
		\vspace{-3mm}
		\begin{tabular}{|@{\hspace{+1.5mm}}c@{\hspace{+1.5mm}}|@{\hspace{+1.5mm}}c@{\hspace{+1.5mm}}|@{\hspace{+1.5mm}}c@{\hspace{+1.5mm}}|@{\hspace{+1.5mm}}c@{\hspace{+1.5mm}}|} \hline
			~& \localpush & \edgepush & Improvements \\ \hline
			$\ell_1$-error $\e$ & $O\left(\frac{m}{\alpha \e}\right)$ & $O\left(\frac{(1-\alpha)}{\alpha \e \|\A\|_1}\cdot \left(\sum\limits_{\la u,v \ra\in E}\hspace{-2mm}\sqrt{\A_{uv}}\right)^2 \right)=O\left(\left((1-\alpha)\cos^2\p\right) \cdot \frac{m}{\alpha \e}\right)$   & $(1-\alpha)\cos^2\p$ \\ \hline
			normalized additive error $r_{\max}$& $O\left(\frac{m}{\alpha r_{\max} \cdot \|\A\|_1}\right)$ & $O\left(\frac{(1-\alpha)}{\alpha r_{\max} \|\A\|_1}\hspace{-0.5mm}\cdot \hspace{-0.5mm} \sum\limits_{v\in V}\hspace{-1mm}\frac{\left(\sum_{x\in N(v)}\hspace{-1mm}\sqrt{\A_{xv}}\right)^2 }{d(v)}\right)=O \left( \left(\frac{1-\alpha}{m}\hspace{-0.5mm}\cdot \hspace{-1mm} \sum\limits_{v\in V}\hspace{-1mm}n(v) \cos^2 \p_v \right) \cdot \frac{m}{\alpha r_{\max} \|\A\|_1}\right)$   & $\frac{(1-\alpha)}{2m}\hspace{-0.5mm}\cdot \hspace{-1mm} \sum\limits_{v\in V}\hspace{-0.5mm}n(v) \cos^2 \p_v$\\ \hline
		\end{tabular}
	\end{small}
	\vspace{-2mm}
 \end{table*}

\vspace{-2mm}
\section{Preliminaries} 
\label{sec:pre}
\header
{\bf Notations.}
Consider an {\em undirected} and {\em weighted} graph $G=(V,E)$ with $|V| = n$ nodes and $|E|=m$ edges. We define $\bar{E}$ as the set of bi-directional edges of $G$, that is, for every edge $(u,v) \in E$, there are two directed edges $\la u,v\ra$ and $\la v, u\ra$ in $\bar{E}$, and these two edges are treated differently. We use $\A$ to denote the {\em adjacency matrix} of graph $G$, and $\A_{uv}$ to denote the {\em weight} of edge $\la u, v\ra \in \bar{E}$. 
{\hz Furthermore, we assume that each $\A_{uv}$ is a non-negative real number. }
For $\forall \la u,v \ra \not\in \bar{E}$, we have $\A_{uv}=0$. As a result, $\|\A\|_1 = \sum_{\la u,v \ra \in \bar{E}} \A_{uv}$ denotes the total weights of all edges. For every edge $\la u,v \ra \in \bar{E}$, we say $v$ is a {\em neighbor} of $u$. For each node $u\in V$, we denote the set of all the {\em neighbors} of $u$ by $N(u)$, and $n(u)=|N(u)|$, the neighborhood size of $u$. Moreover, $d(u)=\sum_{v\in N(u)}\A_{uv}$ denotes the ({\em weighted}) {\em degree} of node $u$, and $\D$ denotes the diagonal degree matrix with $\D_{uu}=d(u)$. Finally, the {\em transition matrix} is denoted by $\P=\A \D^{-1}$.

In this paper, we use $\vpi_s\in \mathbb{R}^n$ to denote the SSPPR vector w.r.t node $s$ as the source node. The $u$-th coordinate $\vpi_s(u)$ records the PPR value of node $u\in V$ w.r.t $s$. Unless specified otherwise, we denote node $s$ as the source node by default and omit the subscript in $\vpi_s$ and $\vpi_s(u)$ for short (i.e. $\vpi$ and $\vpi(u)$). In Section~\ref{sec:related}, we use $\r\in \mathbb{R}^{n}$ and $\epi\in \mathbb{R}^{n}$ to denote the residue and reserve vectors in \localpush, respectively. In Section~\ref{sec:algorithm}, we define three variables: node income vector $\q\in \mathbb{R}^n$, edge expense matrix $\Q\in \mathbb{R}^{n \times n}$ and edge residue matrix $\R\in \mathbb{R}^{n\times n}$ for the \edgepush algorithm. 

\header{\bf The word RAM model. }
The word RAM model (word random-access machine) is first proposed by Fredman et al.~\cite{fredman1993surpassing} to simulate the actual executions of programs implemented by well-adopted programming languages, e.g., C and C++, in real-world computers. The word RAM model assumes a random-access machine can do bitwise operations on a single word of $w$ bits. The value of $w$ is related to the problem size. 
Thus, in the word RAM model, basic operations on words such as arithmetic, comparison and logical shifts can be performed in constant time. 

In this paper, we analyze all of the theoretical complexities under the word RAM model.  Following the aforementioned convention of the model,
we assume that every the numerical value, such as the edge weight $\A_{uv}$ or the constant teleport probability $\alpha$ in PPR computation, can fit into $O(1)$ words of $O(w)$ bits. 
As the problem size studied in this paper is upper bounded by the input graph size $O(n + m) = O(n^2)$, every numerical value can be thus represented by $O(\log n)$ bits.
We note that this assumption is mild and realistic because most (if not all) real-world computers can only support computations on floating-point numbers up to limited precision. 

\header{\bf Single-source Personalized PageRank (SSPPR).} PageRank~\cite{page1999pagerank} is first proposed by Google to rank the overall importance of nodes in the graph. Personalized PageRank (PPR) is a variant of PageRank, which evaluates each node's {\em relative} importance w.r.t a given source node. The single-source PPR (SSPPR) query is a type of PPR computations, which aims to return all the PPR values (w.r.t the source node) in the graph. More precisely, given node $s$ as the source node, the SSPPR query aims to derive an SSPPR vector $\vpi \in \mathbb{R}^n$, where the $u$-th coordinate $\vpi(u)$ represents the PPR value of node $u$. 

In the seminal paper of PPR~\cite{page1999pagerank}, the SSPPR vector $\vpi$ w.r.t source node $s$ is defined as the solution to the recursive equation: 
\begin{align}
\vspace{-3mm}
\label{eqn:ppr_it}
\vpi =  (1-\alpha) \P\vpi +\alpha\bm{e}_s,  
\vspace{-4mm}
\end{align}
where $\alpha\in (0,1)$ is a constant teleport probability, $\P$ is the transition matrix that $\P=\A \D^{-1}$, and $\bm{e}_s$ is an one-hot vector that $\bm{e}_s(s)=1$ and $\bm{e}_s(u)=0$ if $u\neq s$. 
By applying a power series expansion~\cite{avrachenkov2007monte}, the SSPPR vector $\vpi$ can be derived as: 
\begin{align}
\label{eqn:ppr_expansion}
\vspace{-3mm}
\vpi=\alpha \cdot \left(\I-(1-\alpha)\P\right)^{-1}\cdot \bm{e}_s=\sum_{i=0}^\infty \alpha(1-\alpha)^i \P\cdot \bm{e}_s.
\end{align}
{\crc Note that this expansion provides an alternative interpretation of PPR values: 
the SSPPR vector $\vpi$ can be regarded as a probability distribution, where $\vpi(u)$ denotes the probability that an $\alpha$-random walk from the given source node $s$ stops at $u$~\cite{page1999pagerank,lofgren2015PHDthesis}. Each step of an $\alpha$-random walk either stops at the current node with probability $\alpha$, or stays {\em alive} to move forward with $(1-\alpha)$ probability. Specifically, if an $\alpha$-random walk is currently alive at node $u$, then in the next step, the walk will move to one of $u$'s neighbors $v\in N(u)$ with the probability proportional to the edge weight $\A_{uv}$, i.e., with probability $\frac{(1 - \alpha)\cdot \A_{uv} }{d(u)}$. }

\header{\bf Problem definition.} 
As shown in Equation~\eqref{eqn:ppr_expansion}, the exact computation of SSPPR vector involves the inverse of the $n\times n$ matrix:  $(\mathbf{I}-(1-\alpha)\P)$, where $n$ is the number of graph nodes. This is infeasible on large graphs with millions of nodes. 
Thus, in this paper, we aim to approximate the SSPPR vector $\vpi$ on large graphs with specified approximation error. Specifically, we consider the following two problems:
\vspace{-1mm}
\begin{definition} [Approximate SSPPR with normalized additive error] 
	Given an undirected and weighted graph $G=(V,E)$, a source node $s\in V$ and a normalized additive error tolerance $r_{\max}\in (0,1)$, the goal of an approximate SSPPR query (w.r.t the source node $s$) with normalized additive error is to return an estimated SSPPR vector $\epi$ such that for each $u\in V$, $\left|\frac{\vpi(u)}{d(u)} - \frac{\epi(u)}{d(u)} \right| \le r_{\max}$.
\end{definition}

The normalized additive error is a commonly used evaluation metric in local clustering tasks. 
More precisely, the majority of existing local clustering algorithms~\cite{FOCS06_FS, fountoulakis2019variational,yin2017MAPPR, yang2019TEA, Teng2004Nibble, chung2018computing, chung2015distributed, kloster2014heat} operate in two stages: first, they treat the given seed node as the source node and calculate the approximate SSPPR vector $\epi$ (or other scores to rank nodes' relative importance w.r.t the source node). Then the they feed the vector $\epi$ in a {\em sweep} process to identify local cluster around the seed node. The detailed steps in the sweep process are given below: 
\begin{itemize}[leftmargin = *]
	\item (i) Put all the nodes with non-zero $\frac{\epi(u)}{d(u)}$ into a set $S$. 
	\item (ii) Sort each node $u\in S$ in the descending order by $\frac{\epi(u)}{d(u)}$, such that $S=\{v_1, v_2, ..., v_j\}$ and $\frac{\epi(v_1)}{d(v_1)}\ge \frac{\epi(v_2)}{d(v_2)} \ge ... \ge \frac{\epi(v_j)}{d(v_j)}$. 
	\item (iii) Scan the set $S$ from $v_1$ to $v_j$ and find the subset with minimum {\em conductance} among all the partial sets $S_i=\{v_1, v_2, ..., v_i\}$ for $i=1,2,...,j$. 
\end{itemize}

\noindent
In the third step, we calculate the {\em conductance} of all partial sets. Conductance is a popular measure to evaluate the cluster quality. More precisely, given a cluster set $S_i \subseteq V$, the conductance of set $S_i$ is defined as $\Phi(S_i)=\frac{cut(S_i)}{\min \{vol(S_i),vol(V \setminus S_i)\}}$, where $vol(S_i)=\sum_{u\in S_i}d(u)$ denotes the volume of set $S_i$, and $cut(S_i)$ denotes the sum of edge weights for those edges crossing $S_i$ and $V\setminus S_i$. Thus, the conductance values are the smaller, the better. 

Reviewing the sweep process, we note that the quality of the identified clusters heavily depends on the approximation accuracy of the normalized PPR values (i.e. $\frac{\vpi(u)}{d(u)}$ for each $u$). 
Therefore, in this paper, we introduce normalized additive error as one of the evaluation criteria.  
Additionally, we employ a classic evaluation metric, $\ell_1$-error, to assess the approximation quality of each algorithm: 
\vspace{-1mm}
\begin{definition} [Approximate SSPPR with $\ell_1$-error] 
	Given an undirected weighted graph $G=(V,E)$, a source node $s\in V$, a constant teleport probability $\alpha$, and an $\ell_1$-error tolerance $\e\in (0,1)$, 
	the goal of an approximate PPR query with respect to $s$ is to return an estimated PPR vector $\epi$ such that $\|\epi - \vpi\|_1 = \sum\nolimits_{u\in V}|\epi(u)-\vpi(u)|\le \e$. 
\end{definition}

\begin{algorithm}[t]
\caption{The \localpush Algorithm~\cite{yin2017MAPPR}}
\label{alg:APPR}
\BlankLine
\KwIn{undirected and weighted Graph $G=(V,E)$ with adjacency matrix $\A$, source node $s$, constant teleport probability $\alpha \in (0,1)$, termination threshold  $\theta$\\}
\KwOut{an estimation $\epi$ of SSPPR vector $\vpi$ w.r.t $s$\\}
   
$\bm{\epi}\gets \bm{0}$, $\bm{r}\gets \bm{e}_s$\; 
\While{there exists a node $u$ with $\r(u)\ge d(u) \cdot \theta$}{
    $\epi(u)\gets \epi(u)+\alpha\cdot \r(u)$;\\
	\For{every neighbors $v\in N(u)$}{
		$\r(v)\gets \r(v)+(1-\alpha)\r(u) \cdot \frac{\A_{uv}}{d(u)}$\;
	}
	$\r(u)\gets 0$;\\
}
\Return $\epi$ as the estimator for $\vpi$;
\end{algorithm}

\subsection{The \localpush Algorithm}\label{subsec: localpush_alg}
Among existing algorithms for SSPPR queries, the \localpush algorithm~\cite{FOCS06_FS,yin2017MAPPR} is a fundamental method which serves as a cornerstone in various subsequent algorithms~\cite{wang2017fora, wu2021SpeedPPR, fountoulakis2019variational, wei2018topppr, hou2021MPC-PPR}. The basic idea of \localpush is to ``{\em simulate}'' $\alpha$-random walks in a deterministic way by {\em pushing} the probability mass from a node to its neighbors. More specifically, given an undirected and weighted graph $G=(V,E)$, a source node $s$, a constant teleport probability $\alpha \in (0,1)$ and a global termination threshold $\theta$, \localpush maintains two variables for each node $u \in V$ during the executing process: 
\begin{itemize}[leftmargin = *]
\item \underline{{Residue} $\r(u)$}: the probability mass currently on $u$ and will be distributed to other nodes. Alternatively, in an $\alpha$-random walk process, $\r(u)$ records the probability mass of the $\alpha$-random walk from $s$ alive at $u$ at the current state. Note that if a walk has not yet stopped, we say the walk is {\em alive} at the current node; 
\item \underline{{Reserve} $\epi(u)$}: the probability mass that stays at node $u$. $\epi(u)$ is an underestimate of the real PPR value $\vpi(u)$. 
\end{itemize}


\header{\bf The \lpush operation.} 
\lpush is a critical primitive operation that is repeated throughout the \localpush process. A \lpush operation consists of three steps (see the left side of Figure~\ref{fig:comparison} for illustration): 
\begin{itemize}[leftmargin = *]
	\item convert $\alpha$ portion of $\r(u)$ to $\epi(u)$, i.e., $\epi(u) \leftarrow \epi(u) + \alpha \cdot \r(u)$, simulating the fact that $\alpha$ portion of the random walk alive at $u$ has stopped here at $u$;
	\item distribute the rest $(1 - \alpha)$ portion of $\r(u)$ to each neighbor $v$ proportional to the corresponding edge weight by increasing the residue of $v$, i.e., $\r(v) \leftarrow \r(v) + (1 - \alpha)\r(u) \cdot \frac{\A_{uv}}{d(u)}$; this essentially simulates that $(1-\alpha)$ portion of the random walk will move one step forward to each $u$'s neighbor with the probability proportional to the edge weights;
	\item reset the residue $\r(u)$ to $0$, indicating that, after the above two steps, there is no $\alpha$-random walk alive at $u$ at the moment.
\end{itemize}

\header{\bf Invariant.} The analysis of localpush algorithm is built upon an invariant between the residue and the reserve, which is formalized in the following lemma:

\begin{lemma}[Invariant by LocalPush]\label{lem:invariant_localpush}
For each node $t$ in the graph, the reserve $\vpi_t$ and residues satisfy the following invariant after each push operation: 
\begin{align}\label{eqn:invariant-nodepush}
    \vpi(t)=\epi(t)+\sum_{u\in V}\r(u) \cdot \vpi_u(t), 
\end{align}
where $\vpi(t)$ denotes the PPR value of node $t$ w.r.t the source node $s$ (by default) and $\vpi_u(t)$ denotes the PPR value of $t$ w.r.t node $u$. 
\end{lemma}

We defer the proof of Lemma~\ref{lem:invariant_localpush} to the appendix. 
Intuitively, $\epi(t)$ is the probability mass that stays at $t$. Aside from $\epi(t)$, $\vpi(t)$ also includes probability mass that is currently on other nodes and will be delivered to $t$. To calculate this part, recall that $\r(u)$ refers to the probability mass that is currently at $u$ but will be distributed to other nodes. Given that $\vpi_u(t)$ is the probability of a random walk starting at $u$ and ending at $t$, $\r(u)\cdot \vpi_u(t)$ is the probability mass now residing at $u$ and to be distributed to $t$. Summing over all nodes $u \in V$ and the lemma follows. For a illustration, see the left part of  Figure~\ref{fig:invariant}.

\begin{figure}[t]
\vspace{-3mm}
\begin{center}
\hspace{+0mm}\centerline{\includegraphics[width=80mm,trim=10 85 10 60,clip]{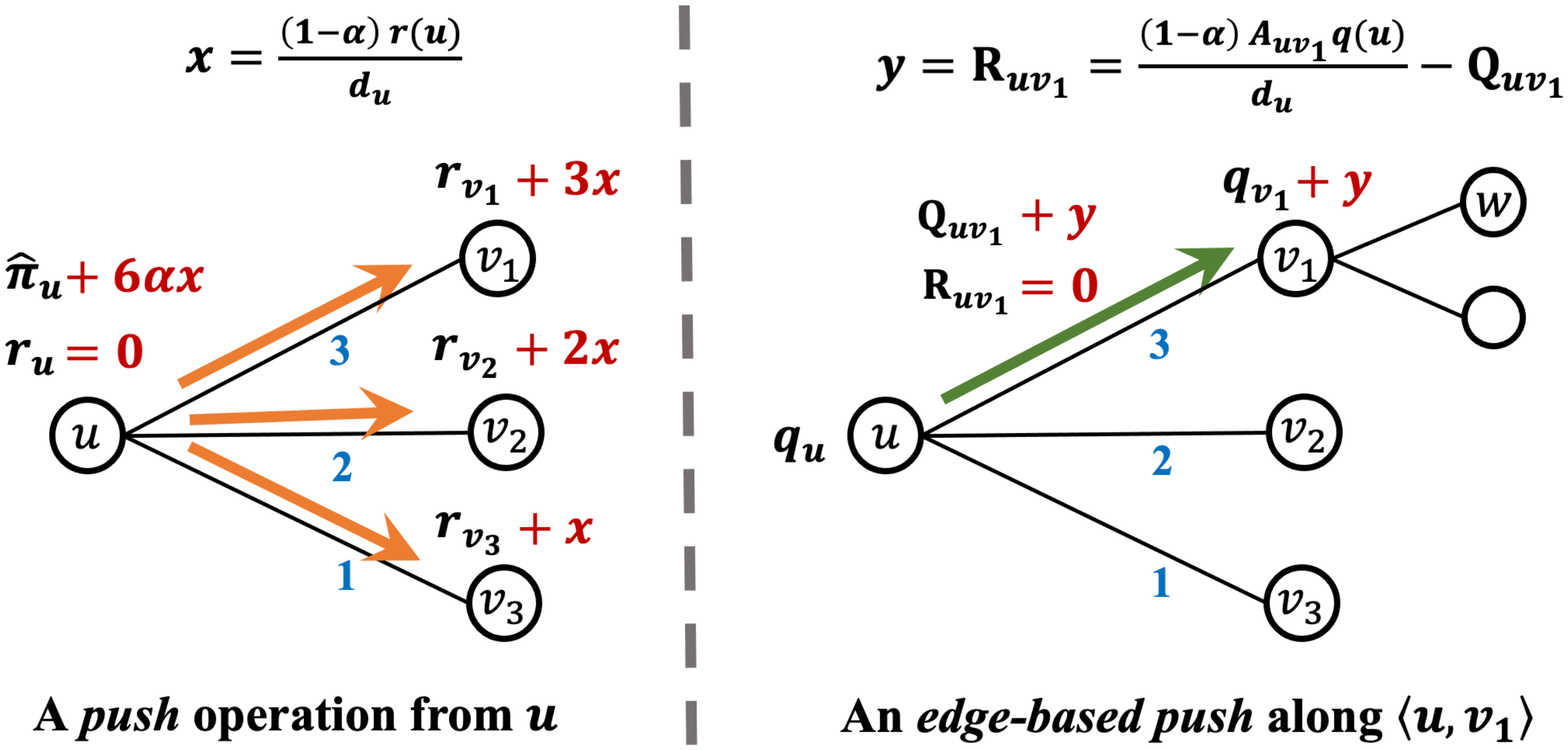}}
\vspace{-9mm}
\caption{A sketch for comparing \lpush and \epush. The number on each edge is the edge weight.}
\label{fig:comparison}
\end{center}
\vspace{-3mm}
\end{figure}

\begin{figure}[t]
\vspace{-3mm}
\begin{center}
\hspace{+0mm}\centerline{\includegraphics[width=80mm,trim=10 105 10 110,clip]{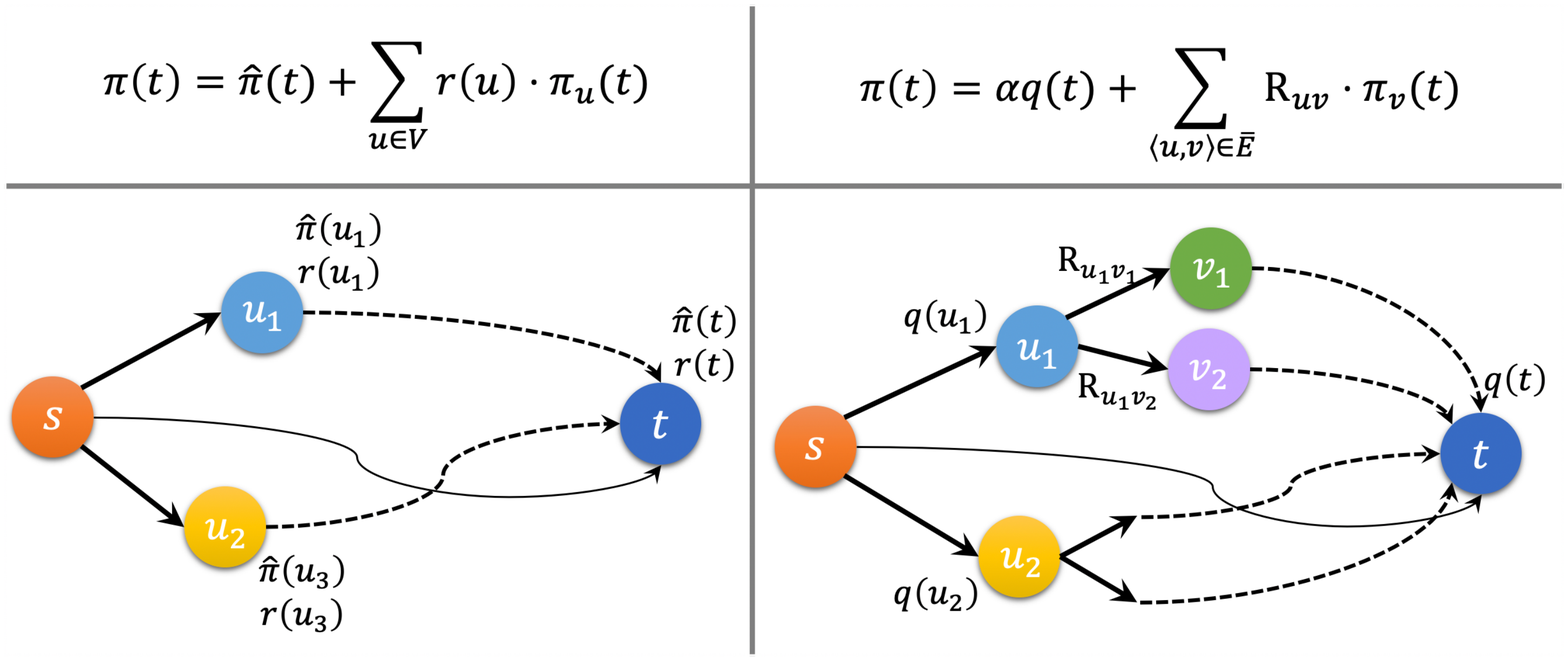}}
\vspace{-9mm}
\caption{A sketch for comparing the invariants of \localpush and \edgepush. }
\label{fig:invariant}
\vspace{-5mm}
\end{center}
\end{figure}


\header{\bf Algorithm descriptions of \localpush. } 
Initially, the node residue vector $\r$ is initialized as $\bm{r} = [\r(u)]_{u\in V} = \bm{e}_s$, indicating that, at the beginning, the walk is only alive at the source node $s$ (with probability $1$). Meanwhile, $\epi = [\epi(u)]_{u\in V} = \bm{0}$. 
During the \localpush process, \localpush repeatedly performs the \lpush operations until there is no node $u$ with residue $\r(u) \geq d(u) \cdot \theta$. After the termination, \localpush returns the reserve vector $\epi$ as the approximation of SSPPR vector $\vpi$. The pseudo code of the \localpush algorithm is illustrated in Algorithm~\ref{alg:APPR}.

\header{\bf Error analysis and time complexity of \localpush. } 
On unweighted graphs, it is known that \localpush costs $O\left(\frac{m}{\alpha \e}\right)$ and $O\left(\frac{1}{\alpha r_{\max}}\right)$ to answer the approximate SSPPR queries with $\ell_1$-error $\e$ and normalized additive error $r_{\max}$, respectively~\cite{FOCS06_FS, wang2017fora, wu2021SpeedPPR}.   
With a slight modification, we can derive the time complexities of \localpush on weighted graphs: 



\begin{fact} \label{thm:bound-local}
By setting the termination threshold $\theta$ in Algorithm~\ref{alg:APPR} as $\theta = \e/\|\A\|_1$, LocalPush answers the SSPPR query with an $\ell_1$-error $\varepsilon$. 
When the source node is randomly chosen according to the degree distribution, the expected cost of LocalPush is bounded by $O\left(\frac{m}{\alpha \e}\right)$. 
\end{fact} 

\begin{fact} \label{thm:bound-local-addErr}
By setting the termination threshold $\theta$ in Algorithm~\ref{alg:APPR} as $\theta = r_{\max}$, LocalPush returns an approximation of SSPPR vector with normalized additive error $r_{\max}$.  
When the source node is randomly chosen according to the degree distribution, the expected cost of LocalPush is bounded by $O\left(\frac{m}{\alpha r_{\max} \cdot \|\A\|_1}\right)$. 
\end{fact}

Considering Fact~\ref{thm:bound-local-addErr}, if the number of edges $m$ equals the total weights $\|\A\|_1$, which always holds for unweigthed graphs, the overall running time bound can be simplified to $O(\frac{1}{\alpha r_{\max}})$.
For the sake of readability, we defer the proofs of Fact~\ref{thm:bound-local} and Fact~\ref{thm:bound-local-addErr} to the appendix. 
Note that in these two facts, we assume the source node is randomly chosen according to the degree distribution for the sake of simplicity. This is also a common assumption in the context of local clustering application (i.e. the seed node is sampled according to the node degree). 

\vspace{-1mm}
\section{Related Work} 
\label{sec:related}

\header{\bf Power Method.} {\crc Power Method~\cite{page1999pagerank} is an iterative method to compute the SSPPR vector $\vpi$. 
Recall that in Equation~\eqref{eqn:ppr_expansion}, we present a power series expansion to compute $\vpi$, which alternatively suggests an iterative algorithm: 
\begin{align}\label{eqn:powerdef}
\vspace{-4mm}
    \vpi^{(\ell+1)}=(1-\alpha) \P \vpi^{(\ell)} +\alpha\bm{e}_s, 
\vspace{-4mm}
\end{align}
where $\vpi^{(\ell)}$ denotes the estimated SSPPR vector after the $\ell$-th iteration. 
Power Method employs the recursive formula $L$ times and regards $\vpi^{(L)}$ as the approximation of $\vpi$, leading to an $O(mL)$ time complexity. By setting $L=\log{\frac{1}{\e}}$ and $L=\log{\frac{1}{r_{\max}}}$, Power Method can achieve an $\ell_1$ error of $\e$ and a normalized additive error of $r_{\max}$, respectively. 
Note that the settings of $L$ show logarithmic dependence on the error parameters (i.e. $\e$ or $r_{\max}$), which enables Power Method to answer high-precision SSPPR queries. However, in each iteration, Power Method needs to touch the every edge on the graph, resulting in a $\Theta(m)$ time cost per iteration. This can severely limit the scalability of Power Method on large graphs. 


\header{\bf Monte-Carlo sampling.} Monte-Carlo sampling~\cite{jeh2003scaling, fogaras2005MC} estimates SSPPR vector $\vpi$ by simulating the random walk process. More precisely, recall that the PPR value $\vpi(u)$ of node $u$  equals to the probability that an $\alpha$-random walk from the given source node $s$ stops at node $u$. Based on this interpretation, Monte-Carlo sampling first generates multiple $\alpha$-random walks from the source node $s$, then uses the fraction that the number of walks terminates at $u$ as an approximation of $\vpi(u)$. However, due to the uncertainty of random walks, 
Monte-Carlo method is inefficient to achieve small approximation error on large graphs~\cite{wang2017fora}. 

\header{\bf FORA.} Existing methods~\cite{lin2020ResAcc, lofgren2013personalized, wang2016hubppr, wang2017fora} adopt various approaches to improve the efficiency of Monte-Carlo Sampling. 
As a representative algorithm, FORA~\cite{wang2017fora} combines \localpush with Monte-Carlo Sampling to approximate the SSPPR vector. By theoretical analysis, FORA provides the optimal settings of the error parameters in \localpush and Monte-Carlo sampling to balance the two phases. 
Additionally, FORA introduces an index scheme, as well as a module for top-$k$ selection with high pruning power, which, however, is out of the scope of this paper. 

\header{\bf PowForPush and SpeedPPR. } 
Recently, Wu et al. ~\cite{wu2021SpeedPPR} proposes PowForPush to accelerate the efficiency of high-precision SSPPR queries. PowForPush is based on the Power Method and adopts two optimization techniques, sequential scanning and dynamic $\ell_1$-error threshold, for better performance from an engineering point of view. 
However, the time complexity of PowForPush is still the same as that of Power Method, which still has a linear dependency on the number of edges $m$. 

As a variant of PowForPush, SpeedPPR~\cite{wu2021SpeedPPR} combines PowForPush with Monte-Carlo Sampling for efficient approximation of SSPPR queries. As shown in~\cite{wu2021SpeedPPR}, SpeedPPR achieves superior time complexity over FORA under relative error constraints. 
The success of SpeedPPR demonstrates that \localpush serves as a cornerstone approach for SSPPR queries, and hence, any improvement on \localpush can be applied to the subsequent methods for advanced developments.  }

\header{\bf Other related work.} Except for SSPPR  queries, there are other lines of researches concerning single-target PPR queries~\cite{andersen2007contribution, wang2020RBS, lofgren2013personalized, andersen2008robust}, single-pair PPR queries~\cite{lofgren2014FastPPR, wang2016hubppr, lofgren2015bidirectional, lofgren2016BiPPR,fujiwara2012efficient}, distributed PPR queries~\cite{hou2021MPC-PPR,lin2019distributed, zhu2005distributed, luo2019distributed, sarma2013fast}, or SSPPR computation on dynamic graphs~\cite{lofgren2016approximate, yu2016random, ohsaka2015efficient}. However, these works are orthogonal to the focus of this paper. 

\vspace{-2mm}
\section{Edge-based Local Push} 
\label{sec:algorithm}
\header{\bf An overview.}
{\crc At a high level, \edgepush decomposes the atomic \lpush operation into edge level granularity, which enables to flexibly distribute probabilities according to the edge weights. To demonstrate the superiority of \epush, consider the toy graph shown in Figure~\ref{fig:weighted_graph}. We note that the atomic \lpush operation hinders \localpush to flexibly arrange the push order in the finer-grained edge granularity, leading to the $O(n)$ time cost even after the first \lpush step (i.e. from node $u$ to $v_1, v_2, \ldots, v_n$). 
However, due to the unbalanced edge weights, the probability mass distributed from $u$ to $v_2, \ldots, v_n$ are ``insignificant" compared to the probability distributed to $v_1$. Consequently, in this toy example, the optimal push strategy is to directly distribute the probability mass at node $u$ along edge $\la u, v_1 \ra$, then along edge $\la  v_1,w \ra$, and ignore the other edges, which is allowed in our \edgepush framework.}  

\vspace{-1mm}
\subsection{A primitive operation: \epush}
\header{\bf Notations. } 
Before presenting the algorithm structure of \edgepush, we first define three variables which are conceptually maintained in the \edgepush process. 
\begin{itemize}[leftmargin = *]
	\item A {\em node income} $\q(v)$ for each node $v\in V$: it records the total  probability mass received by $v$ so far. Therefore, $\alpha \q(v)$ is indeed the reserve $\epi(v)$ in the context of \localpush. 
	\item An {\em edge expense} $\Q_{uv}$ for each edge $\la u,v \ra \in \bar{E}$: this variable records the total probability mass that has been transferred from $u$ to $v$ via the edge $\la u,v \ra$. By definition, the node income of node $v$ is the sum of all the expenses of edges $\forall \la u,v \ra\in \bar{E}$, i.e., $\q(v) = \sum_{u\in N(v)} \Q_{uv}$.
	\item An {\em edge residue} $\R_{uv}$ for each edge $\la u,v\ra \in \bar{E}$: it records the probability mass that is to be transferred from $u$ to $v$ at the moment. 
	Thus, by definition, we have:
    \begin{small}
    \begin{align}
    \label{eqn:edge-relation}
    \vspace{-3mm}
    \R_{uv}=(1-\alpha) \q(u) \cdot \frac{\A_{uv}}{d(u)}-\Q_{uv},
    \vspace{-3mm}
    \end{align}
    \end{small}
    where $(1-\alpha)\q(u) \cdot \frac{\A_{uv}}{d(u)}$ indicates the total probability mass that should be transferred from $u$ to $v$ so far.
\end{itemize}

\header{\bf Correctness. } 
Similar to the invariant maintained by \localpush, we can prove an analogous invariant for \edgepush: 
\begin{lemma}[Invariant by EdgePush]\label{lem:further-inv}
For each node $t$ in the graph, the node income $\q(t)$ and edge residues satisfy the following invariant during the EdgePush process: 
\begin{align}\label{eqn:invariant_edgepush}
\vspace{-2mm}
    \vpi(t)=\alpha \q(t)+\sum_{\la u,v \ra \in \bar{E}}\R_{uv}\cdot \vpi_v(t)\,,
\vspace{-2mm}
\end{align}
where $\vpi_v(t)$ denotes the PPR value of node $t$ w.r.t node $v$ as the source node. When the source node is $s$, we use $\vpi(t)$ by default. 
\end{lemma}
A sketch for the invariant is shown in the right side of Figure~\ref{fig:invariant}. To see the correctness of this invariant, recall that in \localpush, the probability mass to be distributed from each node $u$ is recorded by the {\em node residue} $\r(u)$. In contrast, in our \edgepush, we further decompose the probability mass maintained by {\em node residue} into edge level granularity and use {\em edge residue} to record it. Hence, the PPR value of node $t$ is divided to two types of probabilities. $\alpha \q(t)$ maintains the probability mass that has been received by node $t$. Except that, $\sum_{\la u,v \ra \in \bar{E}} \R_{uv}\cdot \vpi_v(t)$ records the probability mass to be distributed to $t$. The formal proof of this invariant is deferred to the 
appendix for readability.


\header{\bf The {\em edge-based push} operation. }
In an \epush along edge $\la u,v_1\ra$ as shown in the right side of Figure~\ref{fig:comparison}, \edgepush transfers the probability mass at $\R_{u{v_1}}$ to node $v_1$ (i.e. to the edge residues of all edges $\la v_1, w \ra$), and reset $\R_{u{v_1}}$ to $0$. 

Notably, in the \edgepush process, we don't explicitly modify the edge residue. Instead, we update the node income and edge expense with the edge residue maintained implicitly according to Equation~\eqref{eqn:edge-relation}. For example, in the \epush along $\la u,v_1\ra$ as shown in Figure~\ref{fig:comparison}, we first calculate $\R_{uv_1}$ at current stage by $\R_{uv_1}=\frac{(1-\alpha)\A_{uv}\q(u)}{d(u)}-\Q_{uv_1}$. Then we increase $\q(v_1)\gets \q(v_1)+\R_{uv_1}$ and $\Q_{uv_1}\gets \Q_{uv_1}+\R_{uv_1}$. 
By this means, $\R_{uv_1}$ is implicitly set to $0$ and for $\forall \la v_1,w\ra$, $\R_{v_1w}$ is increased simultaneously because of the increment of $q(v_1)$. 

\vspace{-1mm}
\subsection{The \edgepush Algorithm}
In Algorithm~\ref{alg:edge-search}, we present the pseudo code of the \edgepush algorithm. 
More precisely, given a weighted and undirected graph $G=(V,E)$, a source node $s$, a constant teleport probability $\alpha\in (0,1)$ and the termination threshold $\theta(u,v)$ for $\forall \la u,v\ra \in \bar{E}$, \edgepush initializes the node income vector $\q$ as $\bm{e}_s$ and the edge expense matrix $\Q$ as $\mathbf{0}_{n \times n}$. During the \edgepush process,  the \edgepush algorithm {\em conceptually} maintains a set $\mathcal{C}$ for the candidate edges, which is defined as $\mathcal{C}=\left\{\la u,v \ra\in \bar{E} \mid \R_{uv}  \ge  \theta(u,v) \right\}$. 
\edgepush repeatedly picks edges from the candidate set $\mathcal{C}$ to perform \epush operations until $\mathcal{C} = \emptyset$. Specifically, \edgepush repeats the following process until the termination. 
\begin{itemize}[leftmargin = *]
	\item Pick an arbitrary edge $\la u,v \ra \in \mathcal{C}$ and let $y \leftarrow \R_{uv}$;
	\item Perform an \epush operation on $\la u,v \ra$ by ``pushing'' the edge residue $\R_{uv}$ along the edge $\la u,v \ra$ from $u$ to $v$; as a result, both the node income of $v$ and the edge expense of $\la u,v \ra$ are increased by an amount of $y$, i.e., $\q(v) \leftarrow \q(v) + y$ and $\Q_{uv} \leftarrow \Q_{uv} + y$, where $y=\R_{uv}=\frac{(1-\alpha)\A_{uv}\q(u)}{d(u)}-\Q_{uv}$. 
	\item Conceptually maintain the set $\mathcal{C}$ according to the increases of $\q(v)$ and $\Q_{uv}$.
\end{itemize}
When the above process terminates, \edgepush returns $\epi = \alpha \q$ as the estimator of the SSPPR vector $\vpi$. 

\begin{algorithm}[t]
\begin{small}
   \caption{The \edgepush Algorithm}
   \label{alg:edge-search}
   \KwIn{Graph $G\hspace{-0.5mm}=\hspace{-0.5mm}(V,E)$, source node $s\hspace{-0.5mm}\in \hspace{-0.5mm} V$, teleport probability $\alpha$, termination threshold $\theta(u,v)$ for $\forall \la u,v\ra\hspace{-0.5mm}\in \hspace{-0.5mm} \bar{E}$\\}
\KwOut{$\hat{\vpi}$ as the estimation of SSPPR vector $\bm{\pi}$\\}
$\q\gets \bm{e}_s$, $\Q\gets\mathbf{0}_{n \times n}$\;
$\mathcal{C}=\left\{\la u,v \ra \in \bar{E} \mid (1-\alpha)\q(u) \cdot\frac{\A_{uv}}{d(u)}-\Q_{uv}\ge  \theta(u,v) \right\}$\;
\While{$\mathcal{C}$ is not empty}{
	pick an edge $\la u,v \ra \in \mathcal{C}$\;
	$y\gets (1-\alpha)\q(u) \cdot \frac{\A_{uv}}{d(u)}-\Q_{uv}$\;
	$\Q_{uv}\gets \Q_{uv}+y$\;
	$\q(v) \gets \q(v)+y$\;
	Update the set $\mathcal{C}$\;
}
\Return {$\hat{\vpi} = \alpha \q$ as the estimation of $\vpi$}\;
\end{small}
\end{algorithm}


\vspace{-1mm}
\subsection{Maintaining Candidate Set $\mathcal{C}$}\label{subsec:candidate}
\vspace{-1mm}
It remains to show 
how to maintain the candidate set $\mathcal{C}$ efficiently in Algorithm~\ref{alg:edge-search}. Clearly, the cost of maintaining $\mathcal{C}$ explicitly can be expensive. To see this, consider the moment right after an \epush operation on edge $\la u,v \ra$ is performed. First, due to the increment of $\q(v)$, according to Equation~\eqref{eqn:edge-relation}, 
the edge residues of all $v$'s adjacency edges $\la v,w \ra \in \bar{E}$ are increased, and thus, can be possibly inserted to $\mathcal{C}$. This cost can be as large as $\Omega(n(v))$. Second, due to the increment of $\Q_{uv}$, the edge residue $\R_{uv}$ is equivalently reduced to $0$. 
As a result, $\la u,v \ra$ will be removed from $\mathcal{C}$. Therefore, maintaining $\mathcal{C}$ explicitly would lead to a $\Omega(n(v) + 1)$ cost, which is hardly considered as more efficient than a \lpush operation in \localpush. 

Fortunately, a key  observation is  that the \edgepush algorithm just needs to pick an {\em arbitrary} edge in $\mathcal{C}$ to perform an \epush operation.
There is actually no need to maintain $\mathcal{C}$ explicitly. Based on this observation, we propose a two-level structure to efficiently find {``eligible''} edges in $\mathcal{C}$ to push. As we shall show shortly, this two-level structure dramatically brings down the aforementioned $\Omega(n(v) + 1)$ cost to $O(1)$ amortized. 

\header{\bf A two-level structure.}
Recall that in the \edgepush algorithm, the candidate set $\mathcal{C}$ is defined as $\mathcal{C}=\{\la u,v \ra \in \bar{E} \mid \R_{uv}\ge \theta(u,v)\}$. According to Equation~\eqref{eqn:edge-relation} that $\R_{uv}=(1-\alpha)\q(u) \cdot \frac{\A_{uv}}{d(u)}-\Q_{uv} $, we can rewrite the definition of $\mathcal{C}$ as: 
\begin{align}\label{eqn:alterC}
\vspace{-2mm}
    \mathcal{C}=\left\{\la u,v \ra \in \bar{E} \mid \frac{(1-\alpha)\q(u) }{d(u)} \ge \frac{1}{\A_{uv}}\cdot \left(\Q_{uv}+\theta(u,v)\right)\right\}.  
\vspace{-2mm}
\end{align}
We note that in the above inequality, the right side focuses on the probability mass distributed to node $u$'s adjacency edges, at a local level. The left side of the inequality maintains the probability mass of each node $u$ in the graph, that is, at a global level. 
Thus, we propose a two-level structure to update the set $\mathcal{C}$ at the local and global level separately:  
\vspace{-1mm}
\begin{itemize}[leftmargin = *]
	\item At the local level, 
	for each node $u \in V$, we maintain a {\em priority queue} of all the neighbors of $u$, denoted by $\mathcal{Q}(u)$, where the priority of each neighbor $v$ is defined as:
	\begin{equation}\label{eqn:pqi}
	\vspace{-2mm}
	k_u(v) = \frac{1}{\A_{uv}} \cdot \left( \Q_{uv}+\theta(u,v)\right)\,.
	\end{equation}
	\vspace{-2mm}
	\item At the global level, 
	for $\forall u\in V$, we define the key of node $u$ as:
\begin{align}\label{eqn:pq}
\vspace{-2mm}
	K_u = -\frac{(1-\alpha)\q(u)}{d(u)} + \mathcal{Q}(u).\text{top}, 
\vspace{-1mm}
\end{align}
where $\mathcal{Q}(u).\text{top}$ is the smallest priority value in $\mathcal{Q}(u)$. We maintain a {\em linked list} $\mathcal{L}$ for storing all the nodes $u$ such that $K_u \leq 0$. 
\end{itemize}

According to the definition of the two-level structure, we have a crucial observation given below: 
\begin{observation}
	\label{obs:correctness}
 	Let $v$ be the neighbor of $u$ with the smallest priority in $\mathcal{Q}(u)$. Then the edge $\la u,v\ra \in \mathcal{C}$ {\em if and only if} $K_u \leq 0$.  
\vspace{-1mm}
\end{observation} 

\begin{proof}
Since $v$ is the neighbor of $u$ with the smallest priority in $\mathcal{Q}(u)$, we have 
\begin{align}\nonumber
\vspace{-2mm}
K_u =-\frac{(1-\alpha)\q(u)}{d(u)}+\mathcal{Q}(u).\text{top}=-\frac{(1-\alpha)\q(u)}{d(u)}+\frac{\left(\Q_{uv}+\theta(u,v)\right)}{\A_{uv}}. 
\vspace{-1mm}
\end{align}
As a result, $K_u\le 0$ if and only if $\frac{(1-\alpha)\q(u)}{d(u)}-\frac{1}{\A_{uv}}\cdot \left(\Q_{uv}+\theta(u,v)\right)\ge 0$, which is concurs with the definition of $\mathcal{C}$ given in Equation~\eqref{eqn:alterC}. Thus, the observation follows. 
\end{proof}

Based on the two-level structure and Observation~\ref{obs:correctness}, we can pick edges from the candidate set $\mathcal{C}$ without updating the edge residues of all adjacency edges.

\header\underline{\em Conceptually Pick an Edge from $\mathcal{C}$.}
To pick an edge from $\mathcal{C}$, by Observation~\ref{obs:correctness}, it suffices to first pick an arbitrary node $u$ from the linked list $\mathcal{L}$, and then, take the edge $\la u,v \ra$ with $v$ having the smallest priority in $\mathcal{Q}(u)$. 
This can be implemented easily by taking a node from a linked list and by invoking the {\em find-min} operation of the priority queue. 

\header\underline{\em Conceptually Maintain $\mathcal{C}$.}
Let $\la u,v \ra$ be the edge picked. 
After the edge-based push operation on $\la u,v \ra$, 
our \edgepush performs the following steps:
(i) invoke an {\em increase-key} operation for $v$ in $\mathcal{Q}(u)$ (we discuss how to implement this with allowed priority queue operations shortly);
(ii) check the key $K_u$: if $K_u > 0$, remove $u$ from the linked list $\mathcal{L}$; and 
(iii) check the key $K_v$: if $K_v \leq 0$, add $v$ to $\mathcal{L}$.

\header{\bf Correctness of the two-level structure.}
The correctness of the two-level structure for maintaining the candidate set $\mathcal{C}$ can be proved based on two facts. First, we observe that the linked list $\mathcal{L}$ always keeps all the nodes $u$ with $K_u \leq 0$. Second, the edge expense $\Q_{uv}$ can be increased only and this happens only when a push is performed on the edge $\la u, v\ra$. According to Equation~\eqref{eqn:pqi}, if an edge of $u$ is eligible for a push, it will eventually appear at the top of $\mathcal{Q}(u)$, and then be captured by Observation~\ref{obs:correctness}.

\header{\bf Cost per \epush. }
Recall that we assume the word RAM model
where each edge weight value $\A_{uv}$ and each priority value $k_u(v)$ can be represented by $O(\log n)$ bits. In this model, we can sort 
all the edge weights $\A_{uv}$ and all the priorities $k_u(v)$ in $O(m)$ time with the standard Radix sort~\cite{knuth1998art}. The key idea is to perform Counting Sort on every $\log n$ bits rather than on every bit, leading to the number of passes as $O(1)$. 
More precisely, during the sorting process, the Radix Sort can first perform Counting Sort on the least significant $\log n$ bits, then by the next $\log n$ bits and so on so forth, up to the most significant $\log n$ bits. As a result, the Radix Sort requires only $O\left( \frac{\log{\text{poly}(n)}}{\log n}\right)=O(1)$ passes of Counting Sort. By ~\cite{cormen2009introduction}, the time cost of Counting Sort can be bounded by $O(m+n)=O(m)$ to sort $O(m)$ elements in the range from $0$ to $2^{\log n}-1=n-1$. Thus, the overall running time of Radix Sort is bounded by $O(m)$. 
Furthermore, we have the following fact:

\vspace{-1mm}
\begin{fact}[Theorem 1 in~\cite{thorup1995equivalence}]
	\label{fact:heap-cost}
	If all the $m$ priorities $k_u(v)$ can be sorted in $O(m)$ time, there exists a priority queue $\mathcal{Q}$ with capacity of $m$ which supports each:
	(i) {\em find-min} operation in $O(1)$ worst-case time, (ii) {\em delete} operation (removing an element from $\mathcal{Q}$) in $O(1)$ amortized time, and 
	(iii) {\em restricted insert} operation (inserting an element with priority $> \mathcal{Q}.\text{top}$) in $O(1)$ amortized time. 
\end{fact}

In the above implementation, each edge-based push operation only involves: one {\em find-min} and one {\em increase-key} in the priority queues, and $O(1)$ standard operations in the linked list $\mathcal{L}$.
As the {\em increase-key} operation can be implemented by a {\em restricted insert} followed by a {\em delete} operation in $\mathcal{Q}(u)$, the cost of {\em increase-key} is bounded by $O(1)$ amortized. Thus, we can derive the following theorem: 

\begin{theorem}\label{thm:cost-per-push}
\vspace{-2mm}
	The cost of each edge-based push operation is bounded by $O(1)$ amortized.	
\vspace{-1mm}
\end{theorem}

\header{\bf Pre-processing.}
As the edge weights can be sorted in $O(m)$ time and by Fact~\ref{fact:heap-cost},
we can pre-process the input graph $G$ in $O(m)$ time, such that for each node $u \in V$:
(i) all the out-going edges of $u$ are sorted by their weights $\A_{uv}$, and 
(ii) the priority queue $\mathcal{Q}(u)$ is constructed.
Furthermore, we may also store certain aggregated information in memory such as $\|\A\|_1$ and $\sum_{\la u,v \ra \in \bar{E}} \sqrt{\A_{uv}}$ (which we shall see shortly in the analyses).
\section{Theoretical Analysis} 
\label{sec:analysis}

{\hz In this section, we analyze the theoretical error and time complexity of \edgepush. Additionally, we provide a novel notion, $\cos^2 \p$, to characterize the unbalancedness of weighted graphs and to assist in evaluating the theoretical advantage of \edgepush over \localpush.} {\hz For readability, we defer all proofs in this section to the appendix. }


\vspace{-1mm}
\subsection{Analysis for the \edgepush Algorithm}
\header
{\bf Overall time complexity.}
To bound the overall time cost of \edgepush, recall that  Theorem~\ref{thm:cost-per-push} states each edge push operation takes amortized constant time. Consequently, it suffices to bound the  total number of edge  push operations as the overall time complexity. 
\begin{lemma}[Time cost of EdgePush]\label{lem:EdgePushCost}
	The overall running time of EdgePush is bounded by $O\left(\sum_{\la u,v \ra\in \bar{E}}\frac{(1-\alpha)\vpi(u) \A_{uv}}{\alpha \cdot d(u) \cdot \theta(u,v)}\right)$. In particular, when the source node is randomly chosen according to the degree distribution, the expected overall running time of EdgePush is bounded by $O\left(\sum_{\la u,v \ra\in \bar{E}}\frac{(1-\alpha)\A_{uv}}{\alpha \cdot \|\A\|_1 \cdot \theta(u,v)}\right)$. 
\end{lemma}

\header{\bf Error analysis.}
Recap the Invariant~\eqref{eqn:invariant_edgepush} shown in Section~\ref{sec:algorithm}. By using $\alpha \q(t)$ as an approximate PPR value of $\vpi(t)$, we have two straightforward observations: 
(i) $\alpha \q(t)$ is an underestimate because all edge residuals are non-negative, 
and (ii) the additive error $\vpi(t)-\alpha \q(t)$ is bounded by $\sum_{\la u,v \ra \in \bar{E}} \R_{uv}\cdot \vpi_v(t)$, 
the edge residuals to be distributed to $t$. 
Summing over all possible target node $t$, we have the following lemma about the bound on the $\ell_1$-error of \edgepush.

\begin{lemma}[$\ell_1$-error]\label{lem:EdgePushErr} 
The EdgePush method shown in Algorithm~\ref{alg:edge-search} returns an approximate SSPPR vector within an $\ell_1$-error $\sum_{\la u,v \ra\in \bar{E}}\theta(u,v)$. 
\end{lemma}

Moreover, in fact, the proof of Lemma~\ref{lem:EdgePushErr} also derives the additive error bound for \edgepush.  
We have the following Lemma. 
\begin{lemma}[Normalized Additive Error]\label{lem:EdgePushErr-ad} 
For each node $t\in V$ in the graph, the EdgePush method answers the SSPPR queries within a normalized additive error 
$\frac{1}{d(t)}\cdot \sum_{\la u,v \ra\in \bar{E}}\theta(u,v) \cdot \vpi_v(t)$ for each node $t\in V$. 
\end{lemma}



\subsection{Settings of the Termination Threshold}
As shown in Algorithm~\ref{alg:edge-search}, our \edgepush admits an individual termination threshold $\theta(u,v)$ for each edge $\la u,v \ra\in \bar{E}$. As we shall show in the following, by setting $\theta(u,v)$ carefully, the  \edgepush achieves superior query efficiency over \localpush. 

Before illustrating the setting of $\theta(u,v)$, we first present an important inequality: Cauchy-Schwarz Inequality~\cite{steele2004cauchy}, which serves as the basis of the following analysis. 
\begin{fact}[Cauchy-Schwarz Inequality~\cite{steele2004cauchy}]\label{fact:cauchy}
	Given two vectors $\z=\{\z(1),\z(2),...,\z(m)\}\in \mathbb{R}^{m}$, $\x=\{\x(1),\x(2),...,\x(m)\}\in \mathbb{R}^{m}$, the Cauchy-Schwarz Inequality states that:
	\begin{align}\label{eqn:cauchy-entry}
		\left(\sum_{i=1}^m \z(i)\cdot \x(i)\right)^2 \le \left(\sum_{i=1}^m \z^2(i)\right)\left( \sum_{i = 1}^m \x^2(i)\right), 
	\end{align} 
	where the equality holds when $\frac{\z(1)}{\x(1)}=\frac{\z(2)}{\x(2)}=...=\frac{\z(m)}{\x(m)}$. 
\end{fact}

\header{\bf \hz The choice of $\boldsymbol{\theta(u,v)}$. } 
{\hz
We first present the optimal choice of $\theta(u,v)$ for \edgepush with $\ell_1$-error $\e$. Let us recap two results.  
Firstly, Lemma~\ref{lem:EdgePushErr} shows that the $\ell_1$-error is bounded by $\sum_{\la u,v\ra \in \bar{E}}\theta(u,v)$. On the other hand, by Lemma~\ref{lem:EdgePushCost}, when the source node is chosen according to the degree distribution, the overall expected running time is bounded by $O\left( \sum_{\la u,v\ra \in \bar{E}}\frac{(1-\alpha)\A_{uv}}{\alpha\cdot \|\A\|_1 \cdot \theta(u,v)} \right)$. Clearly, there is a trade-off between the error and the running time cost via the values of $\theta(u,v)$ for all $\la u, v\ra \in \bar{E}$. 
As a result, it suffices to aim at a setting of all $\theta(u,v)$'s such that: (i) the overall $\ell_1$-error $\sum_{\la u, v\ra \in \bar{E}} \theta(u,v)\hspace{-0.5mm}=\hspace{-0.5mm} \e$, and (ii) the quantity $Cost \hspace{-1mm} \triangleq \hspace{-1mm} \sum_{\la u,v\ra \in \bar{E}}\frac{(1-\alpha)\A_{uv}}{\alpha \|\A\|_1 \cdot \theta(u,v)}$ is minimized. Consequently, we prove the following theorem:
\begin{theorem}\label{thm:edge-efficiency-l1} 
	By setting $\theta(u,v)=\frac{\e \cdot \sqrt{\A_{uv}}}{\sum_{\la x,y \ra \in \bar{E}}\sqrt{\A_{xy}}}$ for each $\la u,v \ra\in \bar{E}$, the EdgePush algorithm returns an approximate SSPPR vector within $\ell_1$-error at most $\e$. In particular, when the source node is randomly chosen according to the degree distribution, the expected overall running time is bounded by $O\left(\frac{(1-\alpha)}{\alpha \e \|\A\|_1}\cdot \left(\sum_{\la u,v \ra\in \bar{E}}\sqrt{\A_{uv}}\right)^2 \right)$.
\end{theorem}

Likewise, we can derive the optimal choice of $\theta(u,v)$ for \edgepush with normalized additive error $r_{\max}$, illustrated in Theorem~\ref{thm:edge-efficiency-add}. 

}

\begin{theorem}\label{thm:edge-efficiency-add} 
\vspace{-1mm}
By setting $\theta(u,v)=\frac{r_{\max} \cdot d(v)\sqrt{\A_{uv}}}{\sum_{x\in N(v)}\sqrt{\A_{xv}}}$ for each $\la u,v \ra\in \bar{E}$, 
the EdgePush algorithm returns an approximate SSPPR vector within normalized additive error at most $r_{max}$. 
When the source node is randomly chosen according to the degree distribution, 
the expected overall running time is bounded by $O\left(\hspace{-1mm}\frac{(1-\alpha)}{\alpha r_{\max} \|\A\|_1}\hspace{-0.5mm}\cdot \hspace{-0.5mm} \sum_{v\in V}\hspace{-1mm}\frac{\left(\sum_{x\in N(v)}\hspace{-1mm}\sqrt{\A_{xv}}\right)^2 }{d(v)}\hspace{-0.5mm}\right)$.
\end{theorem}

\vspace{-3mm}
\subsection{Comparison to the \localpush Algorithm}\label{subsec:individual}
Next, we show the superiority of \edgepush over \localpush. 
To facilitate our analysis, we first define the following four {\em characteristic vectors} of a given undirected weighted graph:

\begin{definition}[characteristic vectors on weighted graphs]\label{def:vec}
\vspace{-1mm}
	Consider an undirected weighted graph $G=(V,E)$ with $n$ nodes, $m$ edges and $\bar{E}$ being the set of the bi-directional edges of every edge in $E$; clearly, $|\bar{E}| = 2 |E| = 2m$.
	Denote the (weighted) adjacency matrix by $\A$.
	We define four characteristic vectors $\z, \x, \z_v$ and $\x_v$ of $G$ as follows: 
\begin{itemize}
\vspace{-1mm}
	\item $\z \in \mathbb{R}^{2m}$: the vector whose the $i^\text{th}$ entry $\z(i) =\sqrt{\A_{uv}}$ 
		corresponds  to the $i^\text{th}$ edge $ \la u, v \ra \in \bar{E}$; 
	\item $\x \in \mathbb{R}^{2m}$: an all-one vector
		in the $2m$-dimensional space;
	\item $\z_v \in \mathbb{R}^{n_v}$ for each node $v \in V$: the vector whose the $j^\text{th}$ entry $\z_v(j)=\sqrt{\A_{uv}}$ corresponds to the $j^\text{th}$ neighbor node $u$ in $N_v$; 
	\item $\x_v \in \mathbb{R}^{n_v}$ for each node $v\in V$: an all-one vector in the $n_v$-dimensional space.
\end{itemize}
\end{definition}
Then the improvement of \edgepush over \localpush can be quantified by the the above characteristic vectors.
\begin{lemma}[Superiority of \edgepush with $\ell_1$-error] \label{lem:cos-l1} 
For the approximate SSPPR queries with $\ell_1$-error $\e$, 
we have
\begin{align}\label{eqn:impro-l1}
\vspace{-1mm}
	\frac{(1-\alpha)}{\alpha \e \|\A\|_1}\cdot \left(\sum_{\la u,v \ra \in \bar{E}} \sqrt{\A_{uv}}\right)^2= \left((1-\alpha) \cos^2 \p \right) \cdot \frac{2m}{\alpha \e}\,,
\end{align}
where $\p$ is the angle between the characteristic vectors $\z$ and $\x$. 
\end{lemma}
We note that the left hand side of Equation~\eqref{eqn:impro-l1} is the overall expected running time of \edgepush and $\frac{2m}{\alpha \e}$ is that of \localpush, both expressed by ignoring the Big-Oh notation.
To see the correctness of Lemma~\ref{lem:cos-l1}, 
observe that $2m \cdot \|\A\|_1 \cdot \cos^2 \p  =\left( \sum_{\la u,v \ra \in \bar{E}}1 \right)\cdot \left( \sum_{\la u,v \ra \in \bar{E}}\A_{uv} \right) \cdot \cos^2 \p = \|\x\|^2 \cdot \|\z\|^2 \cdot \cos^2 \p = \la \x, \z \ra^2 =  \left(\sum_{\la u,v \ra \in \bar{E}} \sqrt{\A_{uv}}\right)^2$. 
A more detailed proof can be found in the appendix. 
 
Likewise, the superiority of \edgepush for the SSPPR queries with normalized additive error can be quantified as follows.
\begin{lemma} [Superiority of \edgepush with normalized additive error] \label{lem:cos-add}
For the approximate SSPPR queries with specified normalized additive error $r_{\max}$, 
the expected overall running time of EdgePush is at most a portion $\frac{(1-\alpha)}{2m}\cdot \left(\sum_{v\in V}n_v \cdot \cos^2 \p_v\right)$ of the LocalPush's running time cost, 
where $\p_v$ is the angle between vectors $\z_v$ and $\x_v$. 
Specifically, we have
\vspace{-1mm}
\begin{align}\nonumber
\vspace{-4mm}
	\frac{(1\hspace{-0.5mm}-\hspace{-0.5mm}\alpha)}{\alpha r_{\max} \|\A\|_1}\hspace{-0.5mm}\cdot \hspace{-1.5mm} \sum\limits_{v\in V}\hspace{-1.5mm}\frac{\left(\hspace{-0.5mm}\sum\limits_{x\in N(v)}\hspace{-4mm}\sqrt{\A_{xv}}\hspace{-0.5mm}\right)^2 }{d(v)}\hspace{-0.5mm}=\hspace{-0.5mm} \frac{(1\hspace{-0.5mm}-\hspace{-0.5mm}\alpha)\hspace{-0.5mm}\cdot \hspace{-0.5mm}\left(\hspace{-0.5mm}\sum\limits_{v\in V}\hspace{-1.5mm}n(v) \hspace{-0.5mm}\cdot \hspace{-0.5mm}\cos^2 \hspace{-0.5mm}\p_v\hspace{-0.5mm}\right)}{2m}\hspace{-0.5mm} \cdot \hspace{-0.5mm}\frac{2m}{\alpha r_{\max}\|\A\|_1}. 
\vspace{-2mm}
\end{align}
\end{lemma}


{\header {\bf Superiority of \edgepush over \localpush}. } Based on Lemma~\ref{lem:cos-l1} and Lemma~\ref{lem:cos-add}, we can derive several interesting observations: 
\begin{itemize}
	\item First, $\cos^2 \p \le 1$ holds for all values of $\p$ (also applies for $\forall \p_v$). 
		This implies that the overall expected running time bound of \edgepush is never worse than 
		that of \localpush, regardless of the SSPPR queries with either $\ell_1$-error or normalized additive error.
	\item Second, 	
		when $\cos^2 \p = \Theta(1/n)$ (resp., $\cos^2 \p_v = \Theta(1/n)$ for $\forall v$), 
		\edgepush outperforms \localpush in terms of efficiency by a $\Theta(n)$ factor for answering SSPPR queries with $\ell_1$-error (resp., normalized additive error).
		This case could happen (but not necessarily) when all the nodes in a complete graph $G$ are $(a, b)$-unbalanced with $a = 1/n$ and $b = 1 - 1/n$. As an example, one can consider the case that each node in $G$ shares the same structure as node $u$ shown in Figure~\ref{fig:weighted_graph}.
	\item Third, when $\cos^2 \p = o(1)$, e.g. $\cos^2 \p = 1 / \log m$, \edgepush can achieve a {\em sub-linear} expected time   
		complexity $o(\frac{m}{\alpha \e})$ for solving the approximate SSPPR problem with specified $\ell_1$-error $\e$.
		This is impressive because \edgepush can answer the SSPPR query even without ``touching'' all the edges of $G$.
		It can be verified that the aforementioned complete graph example satisfies this condition with $\cos^2 \p = 1/ n$.
\end{itemize}

\begin{table}[t]
	\centering
	\caption{Real-World Datasets.}
	\vspace{-4mm}
	\begin{small}
		\begin{tabular}{|@{\hspace{+1mm}}l@{\hspace{+0.8mm}}|@{\hspace{+0.8mm}}r@{\hspace{+0.8mm}}|@{\hspace{+0.8mm}}r@{\hspace{+0.8mm}}|r|r|r@{\hspace{+1mm}}|} 
			\hline
			\multirow{2}{*}{{\bf Dataset}}& \multirow{2}{*}{{\bf $\boldsymbol{n}$}} & \multirow{2}{*}{{\bf $\boldsymbol{m}$}} & \multicolumn{2}{c|}{{\bf \hz Edge weight}} & \multirow{2}{*}{{\bf {\hz $\boldsymbol{{\cos}^2 \p}$}}}	\\ \cline{4-5}
			~ & & & {\bf mean} & {\bf max} & \\ \hline
			YouTube (YT) &  1,138,499 & 2,795,228 & 6.6 & 4,034 & 0.65 \\
			LiveJournal (LJ) & 4,847,571 & 71,062,058 &24 & 4,445 & 0.51\\
		    IndoChina (IC)& 7,414,768  &295,191,370	& 1,221 & 178,448 & 0.31\\
			Orkut-Links (OL) & 3,072,441 & 202,392,682 & 18 & 9,145 & 0.69\\	
			{\hz Tags (TA)} & {\hz 49,945} & {\hz 8,294,604} & {\hz 13} & {\hz 469,258} & {\hz 0.27}\\ 
			{\hz Threads (TH)} & {\hz 2,321,767} & {\hz 42,012,344} & {\hz 1.1} & {\hz 546} & {\hz 0.97}\\ 
			{\hz blockchair (BC)} & {\hz 595,753} & {\hz 1,773,544} &  {\hz 5.2} & {\hz 17,165} & {\hz 0.5}\\ 
			{\hz Spotify (SP)} & {\hz 3,604,308} & {\hz 3,854,964,026} & {\hz 8.6} & {\hz 2,878,970} & {\hz 0.29}\\ 
			\hline
			
		\end{tabular}
	\end{small}
	\label{tbl:datasets}
	\vspace{-3mm}
\end{table}

\begin{figure*}[t]
	\begin{minipage}[t]{1\textwidth}
		\centering
		\begin{tabular}{cccc}
			\hspace{-4mm} \includegraphics[width=43mm]{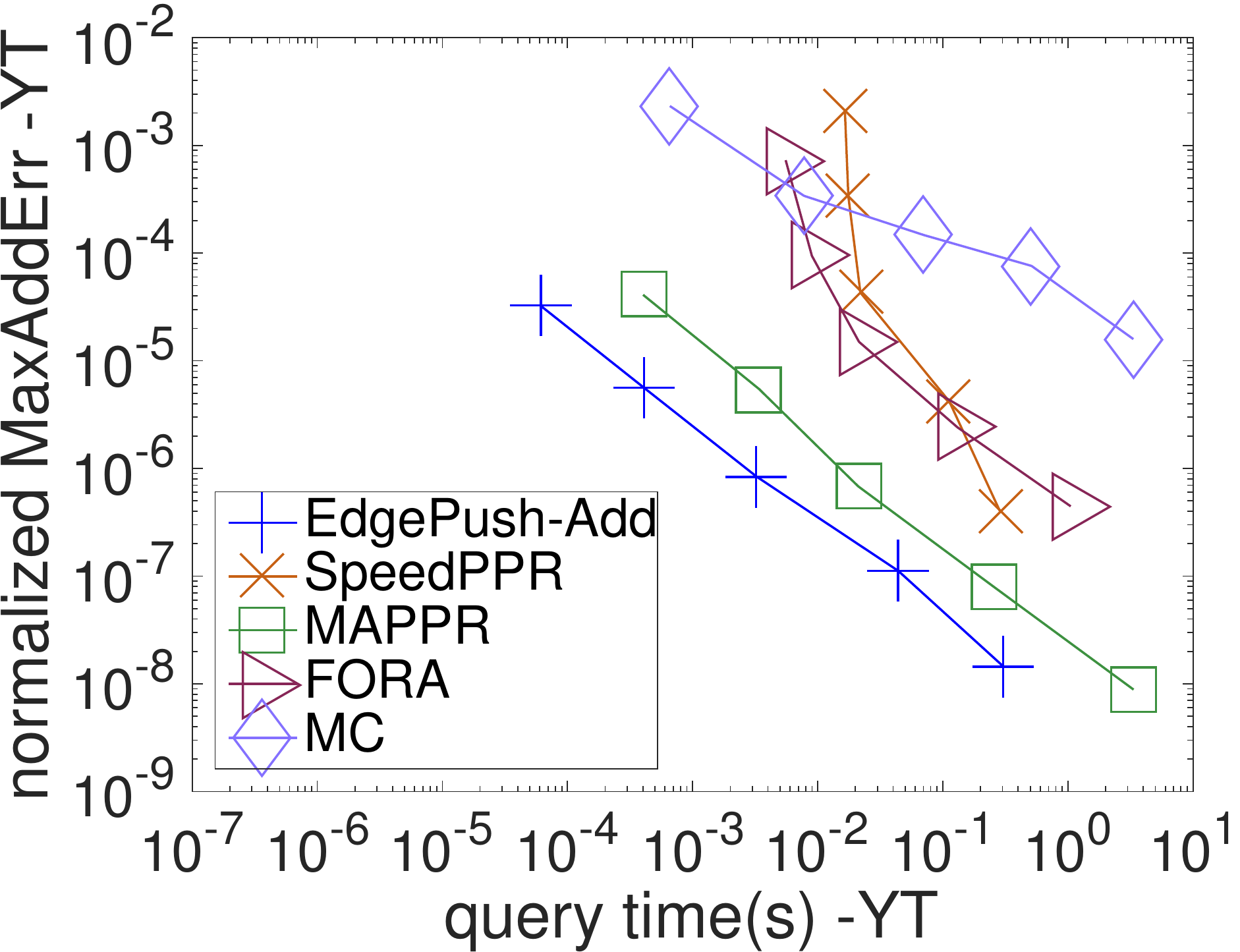} &
			\hspace{-3mm} \includegraphics[width=43mm]{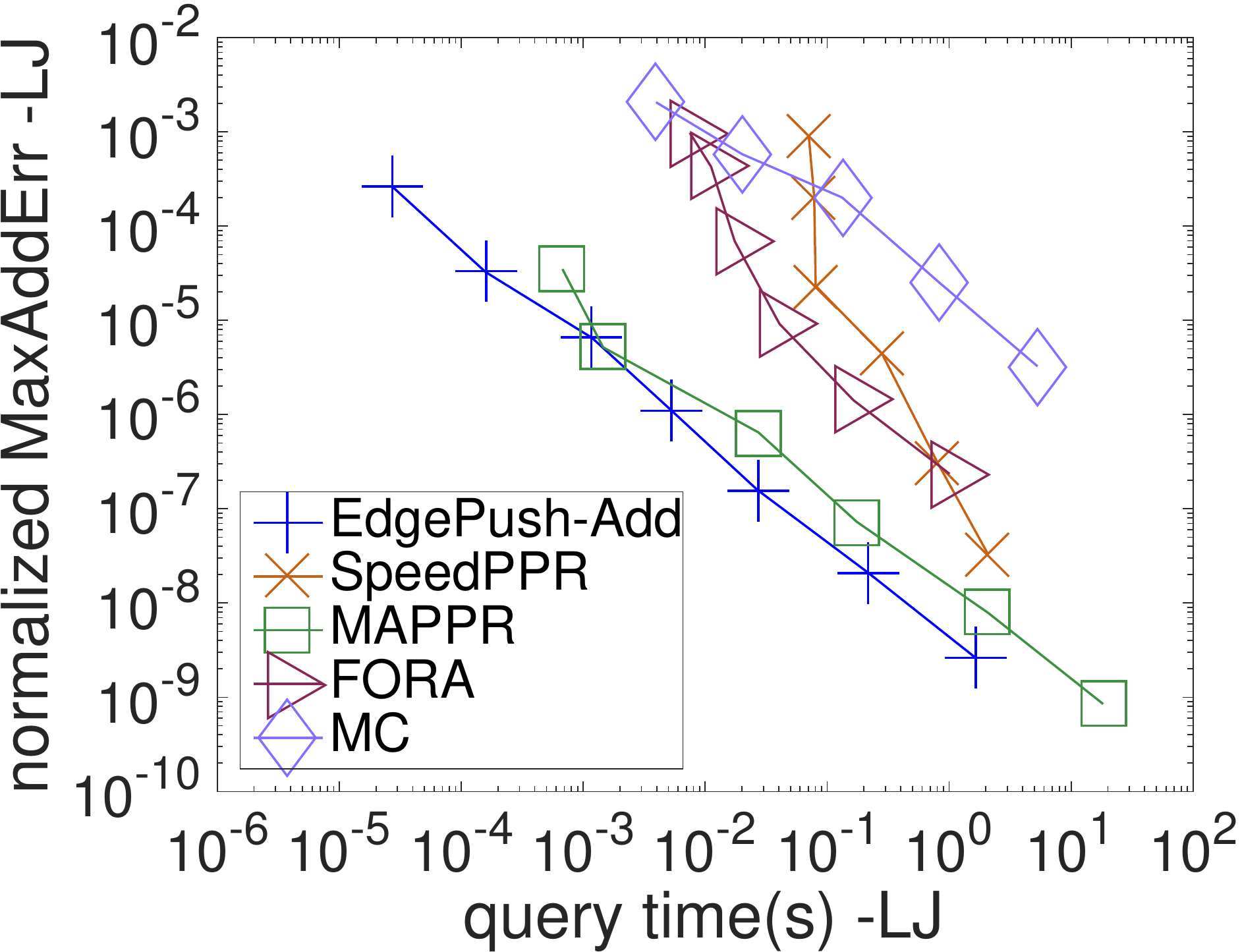} &
			\hspace{-3mm} \includegraphics[width=43mm]{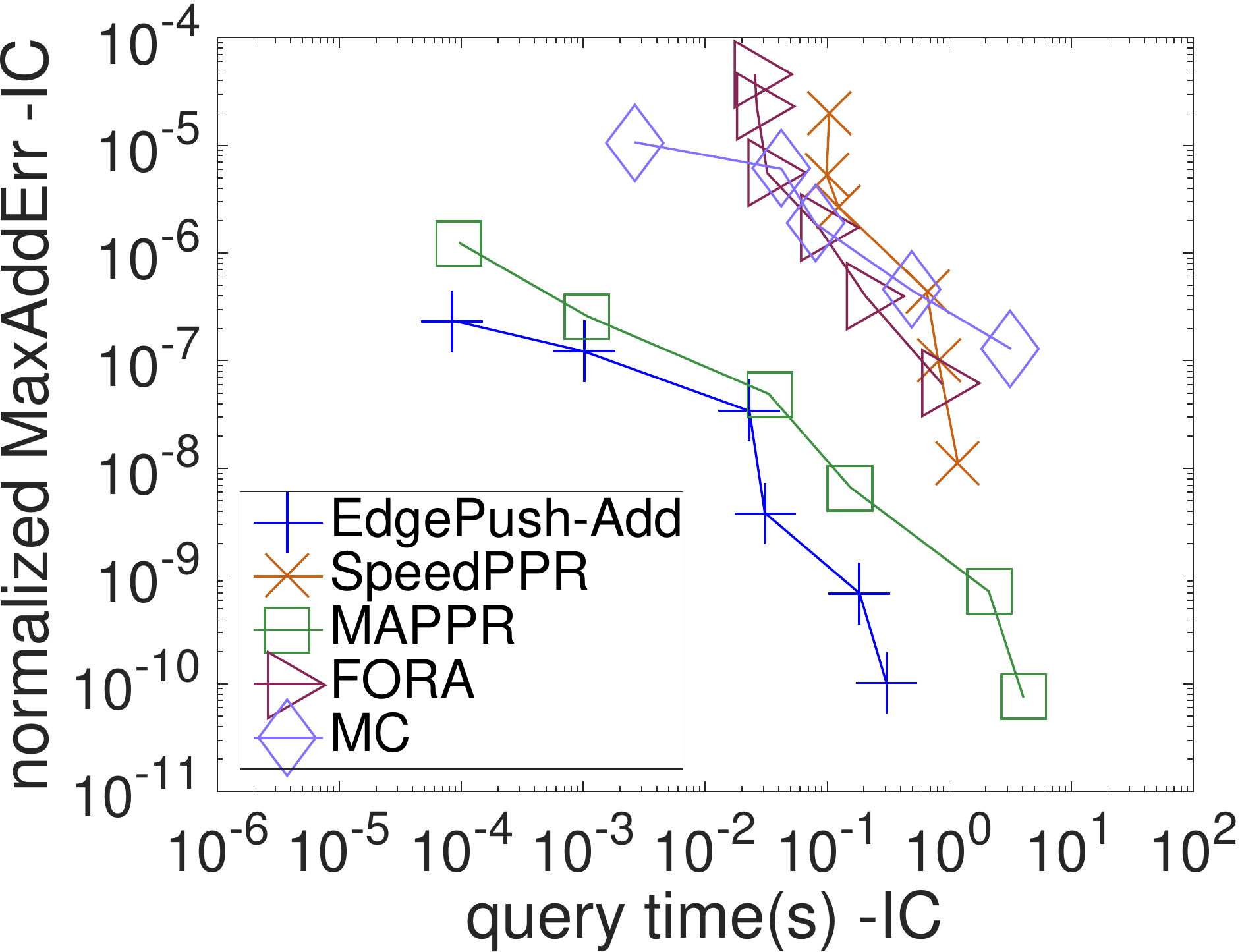} &
			\hspace{-3mm} \includegraphics[width=43mm]{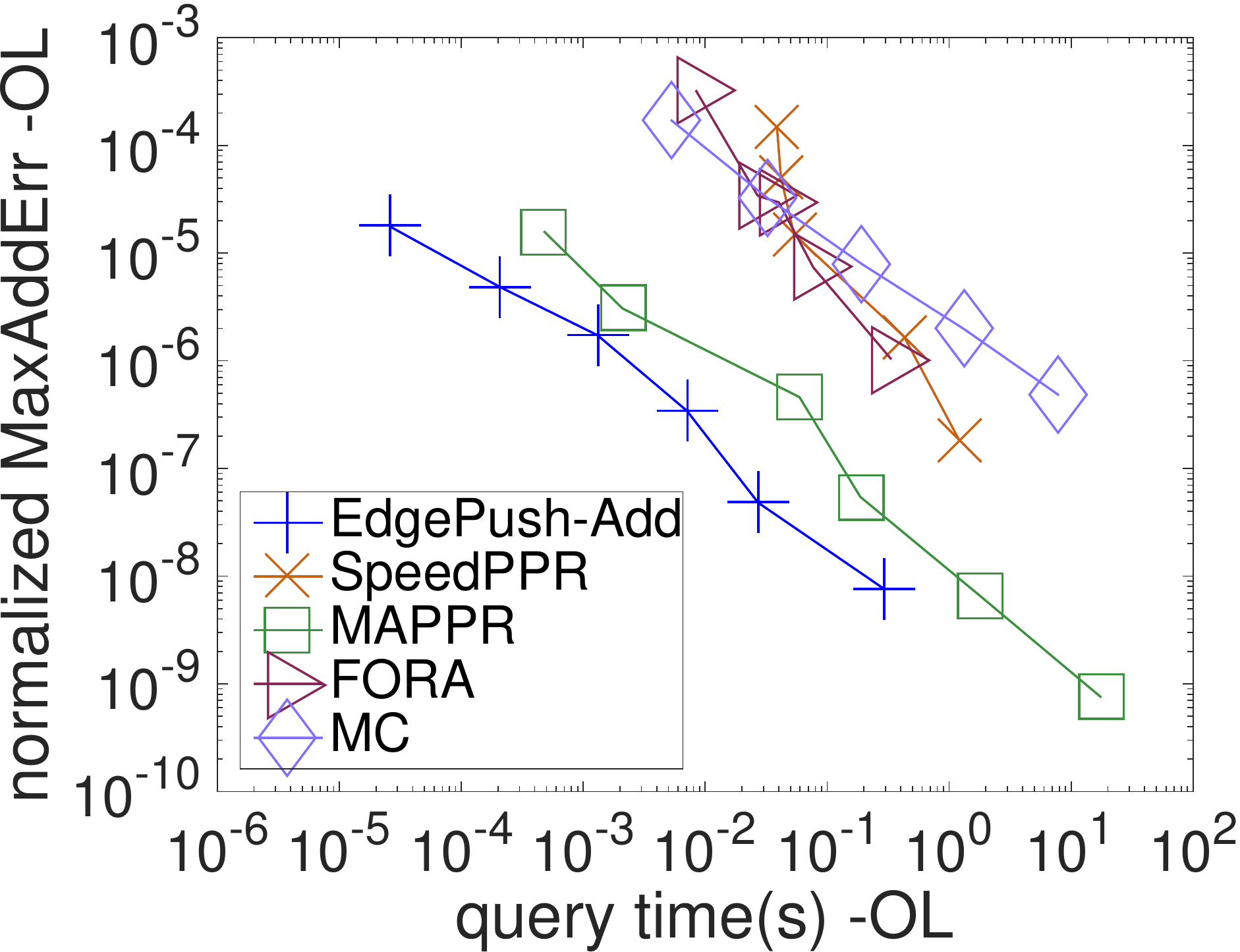} 
		\end{tabular}
		\vspace{-5mm}
		\caption{\hz {\em normalized MaxAddErr} v.s. query time on motif-based weighted graphs.}
		\label{fig:maxerror-query-real_pro}
		\vspace{-1mm}
    \end{minipage}


    \begin{minipage}[t]{1\textwidth}
		\centering
		\vspace{+0.5mm}
		\begin{tabular}{cccc}
			\hspace{-4mm} \includegraphics[width=43mm]{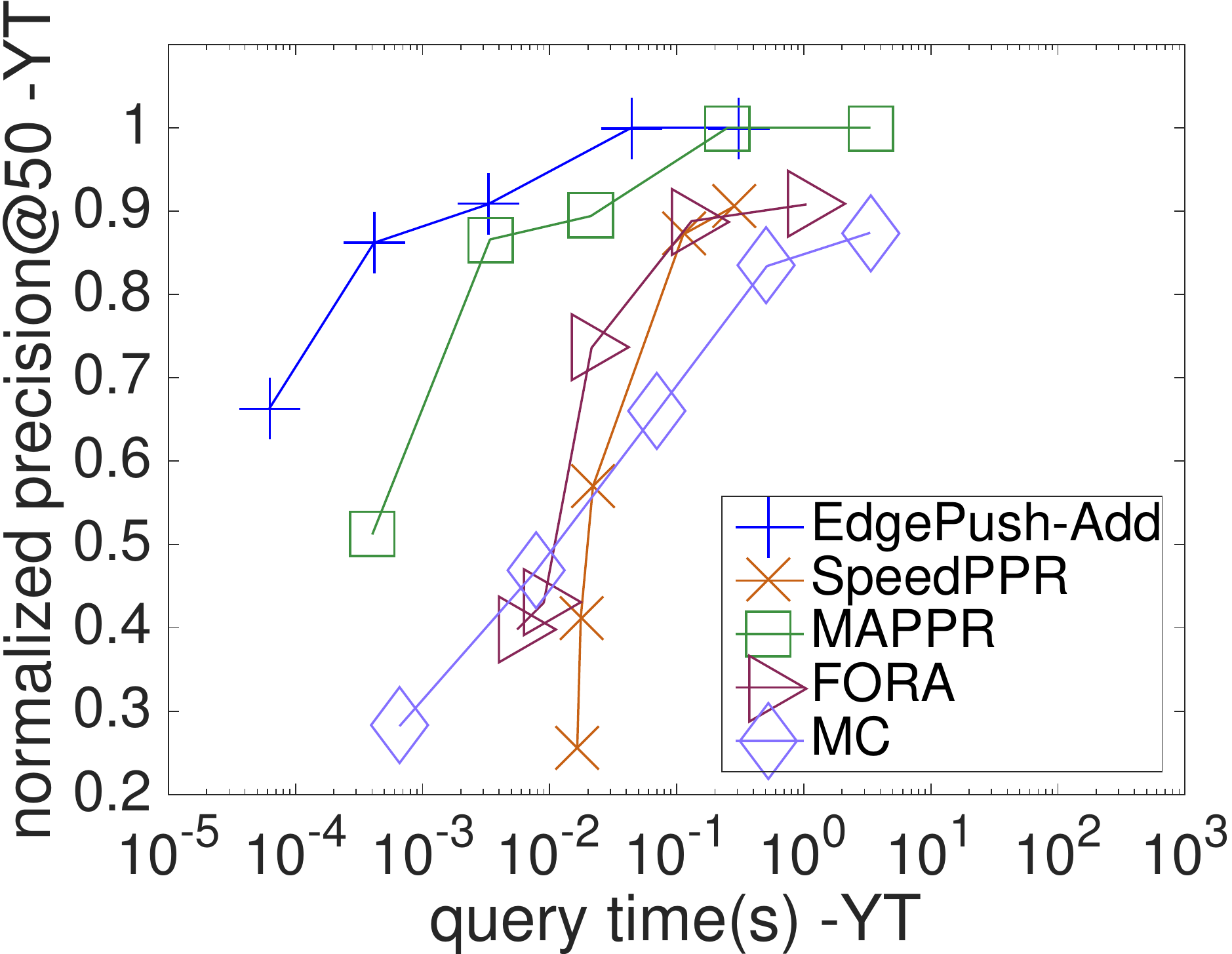} &
			\hspace{-3mm} \includegraphics[width=43mm]{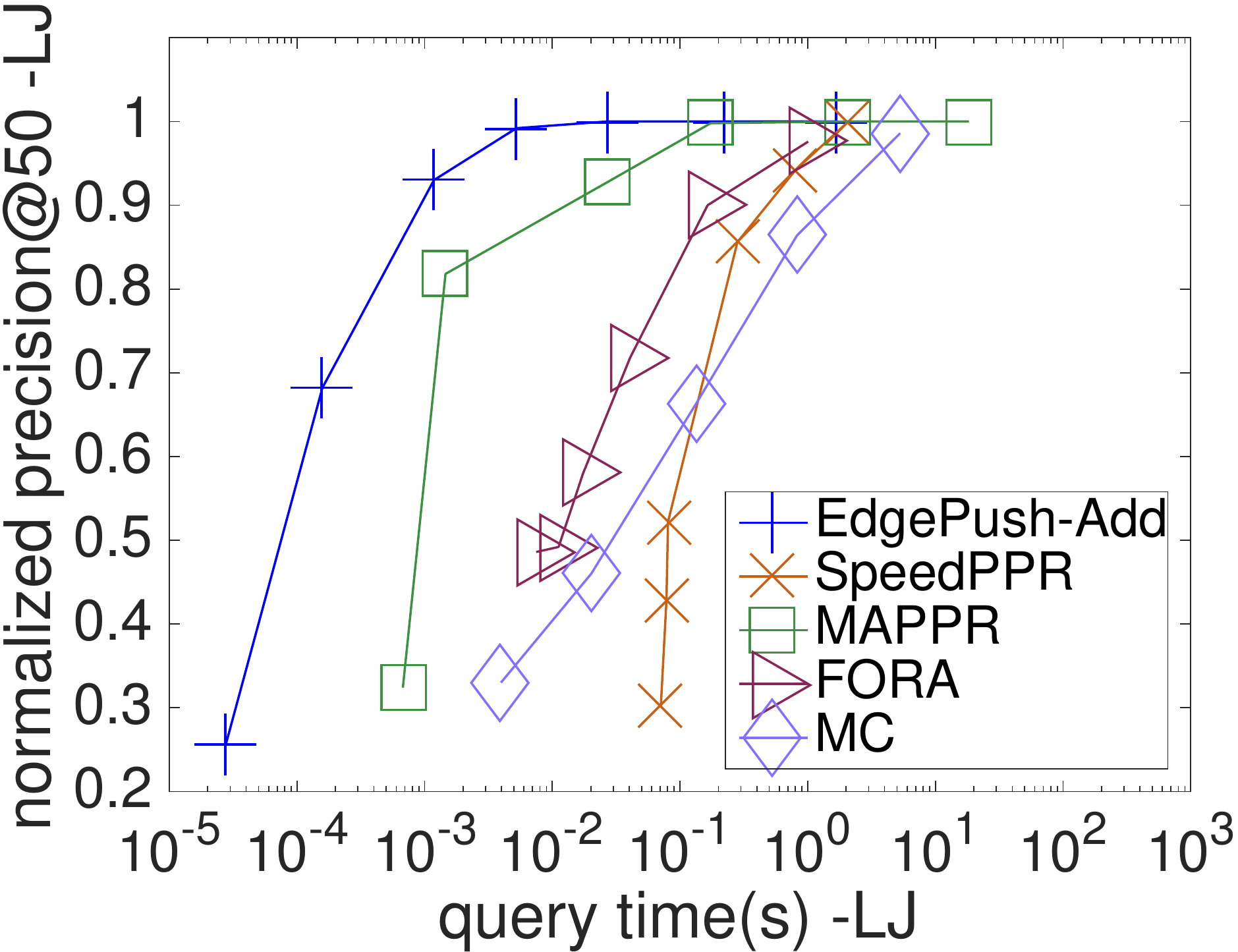} &
			\hspace{-3mm} \includegraphics[width=43mm]{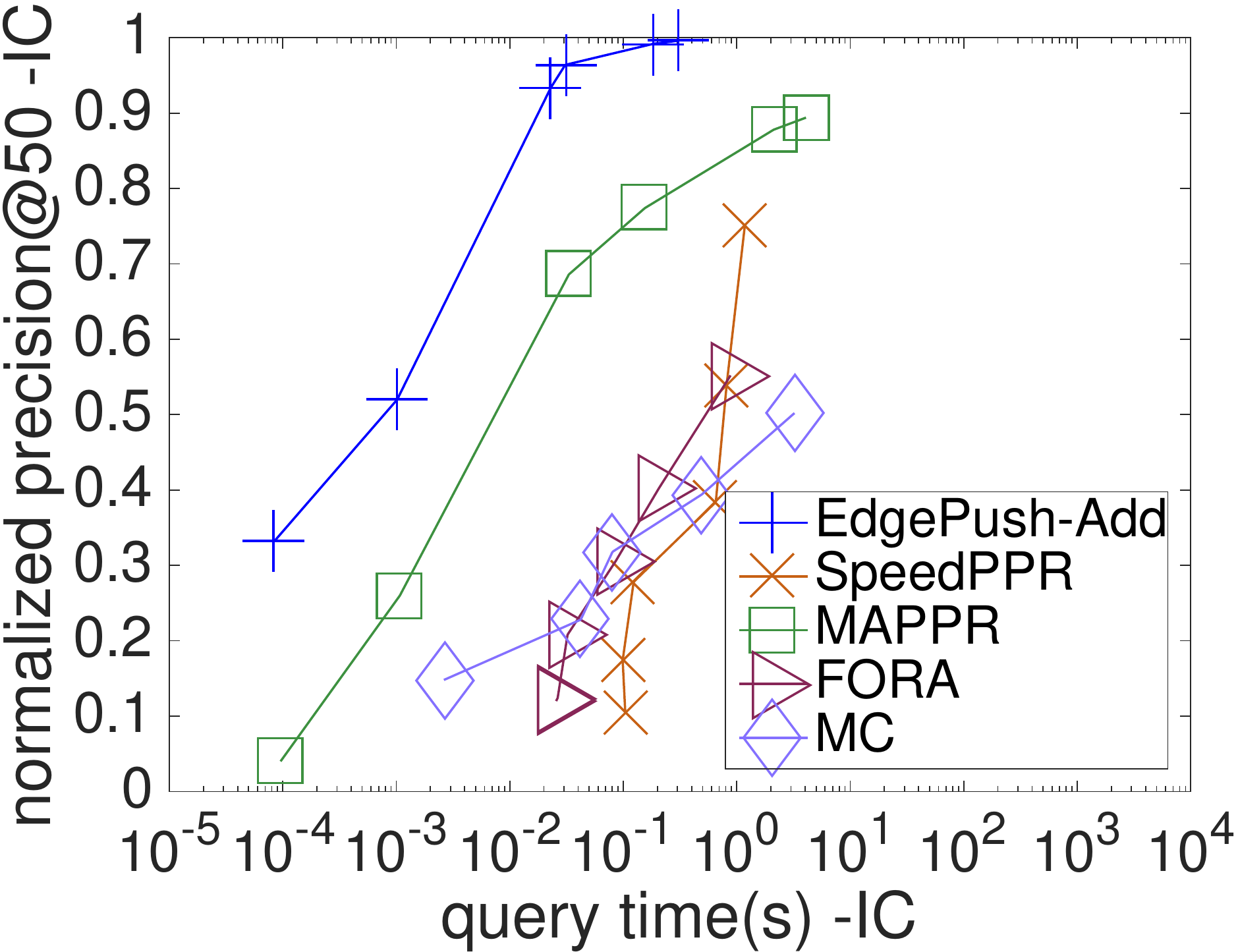} &
			\hspace{-3mm} \includegraphics[width=43mm]{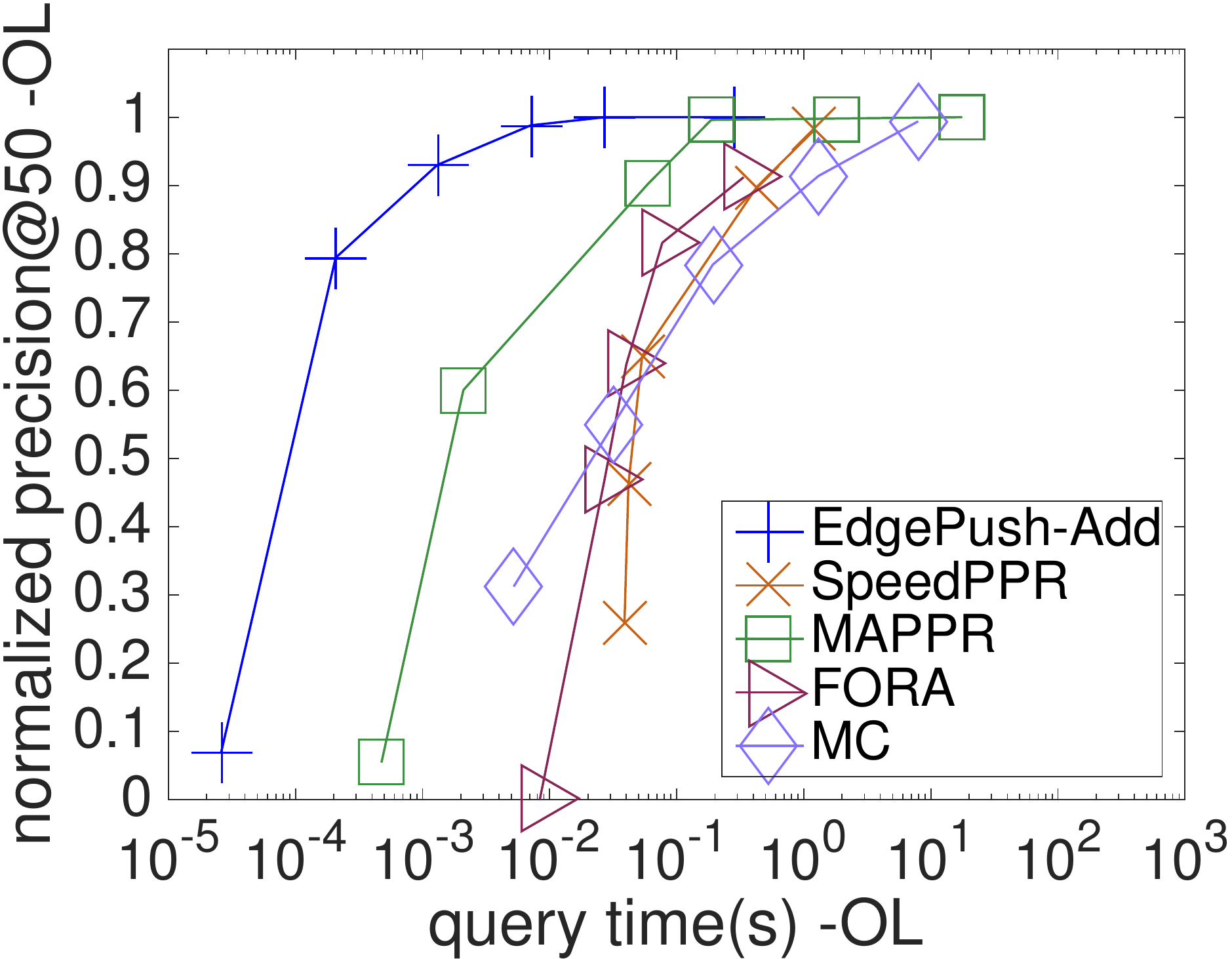} 
		\end{tabular}
		\vspace{-5mm}
		\caption{\hz {\em normalized precision@50} v.s. query time on motif-based weighted graphs.}
		\label{fig:precision-query-real_pro}
		\vspace{-1mm}
	\end{minipage}

	\begin{minipage}[t]{1\textwidth}
		\centering
		\vspace{+0.5mm}
		\begin{tabular}{cccc}
			\hspace{-4mm} \includegraphics[width=43mm]{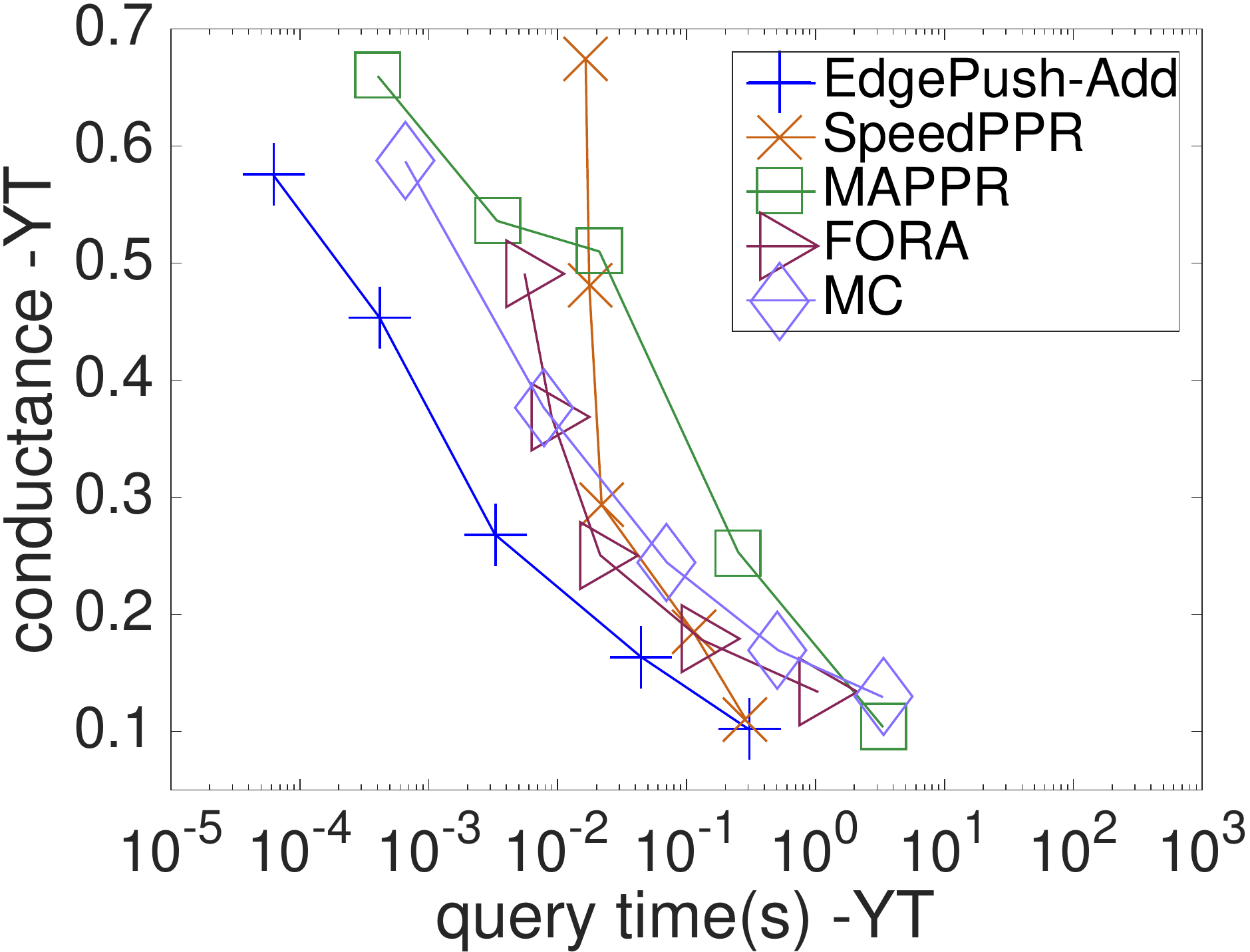} &
			\hspace{-3mm} \includegraphics[width=43mm]{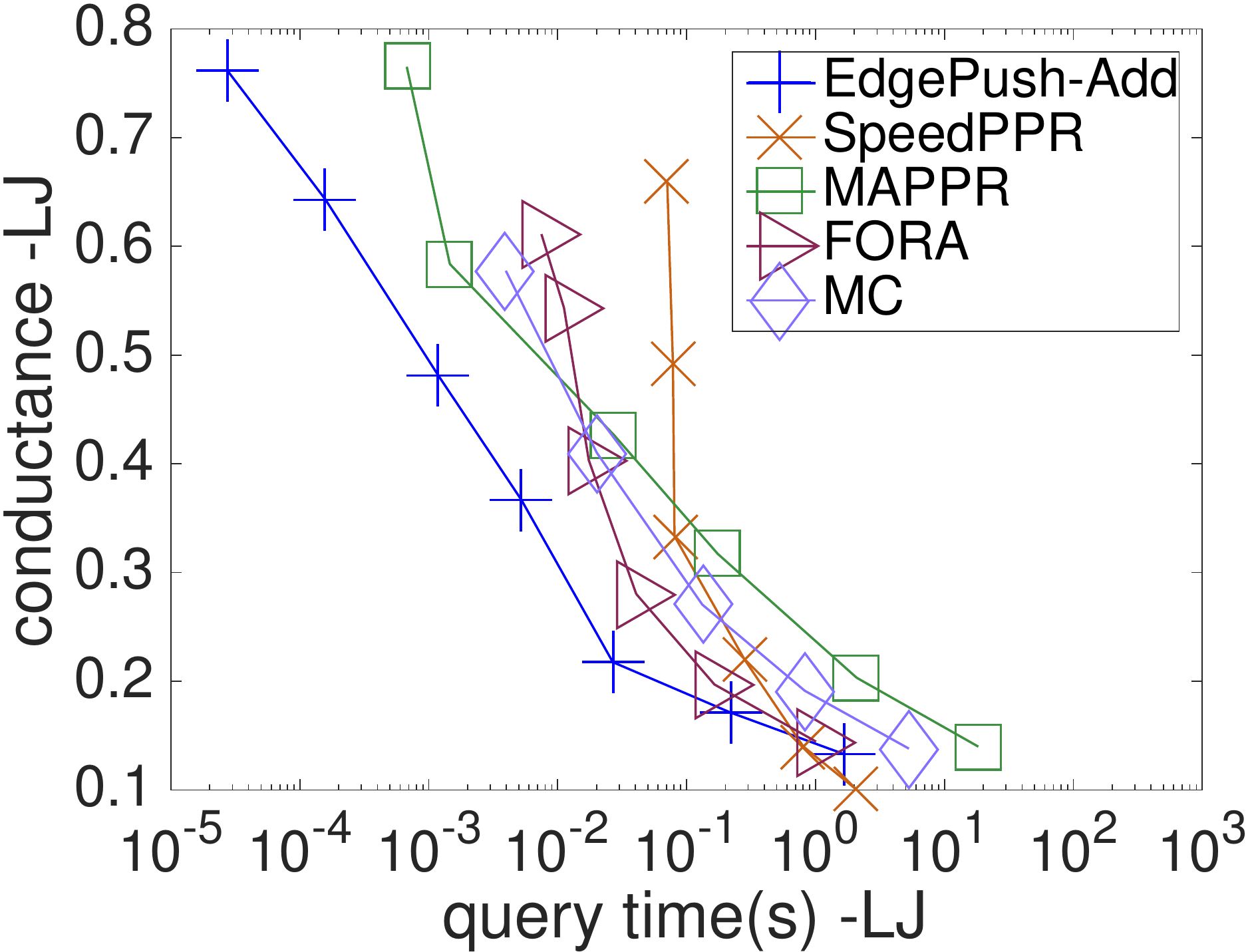} &
			\hspace{-2mm} \includegraphics[width=43mm]{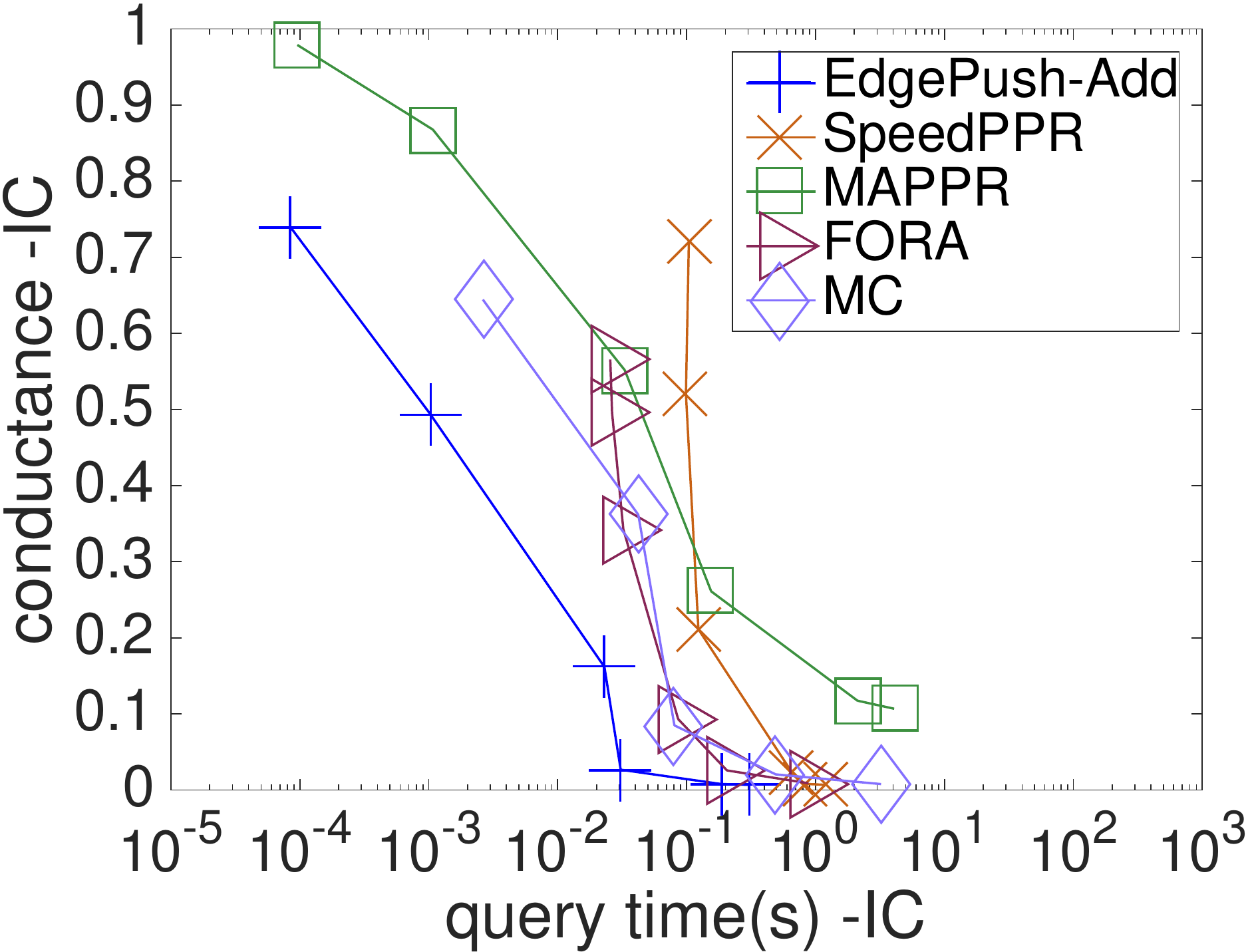} &
			\hspace{-4mm} \includegraphics[width=43mm]{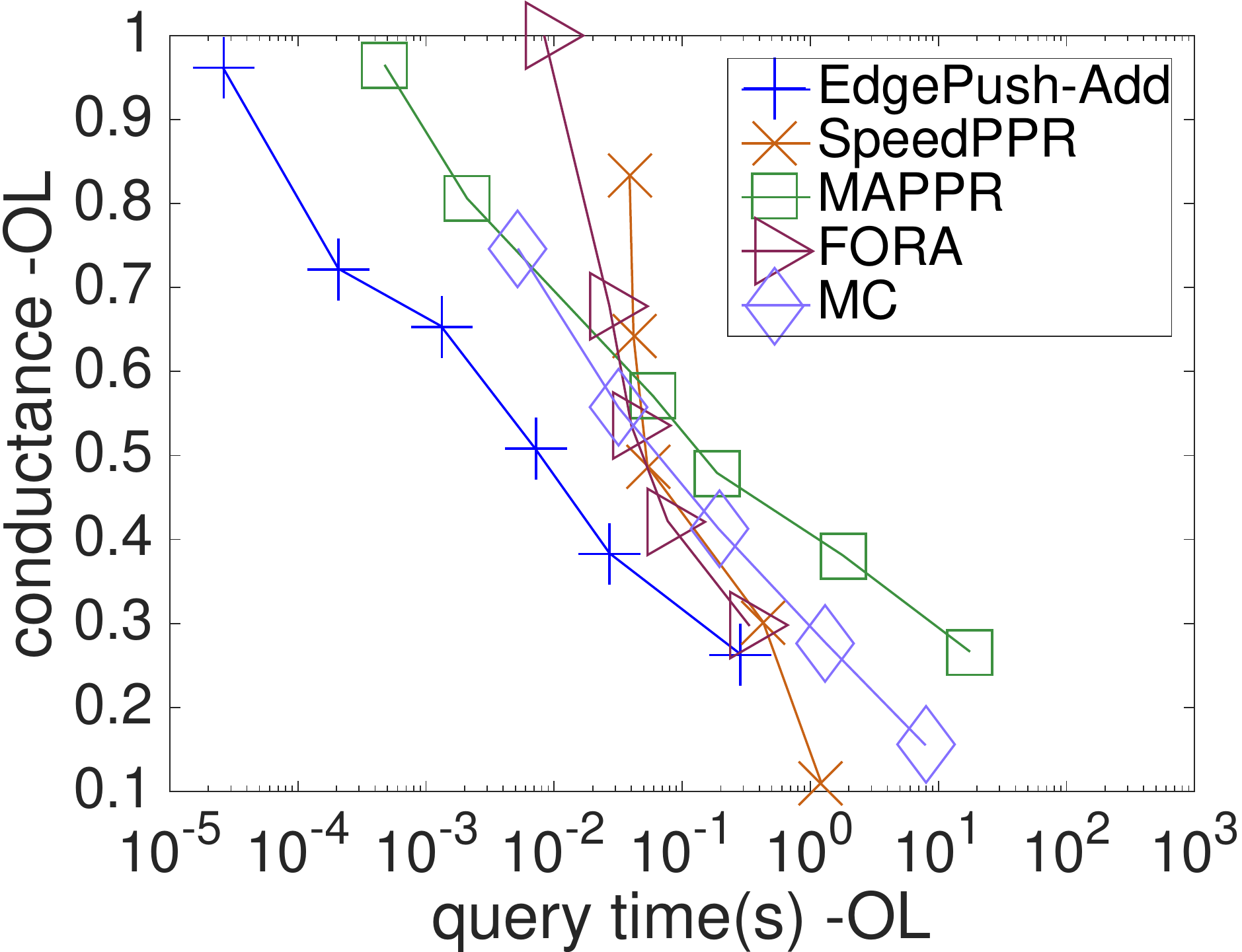} 
		\end{tabular}
		\vspace{-5mm}
		\caption{\hz {\em conductance} v.s. query time on motif-based weighted graphs.}
		\label{fig:conductance-query-real_pro}
		\vspace{-2mm}
	\end{minipage}
\end{figure*}

\subsection{$(a, b)$-Unbalanced Node}
Recall that in Section~\ref{sec:intro}, we also introduce a concept, {\em $(a, b)$-unbalanced} node, to quantify the unbalancedness of the edge weight distributions. 
In this subsection, we will present a further illustration on the definition and properties of {\em $(a, b)$-unbalanced} node. 

As stated in Section~\ref{sec:intro}, a node $v \in V$ is said to be {\em $(a, b)$-unbalanced} if $a\cdot n(v)$ of its edges taking $b \cdot d(v)$ edge weights, where $0 \leq a \leq b \leq 1$. This notion reflects the unbalancedness among the weights of $v$'s edges. For example, if $ a = \frac{1}{n(v)}$ and $b = 1$, it means that only one edge of $v$ taking all the weights of $v$'s edges. To be more specific, the definition of the  $(a,b)$-unbalanced node implies an inequality shown in the following lemma. 

\begin{lemma}\label{lmm:ab}
If all the nodes $v \in V$ are $(a, b)$-unbalanced, we have:
\begin{align}\nonumber
\vspace{-2mm}
\sum_{u \in N(v)}\sqrt{\A_{uv}} 
\le \left(\sqrt{ab} +\sqrt{(1-a)(1-b)}\right) \cdot \sqrt{n(v) d(v)}. 
\vspace{-2mm}
\end{align}
\end{lemma}

{\header {\bf Superiority of \edgepush based on the $(a,b)$-unbalancedness}. }
Relying on the concept of $(a,b)$-unbalanced node, we can also analyze the superiority of \edgepush over \localpush. Specifically, to achieve an $\ell_1$-error $\e$, the following lemma states that the running time bound of \edgepush in Theorem~\ref{thm:edge-efficiency-l1} is superior to that of \localpush in Fact~\ref{thm:bound-local}. 

\begin{lemma}\label{lem:super-l1}
If all the nodes $u \in V$ are $(a, b)$-unbalanced, we have
\vspace{-2mm}
\begin{align}\label{eqn:super-l1-ab}
\frac{(1-\alpha)}{\alpha \e \|\A\|_1}\cdot \left(\sum_{\la u,v \ra \in \bar{E}}\sqrt{\A_{uv}}\right)^2 
\hspace{-2mm}\leq \hspace{-1mm}\left(\sqrt{ab} +\sqrt{(1-a)(1-b)}\right)^2\hspace{-2mm} \cdot \frac{2m}{\alpha \e}\,. 
\end{align}
\end{lemma}

Recall that the expected overall running time of \edgepush is $\frac{(1-\alpha)}{\alpha \e \|\A\|_1}\cdot \left(\sum_{\la u,v \ra \in \bar{E}}\sqrt{\A_{uv}}\right)^2$, which equals to the left hand of Equation~\eqref{eqn:super-l1-ab}. Moreover, the expected running time of \localpush is $\frac{2m}{\alpha \e}$. Thus, Lemma~\ref{lem:super-l1} illustrates a $\left(\sqrt{ab} +\sqrt{(1-a)(1-b)}\right)^2$ superiority of \edgepush over \localpush. For the sake of simplicity, we denote $\gamma = \left(\sqrt{ab} +\sqrt{(1-a)(1-b)}\right)^2$. 

Likewise, the following lemma shows that \edgepush still outperforms \localpush by $\gamma$ in terms of the overall running time bound with normalized additive error $r_{max}$. 
\begin{lemma}\label{lem:super-add}
If all the nodes $u \in V$ are $(a, b)$-unbalanced, we have
\vspace{-2mm}
\begin{align}\nonumber
\frac{(1-\alpha)}{\alpha r_{\max}\|\A\|_1}\hspace{-0.5mm}\cdot \hspace{-1mm}\sum_{v\in V}\hspace{-0.5mm}\frac{1}{d_v}\hspace{-0.5mm}\left(\sum_{u\in N_v}\hspace{-2mm}\sqrt{\A_{uv}}\hspace{-0.5mm}\right)^2 
\hspace{-1.5mm}\leq \hspace{-0.5mm}\left(\hspace{-0.5mm}\sqrt{ab}\hspace{-0.5mm}+\hspace{-0.5mm} \sqrt{(\hspace{-0.5mm}1\hspace{-1mm}-\hspace{-1mm}a)(1\hspace{-1mm}-\hspace{-1mm}b)}\hspace{-0.5mm}\right)^2 \hspace{-1.5mm} \cdot \hspace{-0.5mm}\frac{2m}{\alpha r_{\max}\|\A\|_1}\,. 
\end{align}
\end{lemma}

The proof of Lemma~\ref{lem:super-l1} and Lemma~\ref{lem:super-add} are both deferred to the appendix for readability. Based on the two lemmas, we can prove the three implications shown in ``Our contributions" in Section~\ref{sec:intro}. That is, the overall running time bound of \edgepush is no worse than that of \localpush because $\gamma \leq 1$ holds for all values of $a$ and $b$. Secondly, the superiority of \edgepush over \localpush can indeed be quantified by the $(a,b)$-unbalancedness as shown in Lemma~\ref{lem:super-l1} and Lemma~\ref{lem:super-add}. Specifically, when $a = o(1)$ and $b = 1 - o(1)$, we have $\gamma = o(1)$, which implies that \edgepush can achieve a sub-linear time complexity $o\left(\frac{m}{\alpha \e}\right)$ for solving the approximate SSPPR problem with specified $\ell_1$-error $\e$. 
Additionally, in the extreme case, when the graph $G$ is a complete graph, $a = 1 / n$ and $b = 1 - 1/n$,  we have $\gamma =O(\frac{1}{n})$, implying that, in this case, \edgepush outperforms \localpush by a $O(n)$ factor. 


\header{\bf Comparisons between the notions of $(a,b)$-unbalanced node and $\boldsymbol{{\cos^2 \p}}$. } Intuitively, both the notions of $\cos^2 \p$ (resp., $\cos^2 \p_v)$ and the $(a,b)$-unbalanced node capture the {\em unbalance} of the undirected weighted graph $G$.
However, 
the former is actually more general than the latter.
In other words, the superior results obtained using the notion of $(a,b)$-unbalanced are based on a more restricted requirement. 
To see this, recall that in lemma~\ref{lem:super-l1}, we assume that all nodes $u$ in the graph are $(a,b)$-unbalanced. Or equivalently, we use the maximum of $a/b$ in the graph to bound the superiority of \edgepush. Whenever a graph has a node with degree $1$ (e.g. node $w$ in Figure~\ref{fig:weighted_graph}), we have $a=b=1$, leading to $\gamma=1$. Thus, in this case, Lemma~\ref{lem:super-l1} only suggests that the overall expected running time of \edgepush is no worse than that of \localpush. However, in this case, we observe $\cos^2 \p=O\left(\frac{1}{n}\right)$. By Lemma~\ref{lem:cos-l1}, this indicates \edgepush outperforms \localpush by a $O(n)$ factor. The intuition here is \edgepush is sufficient to significantly outperform \localpush as long as a subset of nodes in the graphs have unbalanced weight distributions, rather than all nodes.

\vspace{-2mm}
\section{Experiments} 
\label{sec:exp}
\begin{figure*}[t]
	\begin{minipage}[t]{1\textwidth}
		\centering
		\begin{tabular}{cccc}
			\hspace{-4mm} \includegraphics[width=43mm]{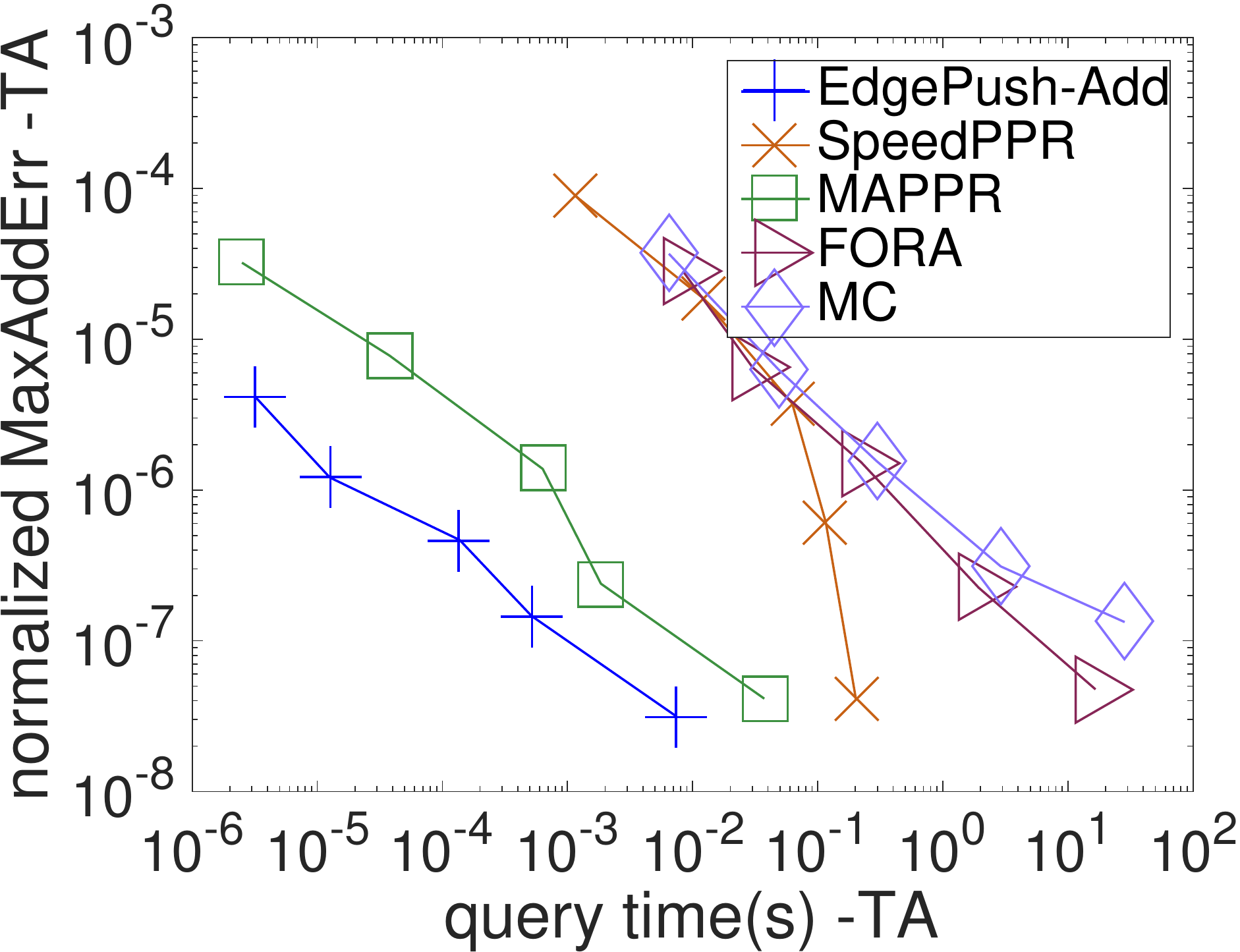} &
			\hspace{-3mm} \includegraphics[width=43mm]{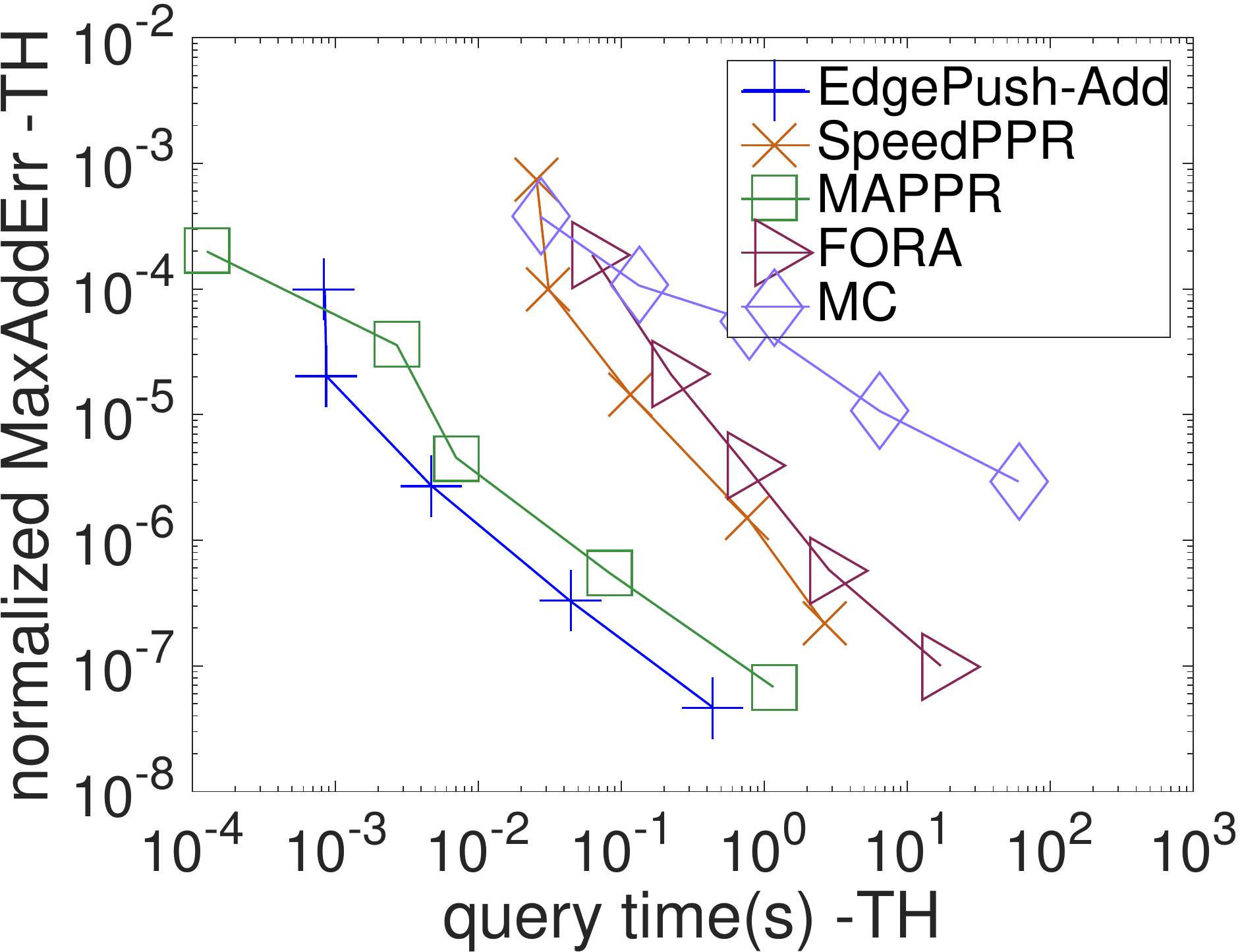} &
			\hspace{-3mm} \includegraphics[width=43mm]{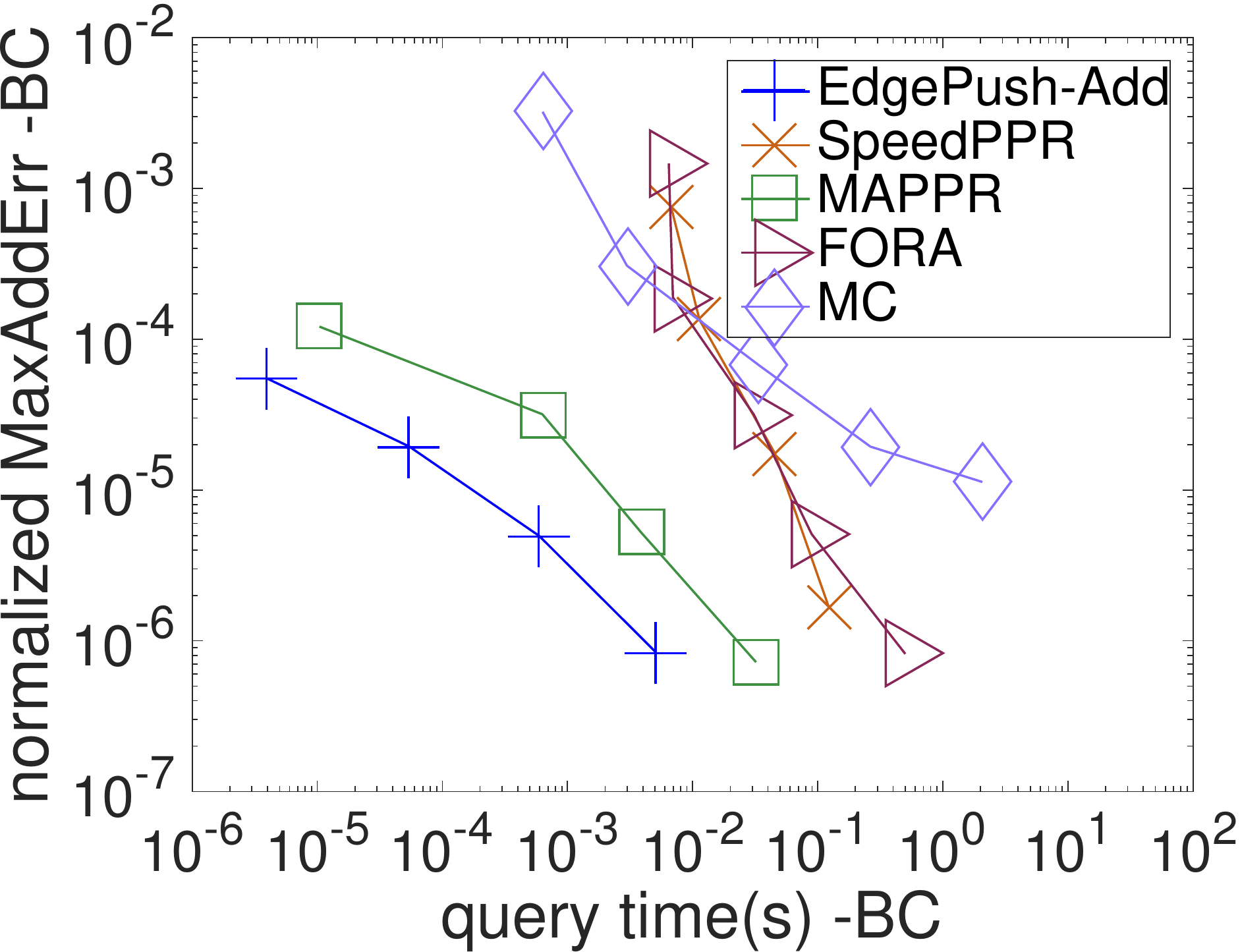}&
			\hspace{-3mm} \includegraphics[width=43mm]{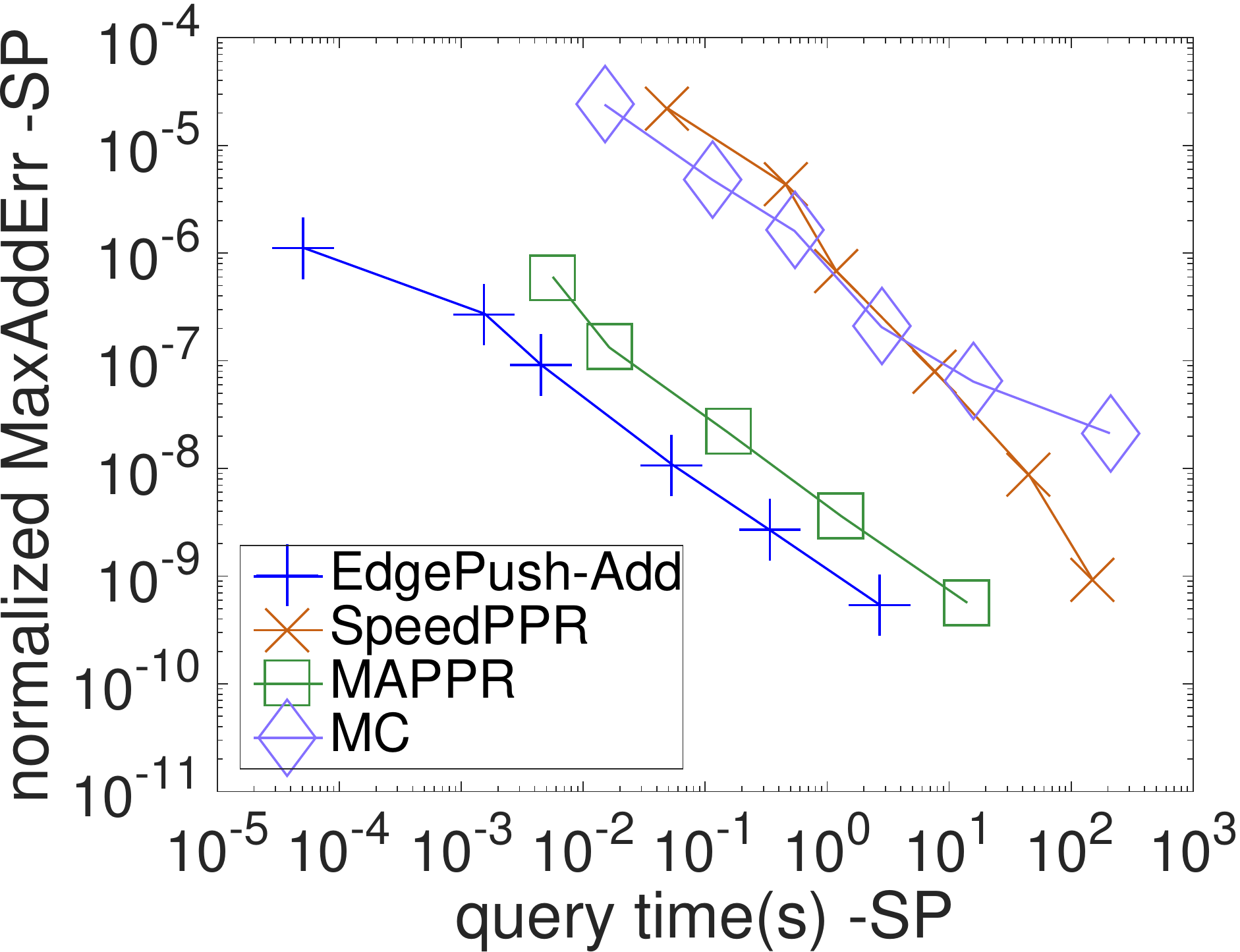}\\
		\end{tabular}
		\vspace{-5mm}
		\caption{{\hz {\em normalized MaxAddErr} v.s. query time on real weighted graphs.}}
		\label{fig:maxerror-query-real_pro-real}
		\vspace{-1mm}
    \end{minipage}

    \begin{minipage}[t]{1\textwidth}
		\centering
		\vspace{+0.5mm}
		\begin{tabular}{cccc}
			\hspace{-4mm} \includegraphics[width=43mm]{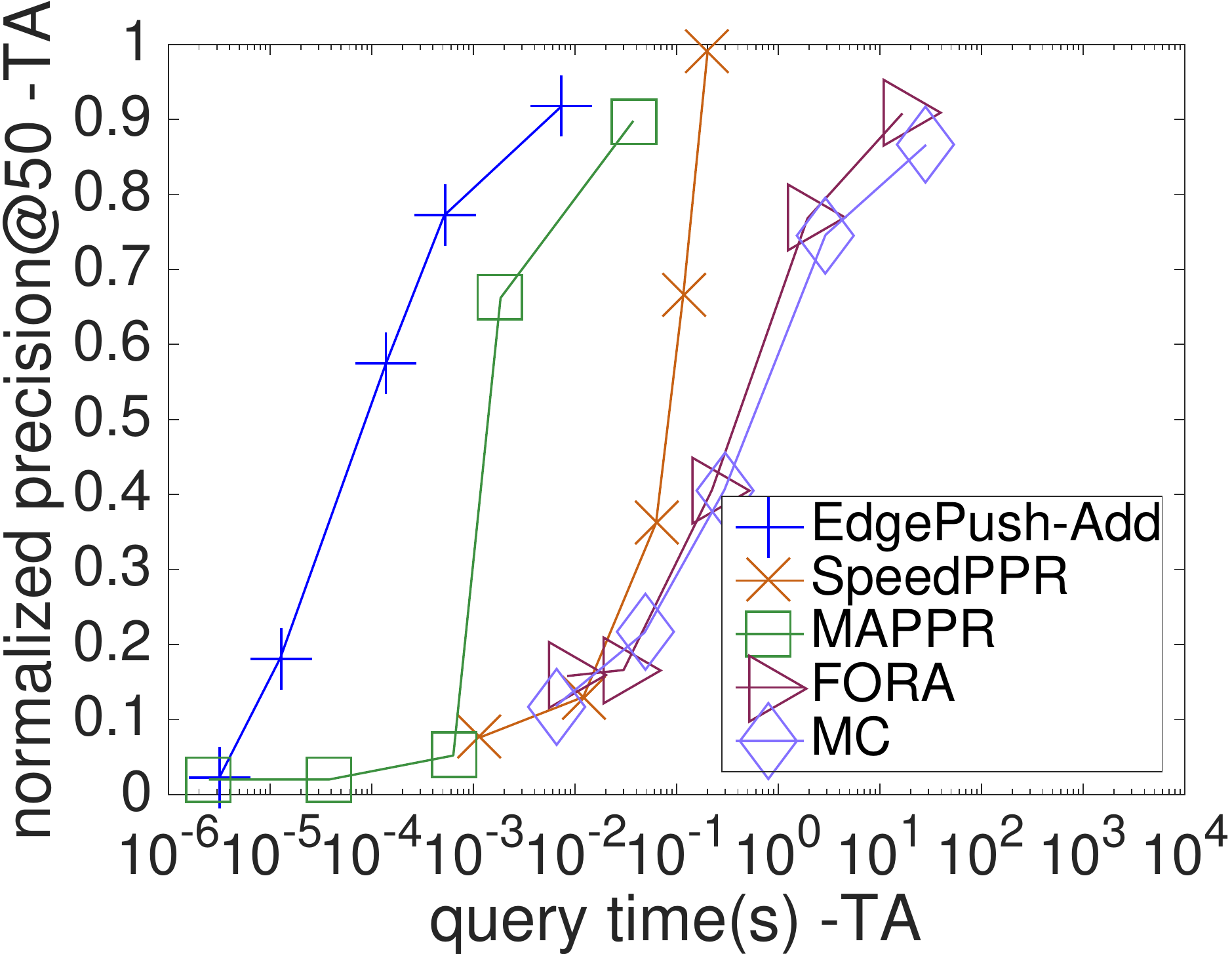} &
			\hspace{-3mm} \includegraphics[width=43mm]{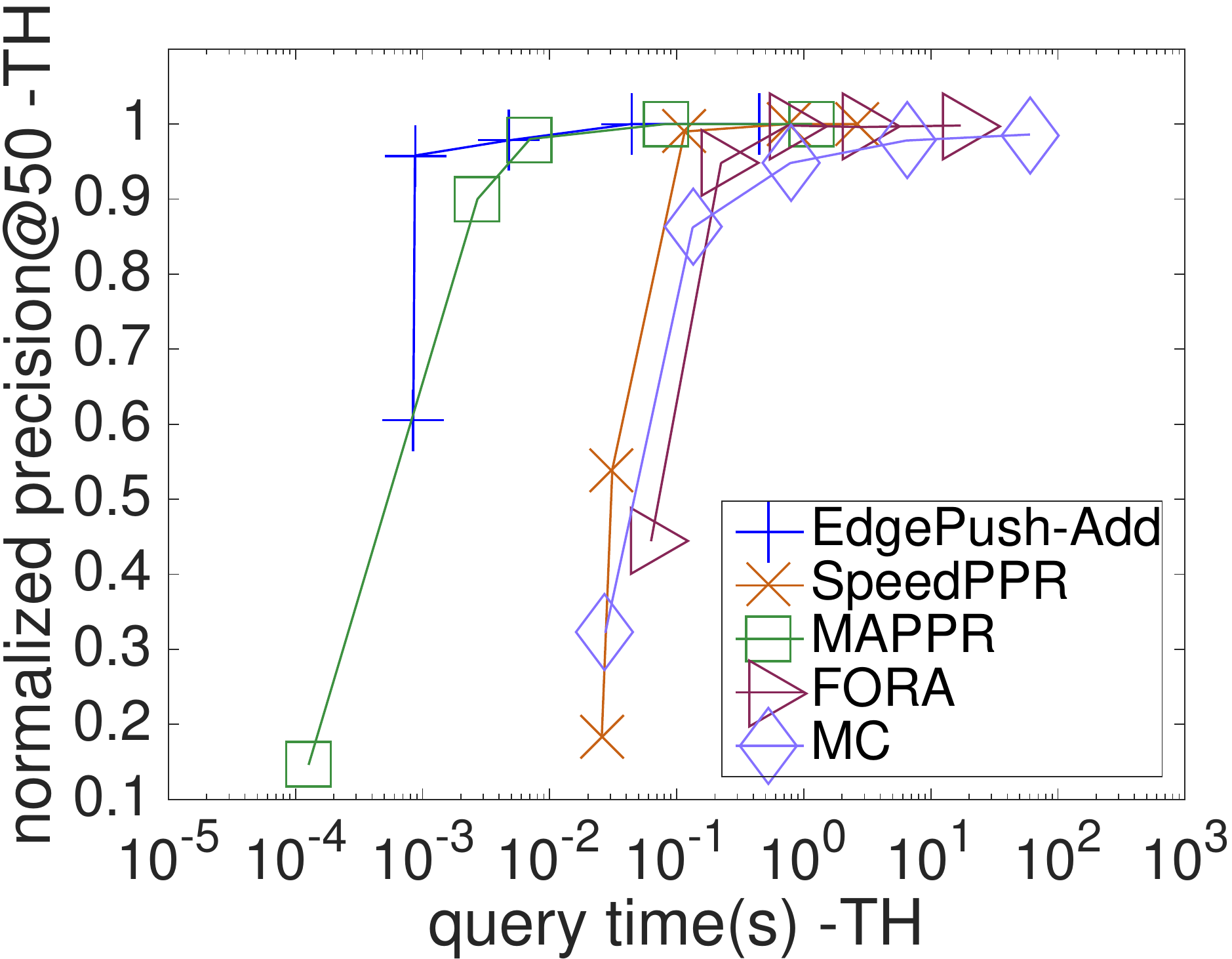} &
			\hspace{-3mm} \includegraphics[width=43mm]{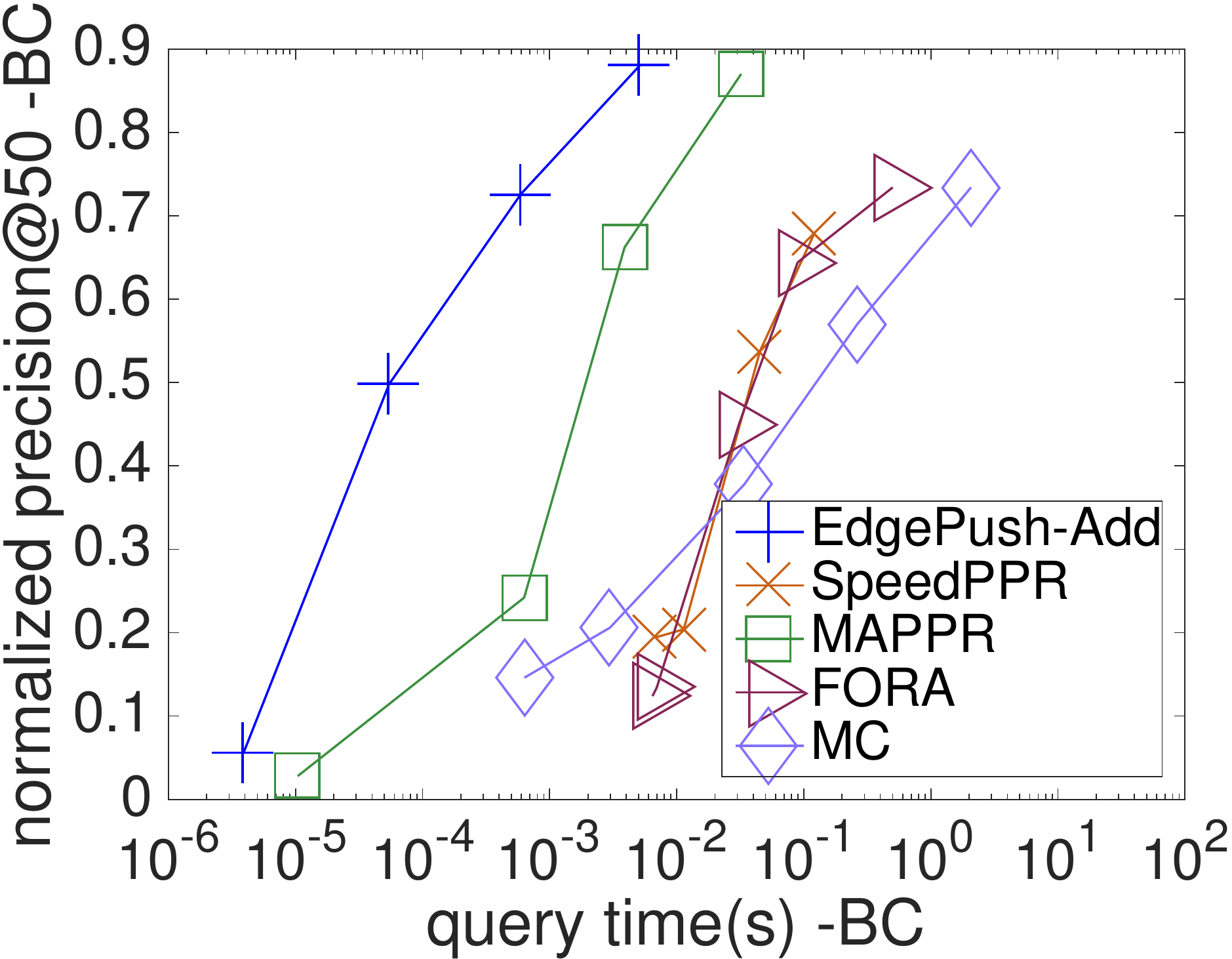} &
			\hspace{-3mm} \includegraphics[width=43mm]{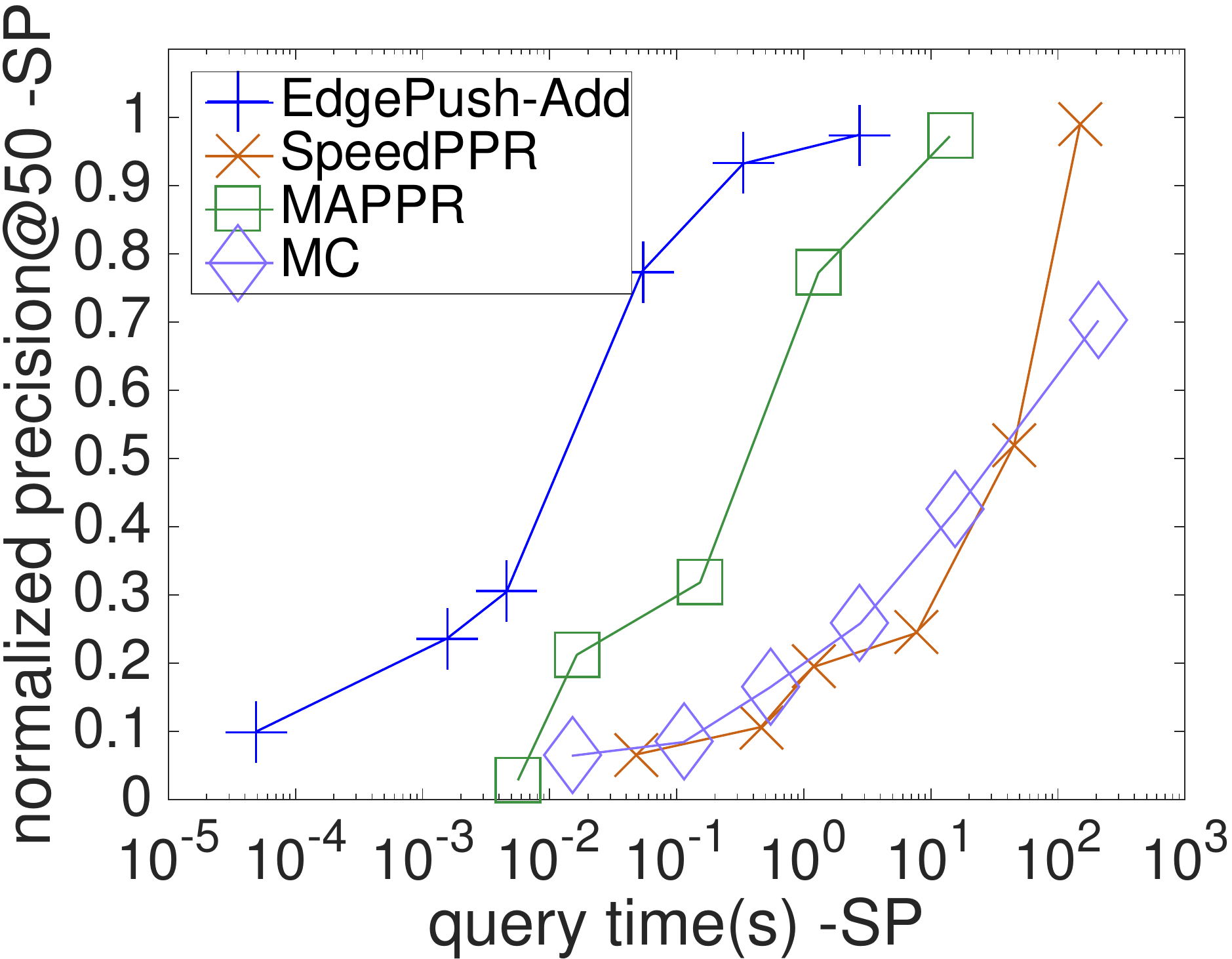}\\
		\end{tabular}
		\vspace{-5mm}
		\caption{\hz {\em normalized precision@50} v.s. query time on real weighted graphs.}
		\label{fig:precision-query-real_pro-real}
		\vspace{-1mm}
	\end{minipage}

	\begin{minipage}[t]{1\textwidth}
		\centering
		\vspace{+0.5mm}
		\begin{tabular}{cccc}
			\hspace{-4mm} \includegraphics[width=43mm]{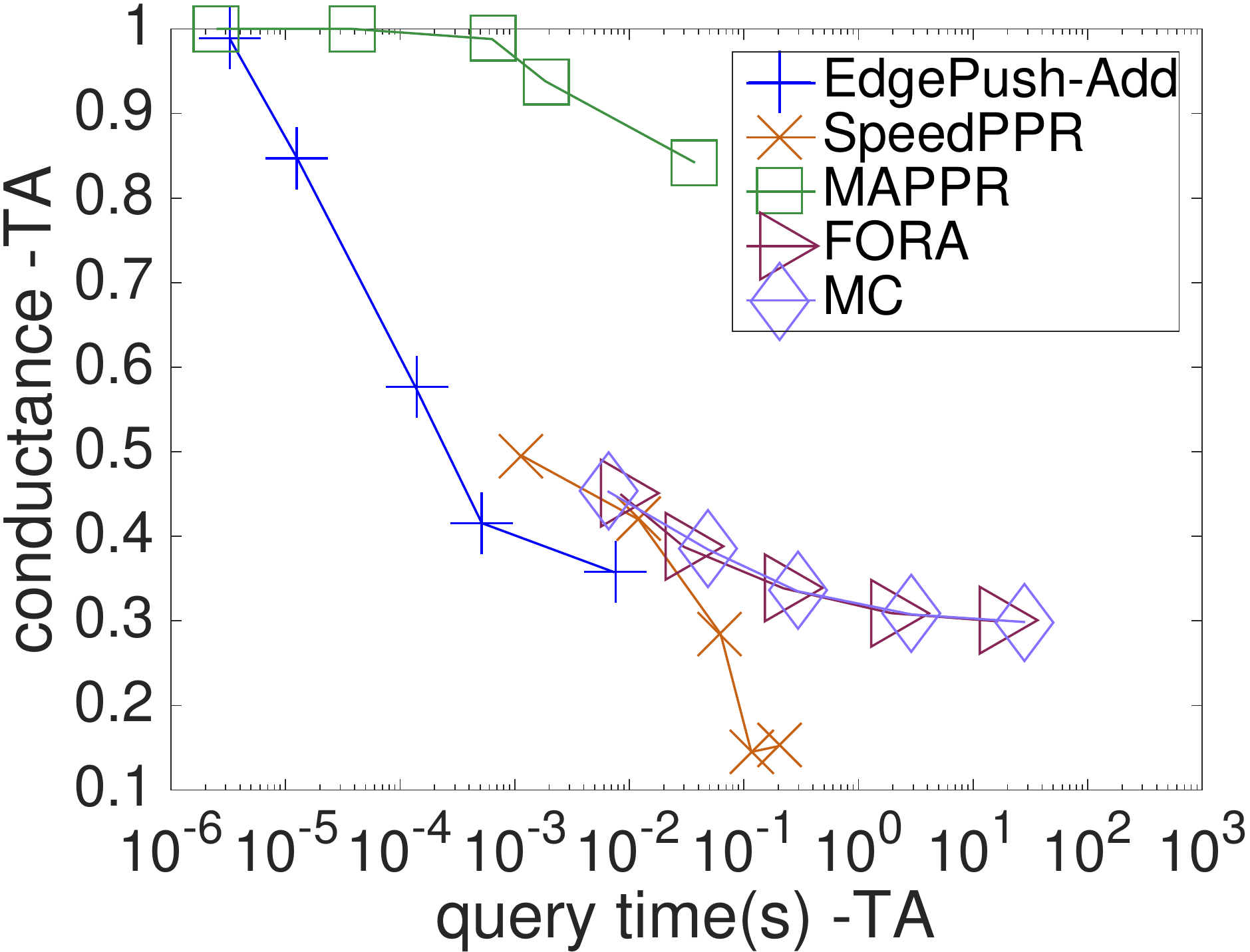} &
			\hspace{-3mm} \includegraphics[width=43mm]{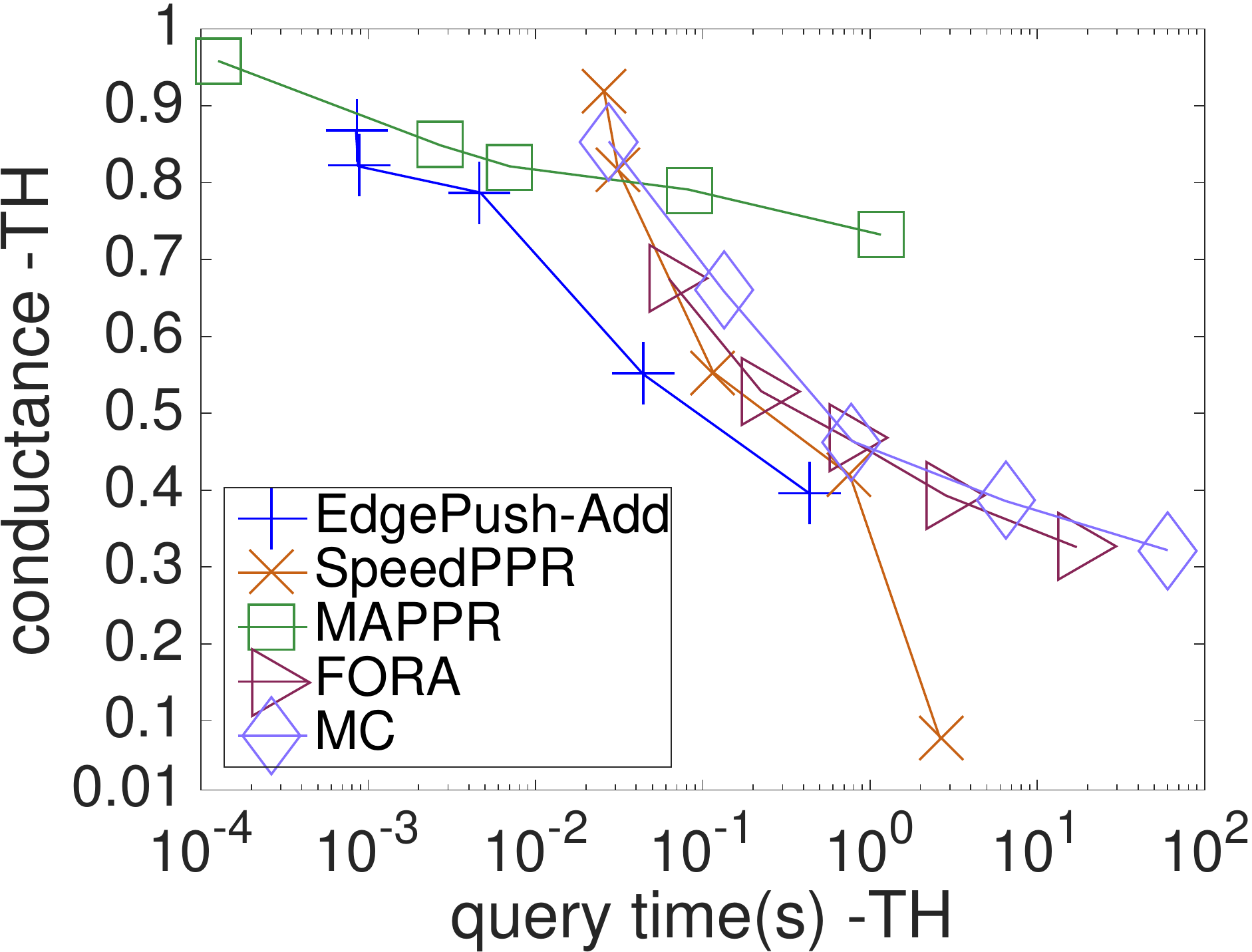} &
			\hspace{-3mm} \includegraphics[width=43mm]{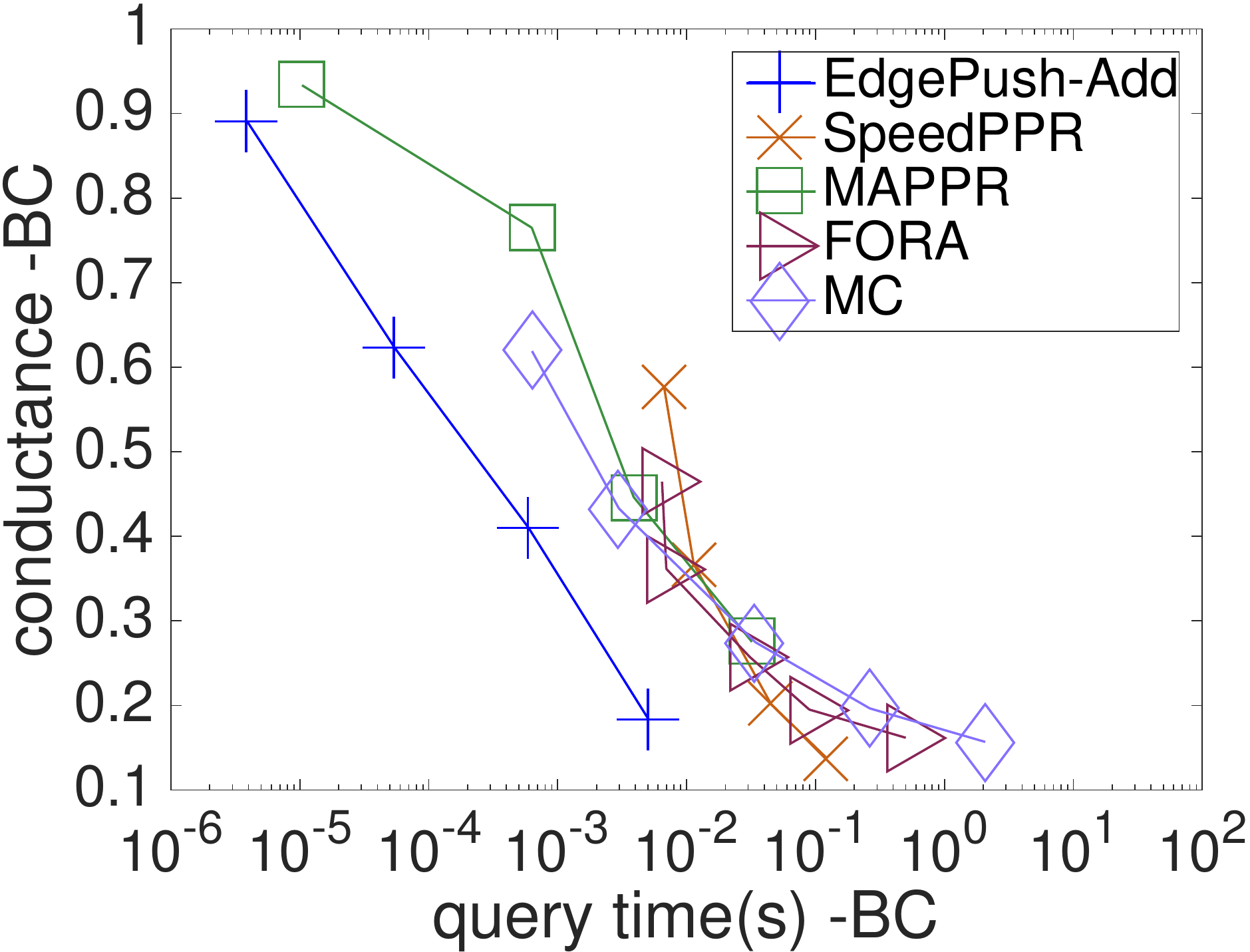} &
			\hspace{-3mm} \includegraphics[width=43mm]{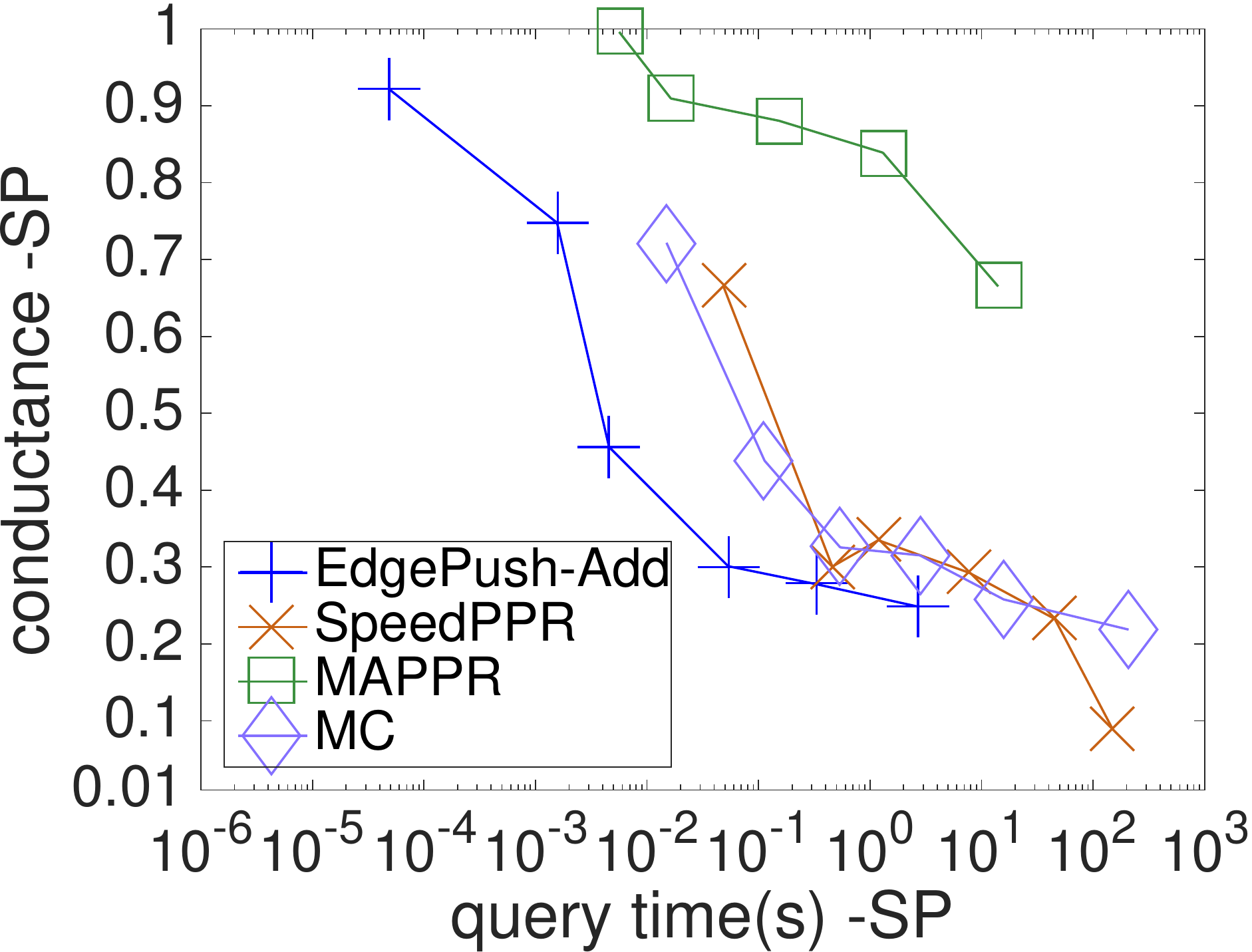} \\
		\end{tabular}
		\vspace{-5mm}
		\caption{\hz {\em conductance} v.s. query time on real weighted graphs.}
		\label{fig:conductance-query-real_pro-real}
		\vspace{-2mm}
	\end{minipage}
\end{figure*}

In this section, we conduct experiments to show the effectiveness of \edgepush on large real-world datasets. 

\header{\bf Experiment environment.}
We conduct all the experiments on a machine with an Intel(R) Xeon(R) Gold 6126@2.60GHz CPU and 500GB memory in Linux OS. All the methods are implemented in C++ compiled by g++ with O3 turned on. 

\header{\bf \hz Datasets.}
In the experiments, we use eight real-world datasets: Youtube~\cite{yang2015OL_YTdata}, LiveJournal~\cite{leskovec2009LJdata}, IndoChina~\cite{Paolo2004LAWdata1,Paolo2004LAWdata2}, Orkut-Links~\cite{yang2015OL_YTdata}, Tags~\cite{benson2018simplicial}, Threads~\cite{benson2018simplicial}, BlockChair~\cite{blockchair} and Spotify~\cite{kumar2020retrieving}. All the datasets are available at~\cite{snapnets,LWA,cornell,blockchair-transactions}. The first four datasets (i.e. YT, LJ, IC and OL) are unweighted and undirected graphs. Following~\cite{yin2017MAPPR},we convert the four unweighted graphs to motif-based weighted graphs by counting the number of motifs. More precisely, we first calculate the motif number $\phi(e)$ of each edge $e \in E$, which is defined as the number of motifs that $e$ participates in. Then we set the weight of $e$ as $\phi(e)$ to obtain a weighted graph. Note that $\phi(e)$ might be $0$ if $e$ doesn't participate in any motif. In the experiments, we set the type of motif as ``clique3'' defined in~\cite{motif}. The other four datasets (i.e. TA, TH, BC and SP) are real-world weighted graphs. Specifically, Tags (TA) and Threads (TH) are derived from Stack Overflow, a well-known Q\&A (Question-and-Answer) website for programmers. In Tags, nodes are tags of questions and edge weights are the number of questions jointly annotated by two tags. In comparison, nodes in Threads represent users and the edge weight indicates how frequently two users appear in the same question thread. BlockChair (BC) is a bitcoin transaction network, where nodes are addresses and edge weights are the number of transactions between two addresses. In our experiments, we use the transactions records on December 5, 2021. Spotify (SP) is a music streaming network. The nodes in Spotify represent songs and the edge weight is the number of times that two songs appear in one session. More detailed descriptions of these real weighted datasets can be founded in~\cite{kumar2020retrieving}.

In Table~\ref{tbl:datasets}, we list the meta data of all these datasets. 
{\crc We also report the average and maximum edge weight and the value of $\cos^2 \p$ to quantify the unbalancedness of graphs. Recall that the $\cos^2 \p$ notation is defined in Lemma~\ref{lem:cos-l1}. The smaller the $\cos^2 \p$ is, the more unbalanced the graph is. }

\vspace{-1mm}
\header{\bf Ground truths and query sets.}
In the experiments, we employ Power Method~\cite{page1999pagerank} to compute the ground truth results. 
More precisely, we compute Equation~\eqref{eqn:powerdef} for  $100$ iterations 
and regard the returned results as ground truths for comparison. For each dataset, we randomly generate $10$ source nodes for SSPPR queries according to the degree distribution. We issue one SSPPR query from each query node and report the average performances over the 10 query nodes for each method and each set of parameters. 

\begin{figure*}[t]
    \begin{minipage}[t]{1\textwidth}
		\centering
		\begin{tabular}{cccc}
			\hspace{-3mm} \includegraphics[width=41.5mm]{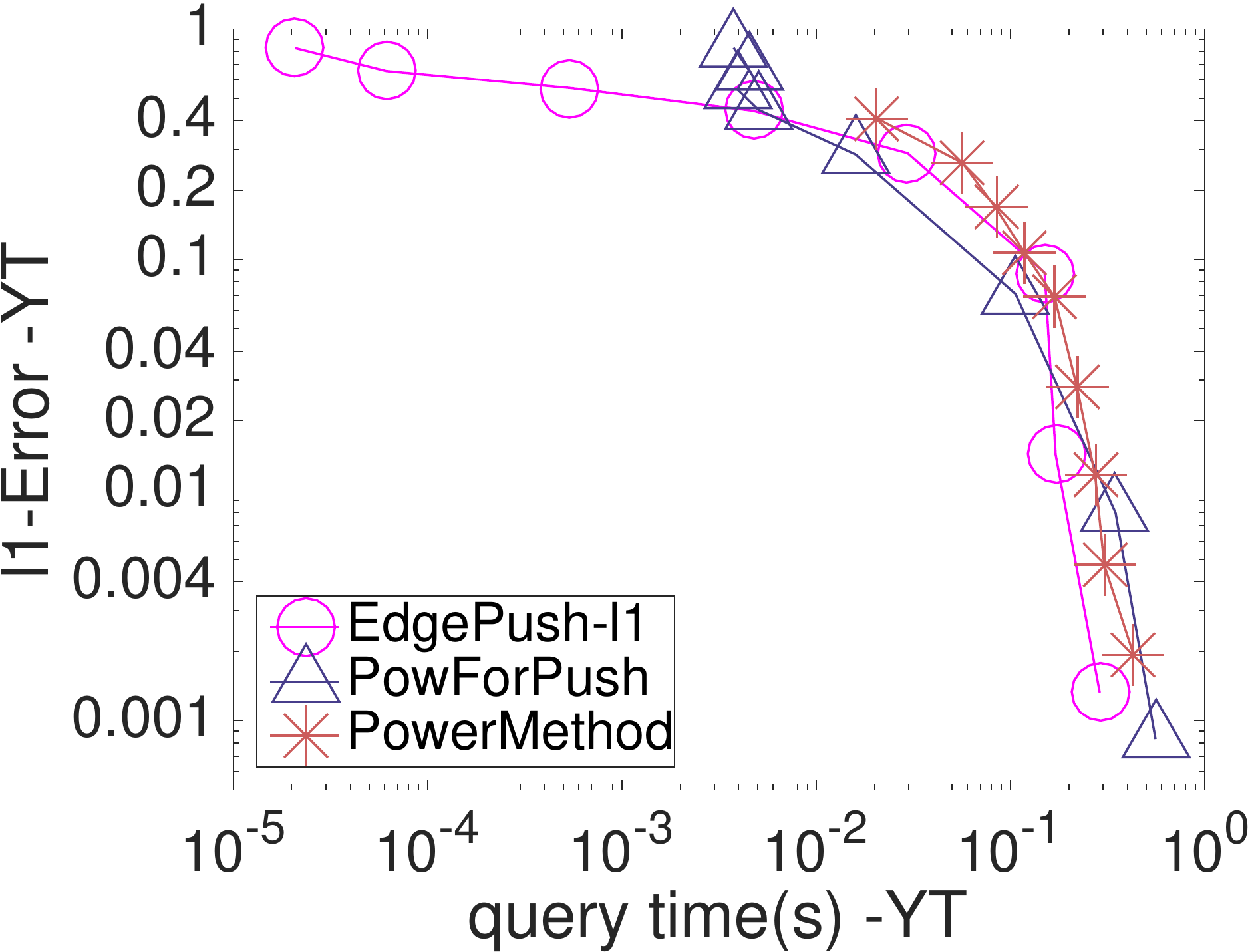} &
    		\hspace{-2mm} \includegraphics[width=41mm]{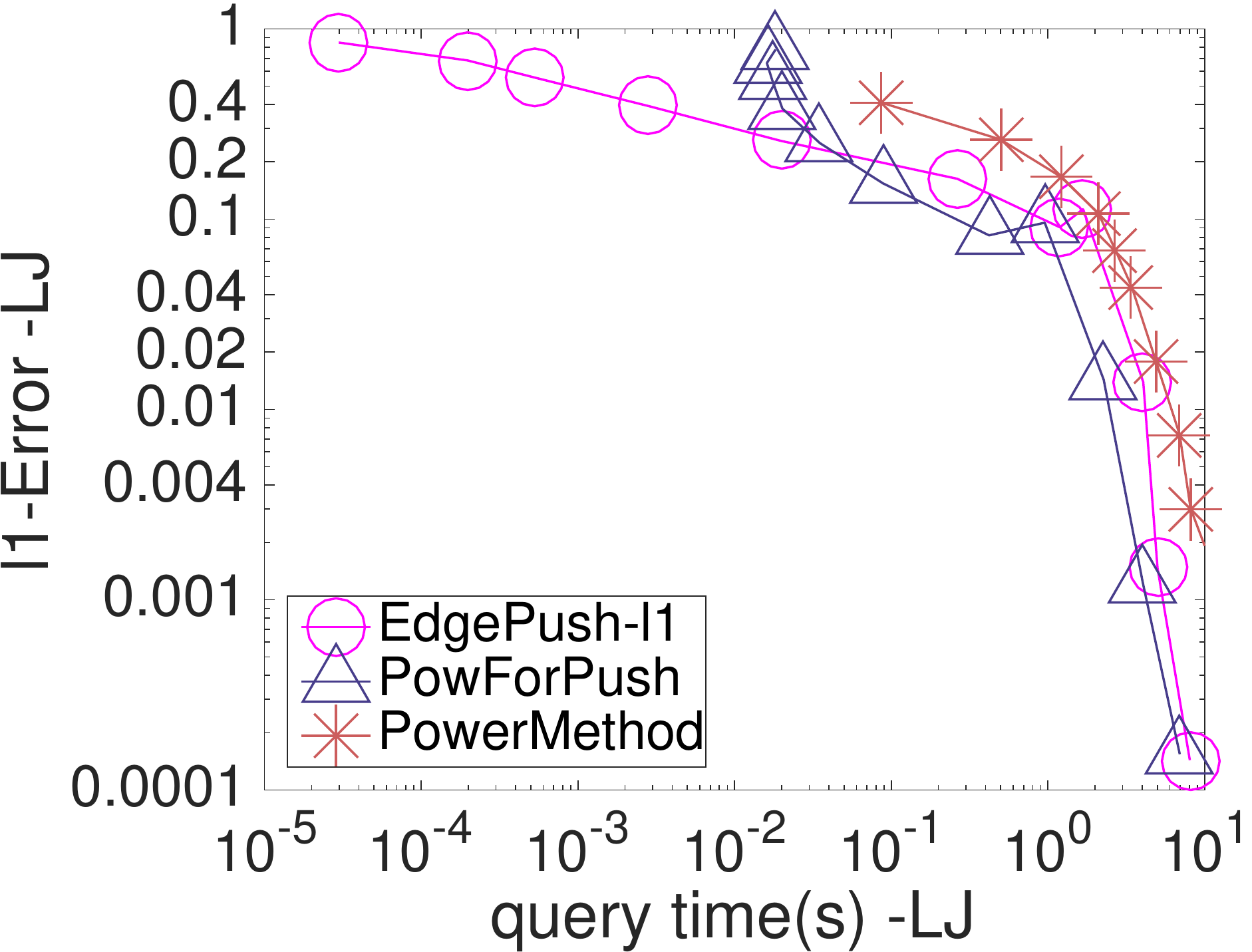} &
			\hspace{-2mm} \includegraphics[width=41mm]{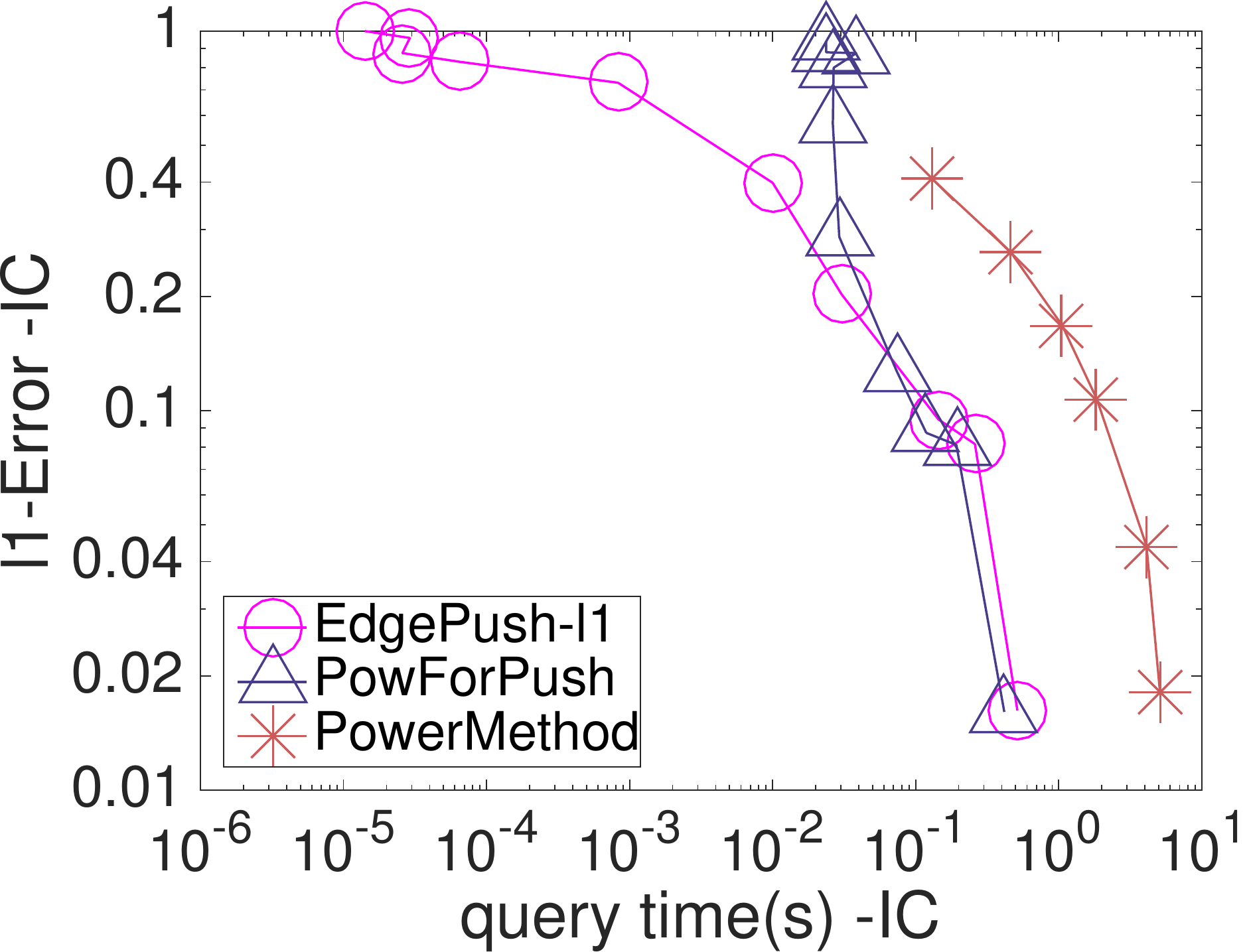} &
			\hspace{-2mm} \includegraphics[width=41mm]{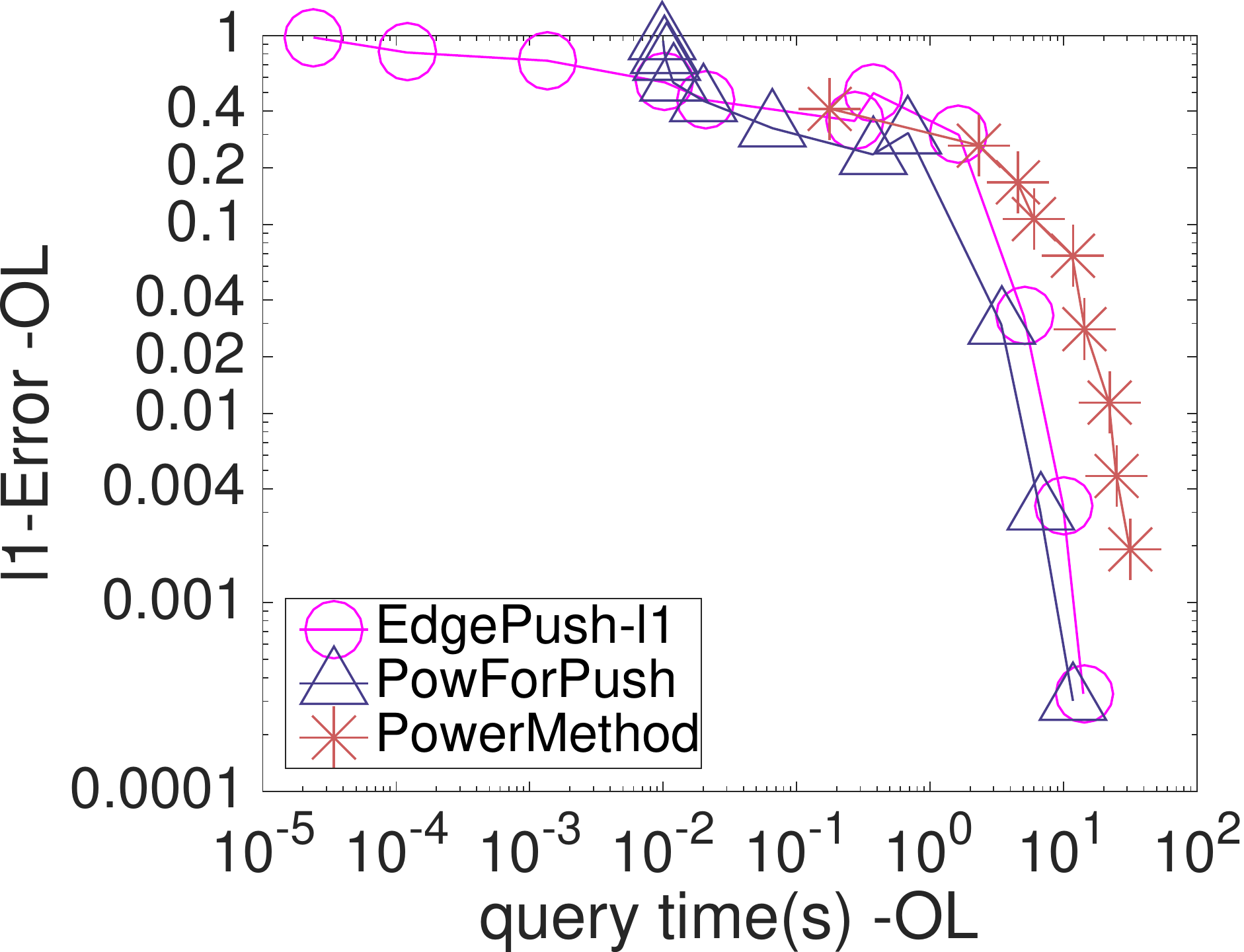} 
		\end{tabular}
		\vspace{-5mm}
		\caption{\hz {\em $\ell_1$-error} v.s. query time on motif-based weighted graphs.}
		\label{fig:l1error-query-real_l1}
		\vspace{-1mm}
    \end{minipage}
    
	\begin{minipage}[t]{1\textwidth}
		\centering
		\vspace{+0.5mm}
		\begin{tabular}{cccc}
			\hspace{-3mm} \includegraphics[width=41mm]{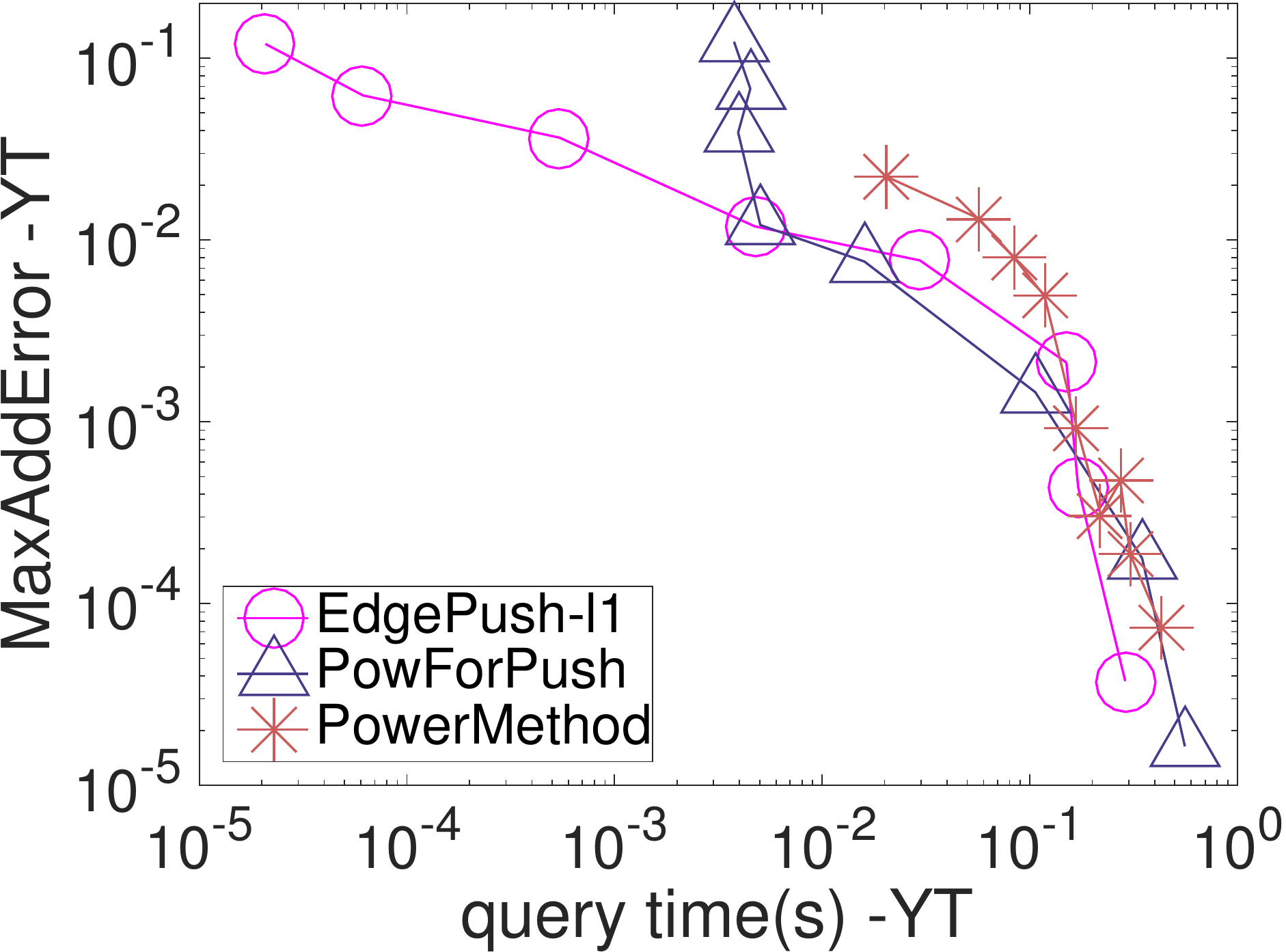} &
			\hspace{-2mm} \includegraphics[width=41mm]{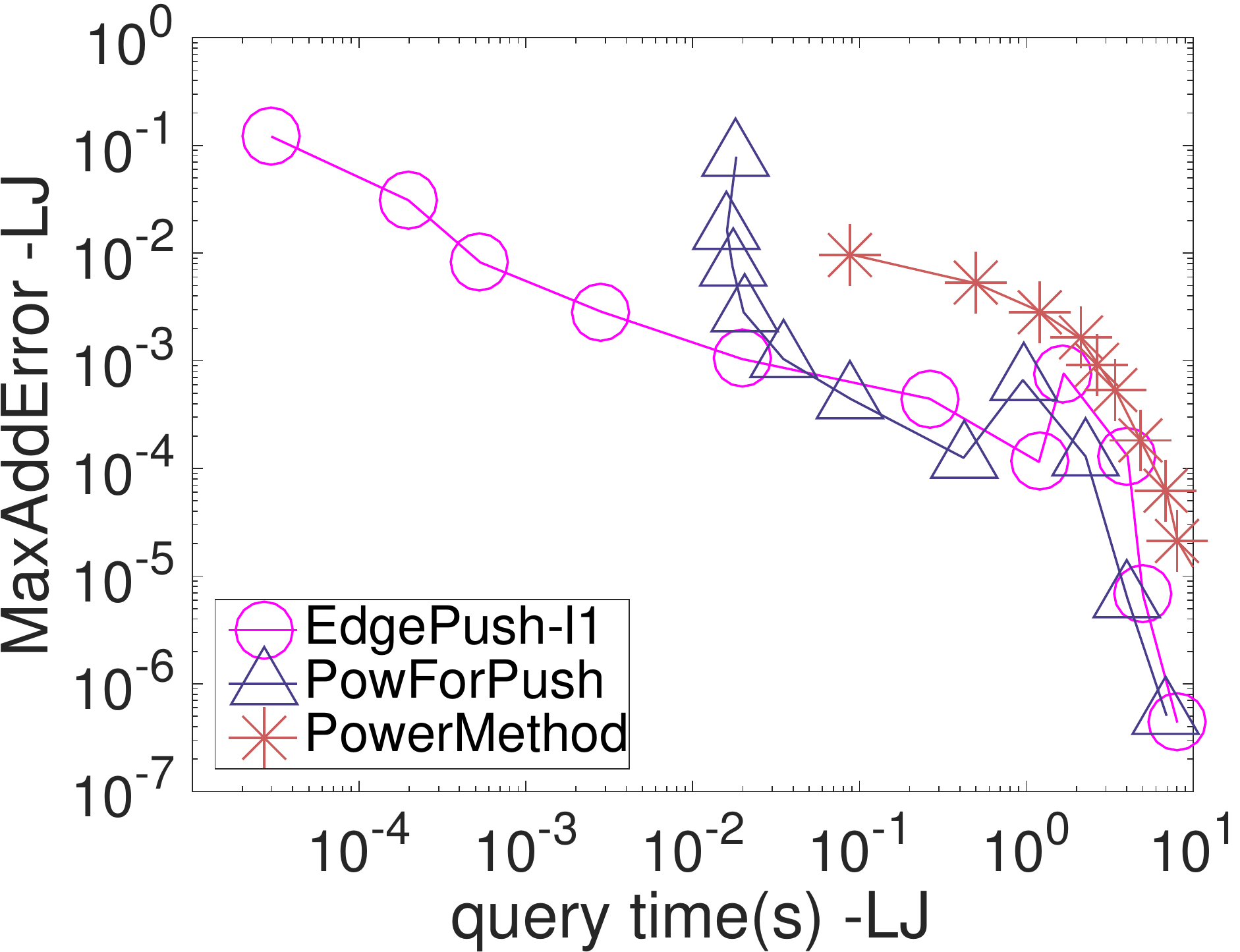} &
			\hspace{-2mm} \includegraphics[width=41mm]{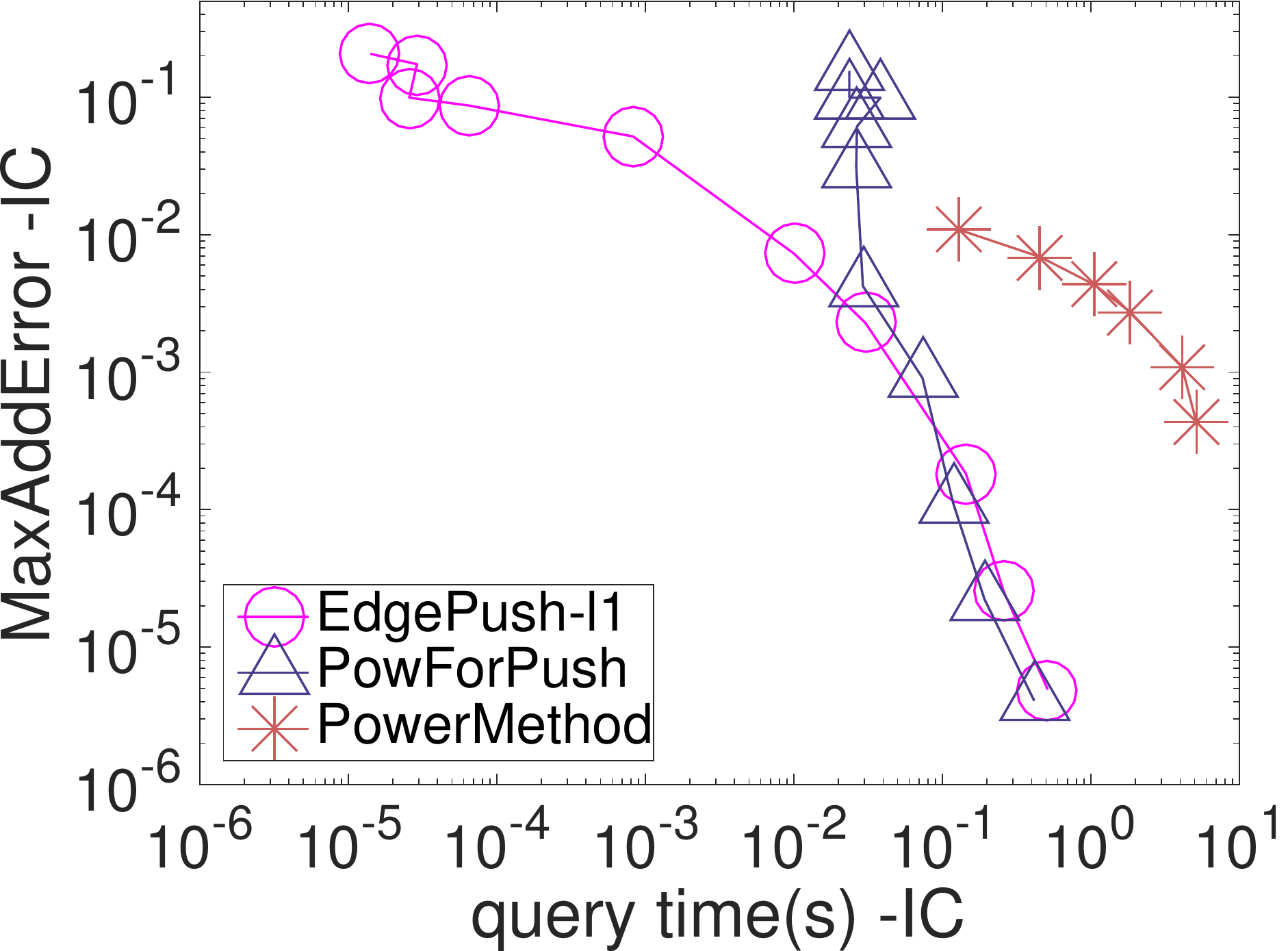} &
			\hspace{-2mm} \includegraphics[width=41mm]{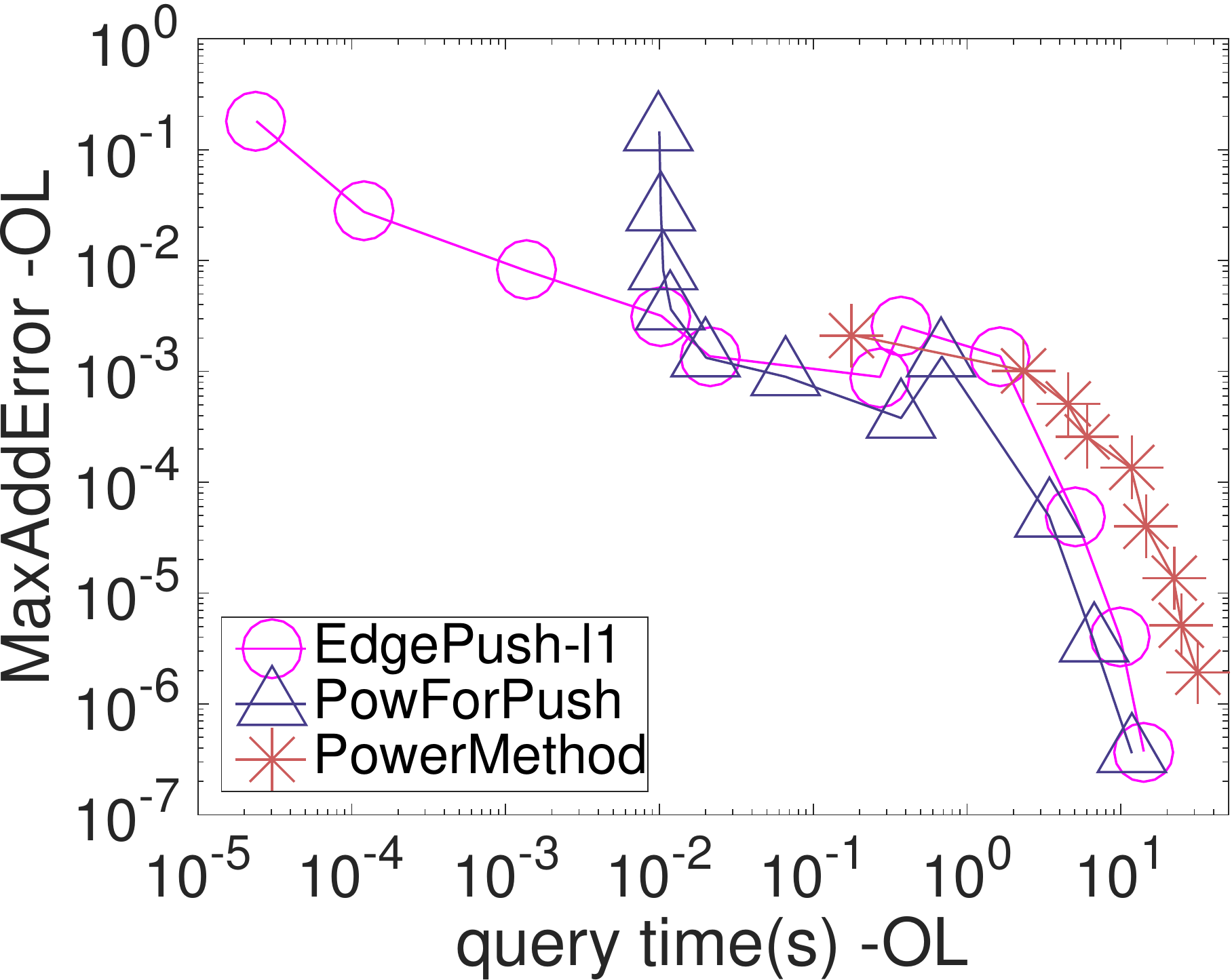} 
		\end{tabular}
		\vspace{-5mm}
		\caption{\hz {\em MaxAddErr} v.s. query time on motif-based weighted graphs.}
		\label{fig:maxerror-query-real_l1}
		\vspace{-1mm}
    \end{minipage}

    \begin{minipage}[t]{1\textwidth}
		\centering
		\vspace{+0.5mm}
		\begin{tabular}{cccc}
			\hspace{-3mm} \includegraphics[width=41mm]{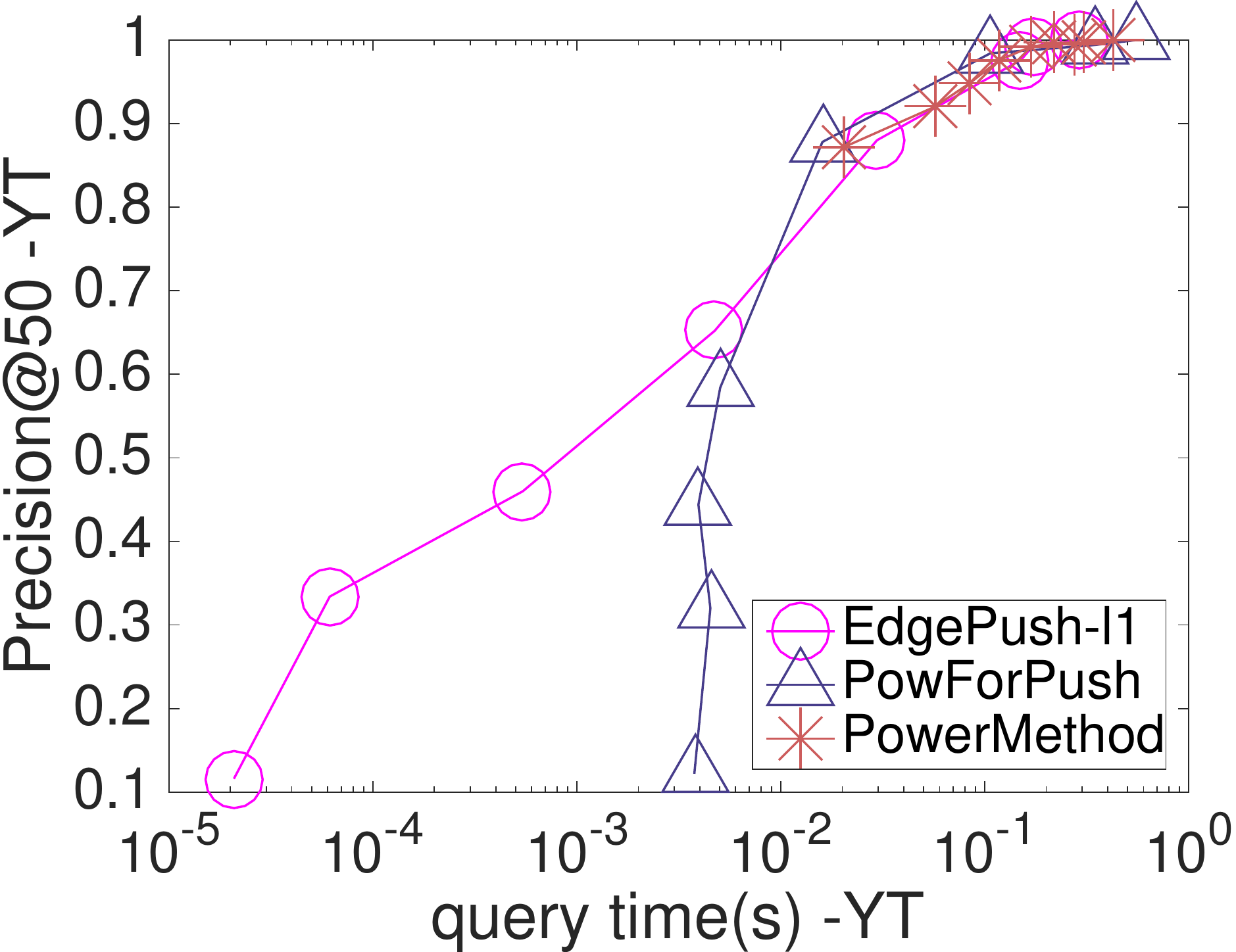} &
			\hspace{-2mm} \includegraphics[width=41mm]{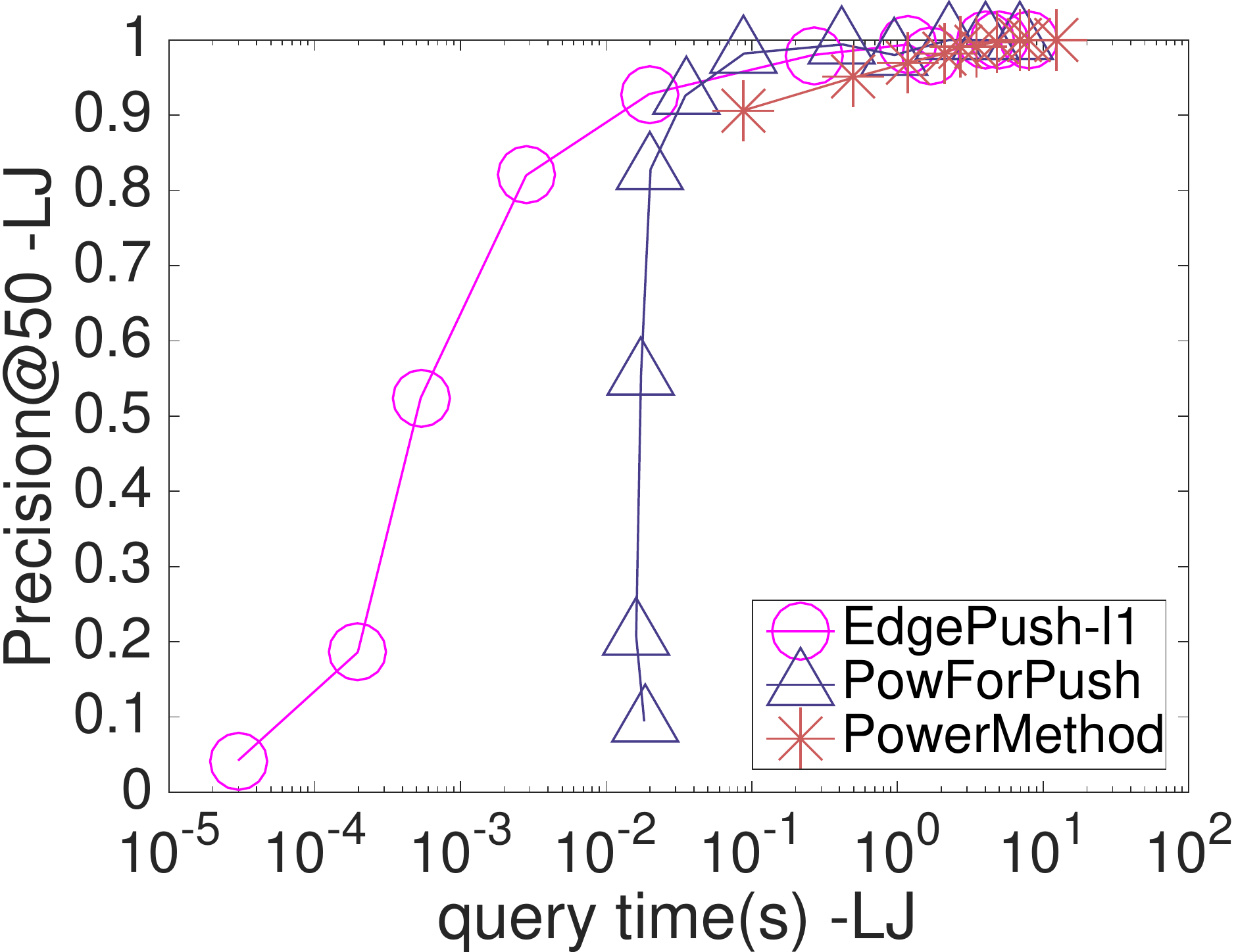} &
			\hspace{-2mm} \includegraphics[width=41mm]{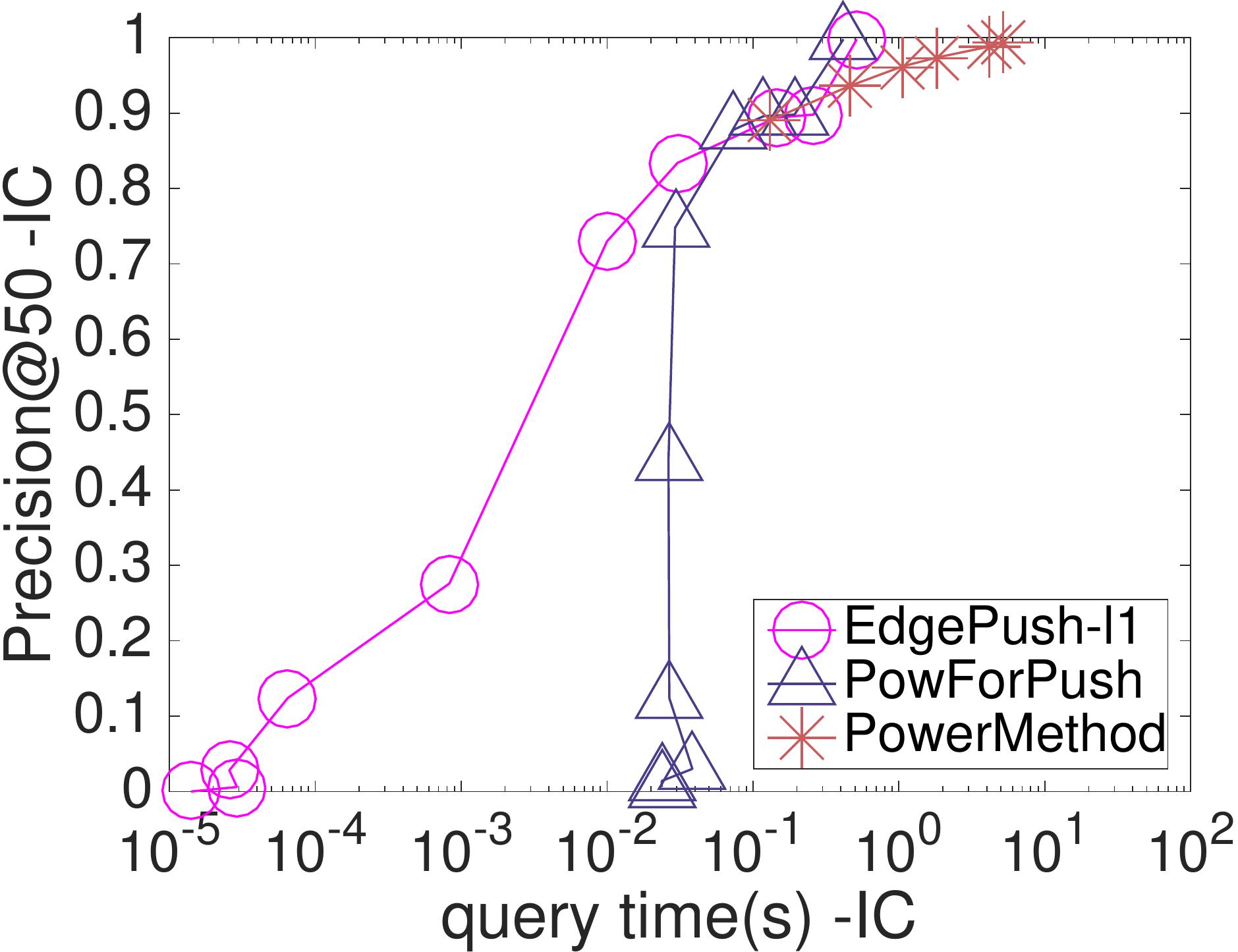} &
			\hspace{-2mm} \includegraphics[width=41mm]{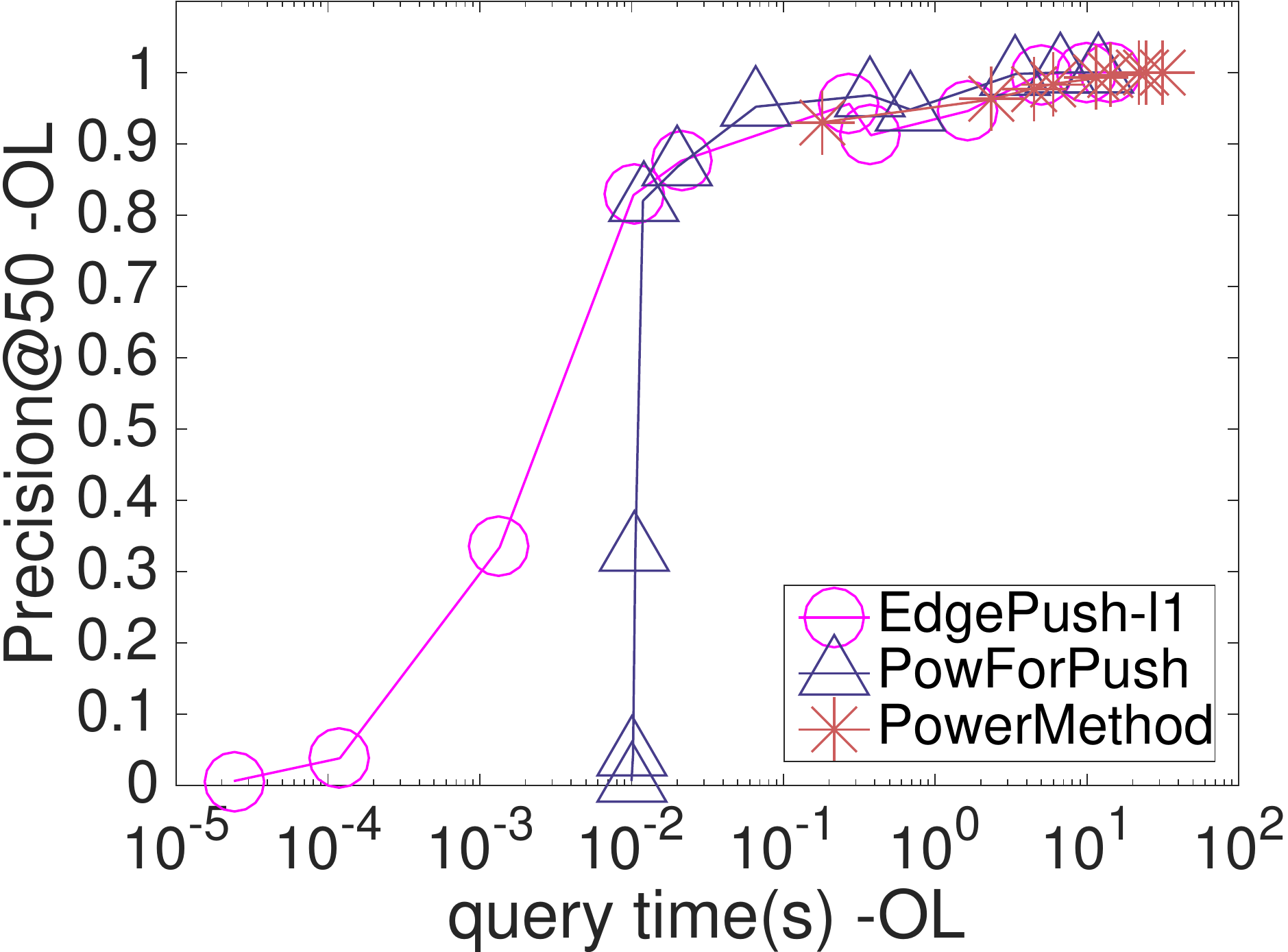} 
		\end{tabular}
		\vspace{-5mm}
		\caption{\hz {\em precision@50} v.s. query time on motif-based weighted graphs.}
		\label{fig:precision-query-real_l1}
		\vspace{-2mm}
	\end{minipage}
\end{figure*}

\begin{figure*}[t]
    \begin{minipage}[t]{1\textwidth}
		\centering
		\vspace{-5mm}
		\begin{tabular}{cccc}
			\hspace{-4mm} \includegraphics[width=43mm]{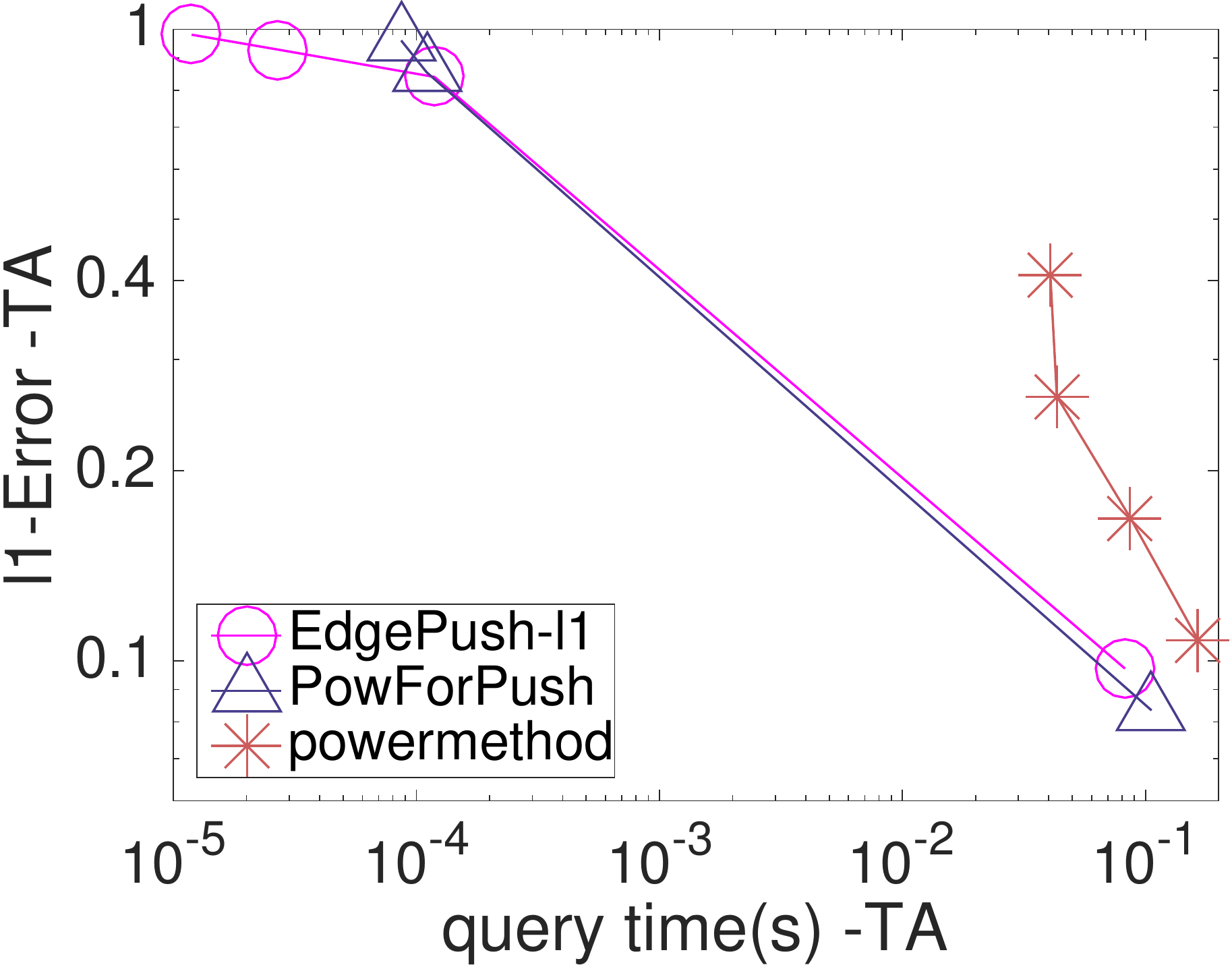} &
    		\hspace{-3mm} \includegraphics[width=43mm]{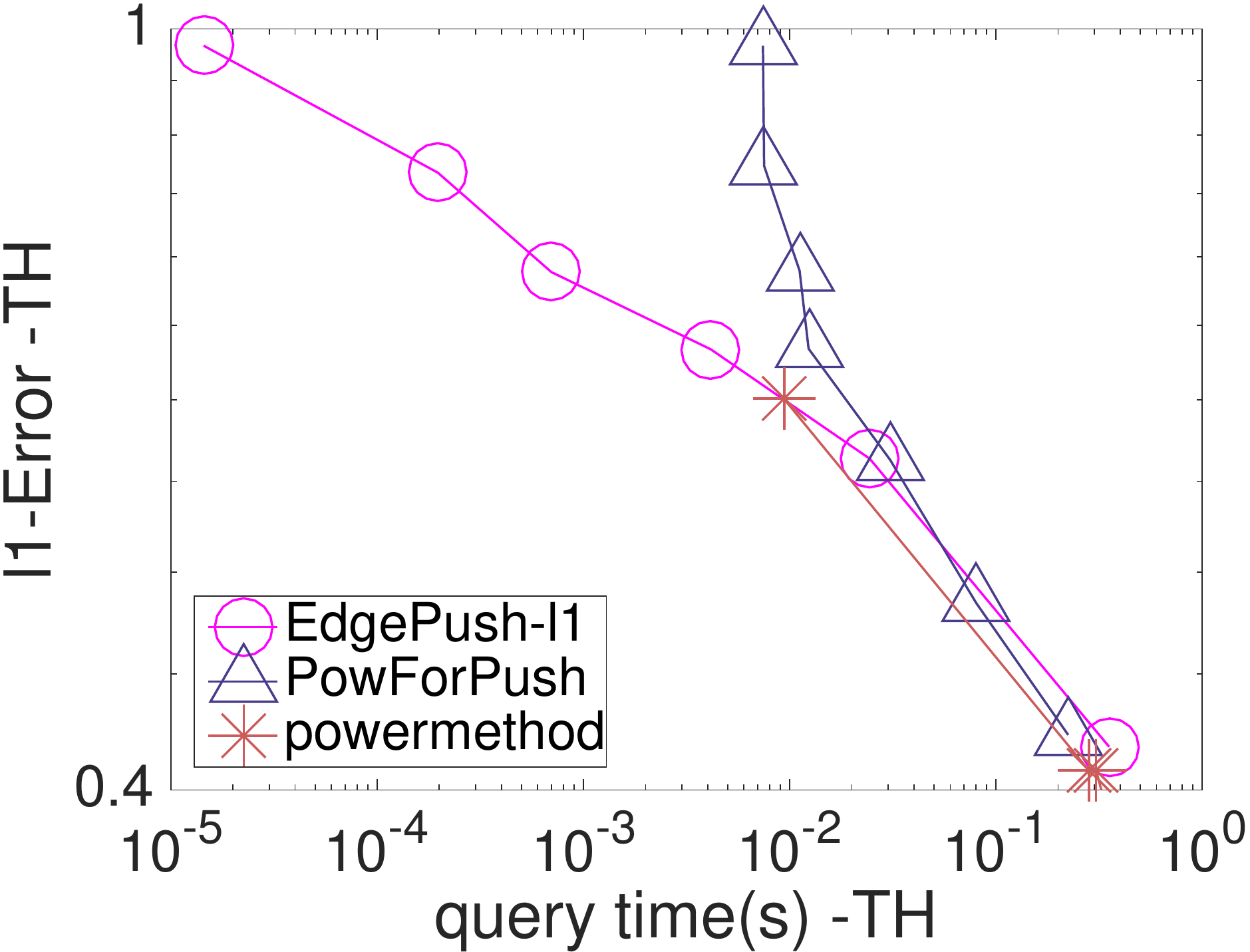} &
			\hspace{-3mm} \includegraphics[width=43mm]{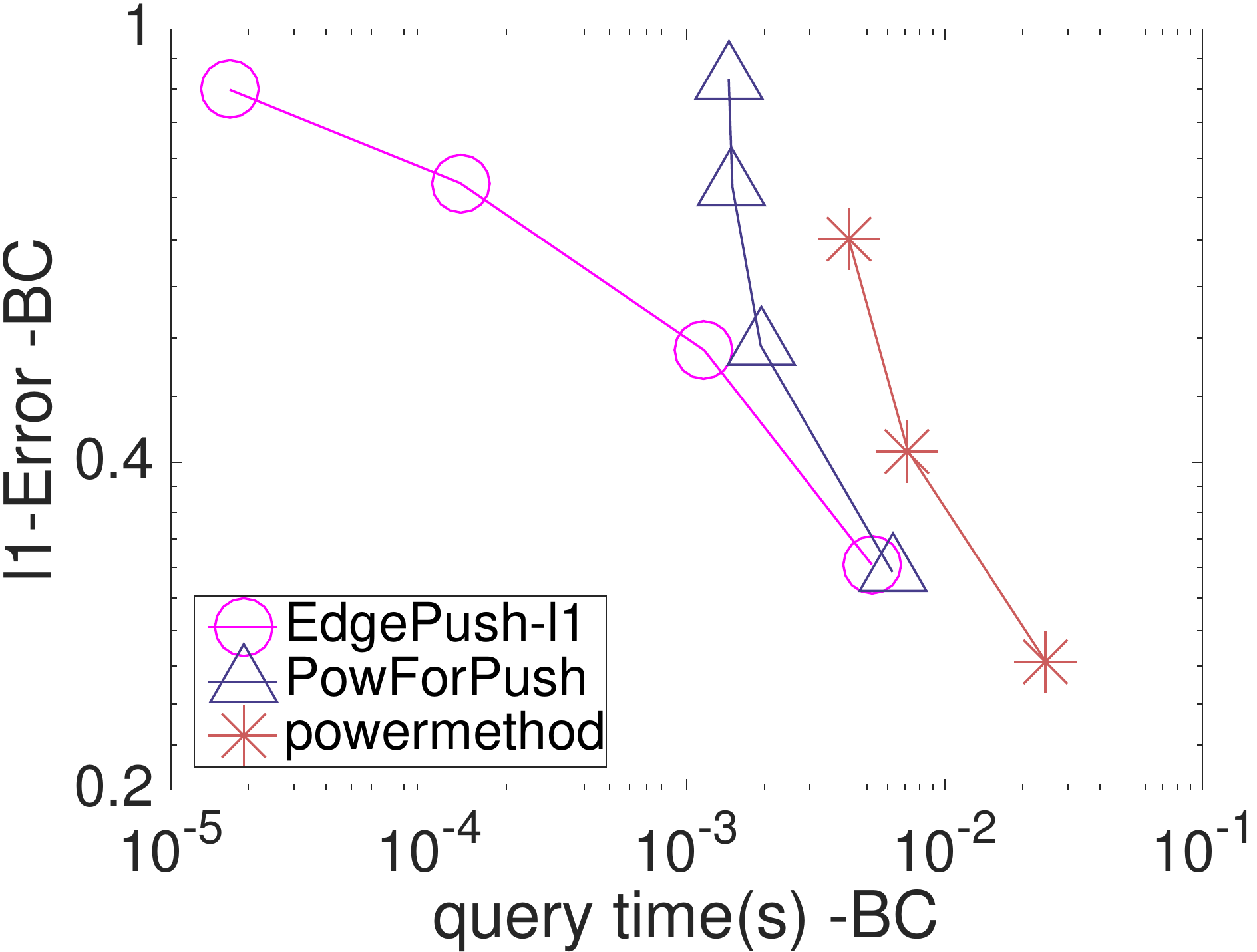} &
			\hspace{-3mm} \includegraphics[width=43mm]{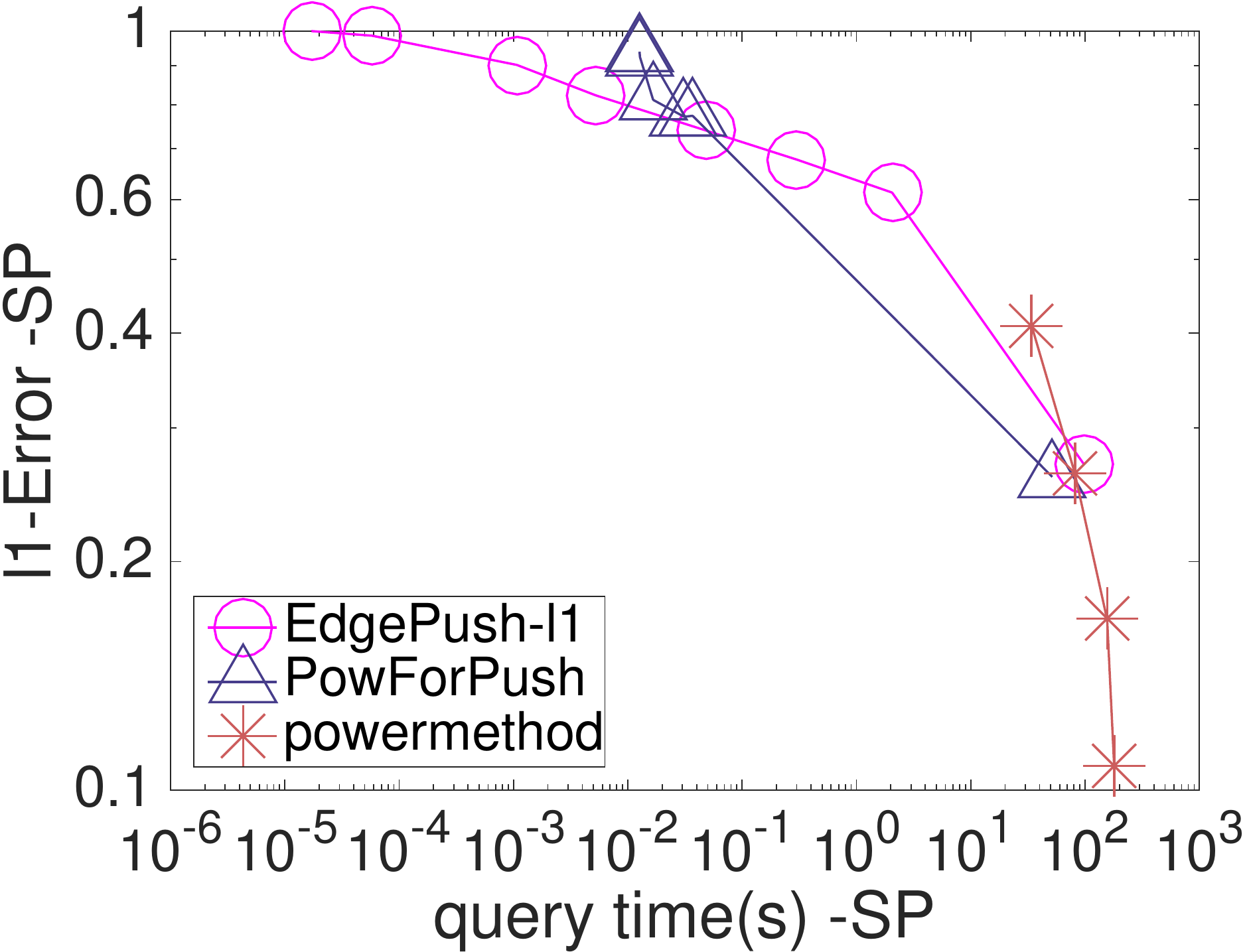}\\
		\end{tabular}
		\vspace{-5mm}
		\caption{{\em $\ell_1$-error} v.s. query time on real weighted graphs.}
		\label{fig:l1error-query-real_l1-real}
    \end{minipage}
    
	\begin{minipage}[t]{1\textwidth}
		\centering
		\vspace{+1mm}
		\begin{tabular}{cccc}
			\hspace{-4mm} \includegraphics[width=43mm]{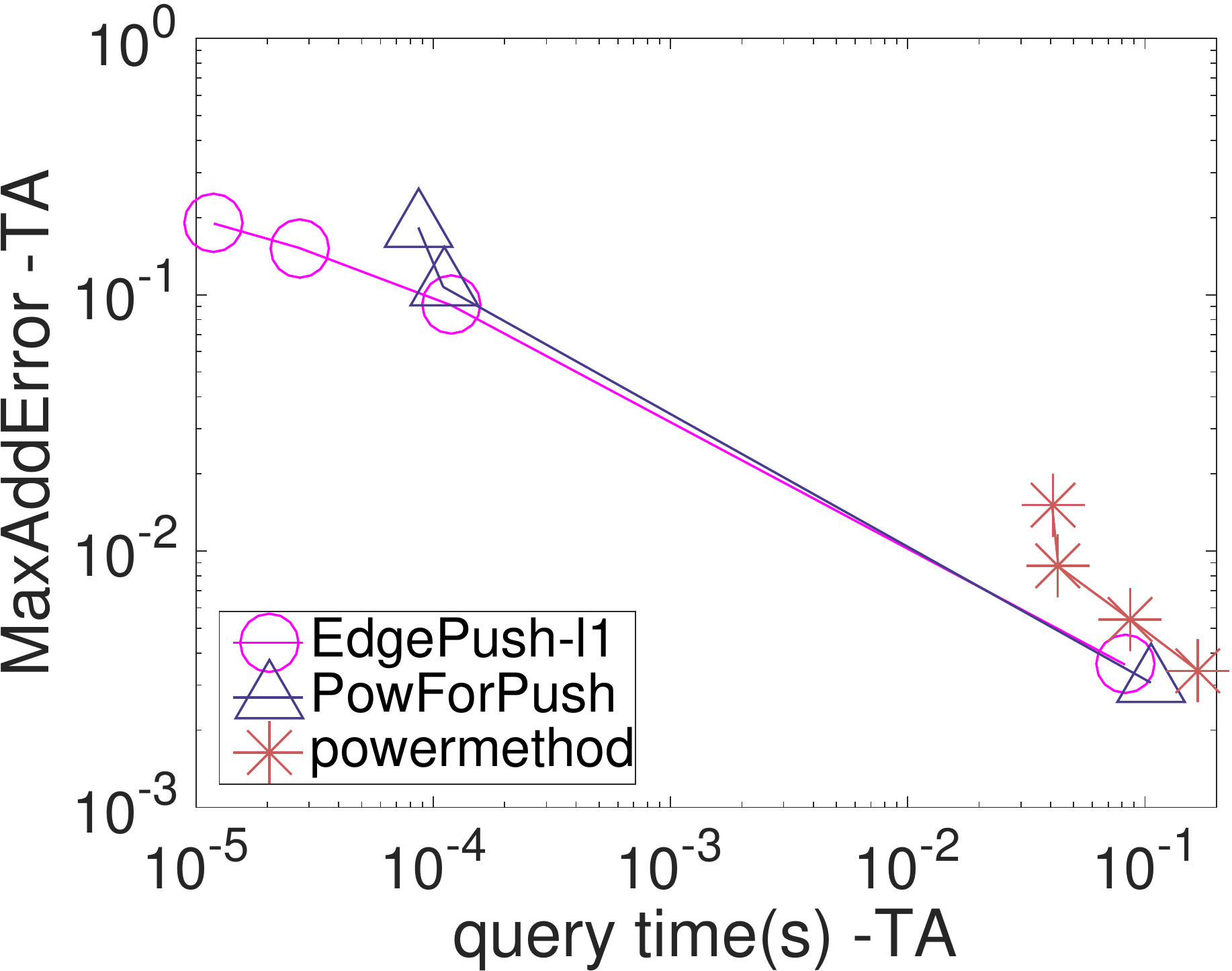} &
			\hspace{-3mm} \includegraphics[width=43mm]{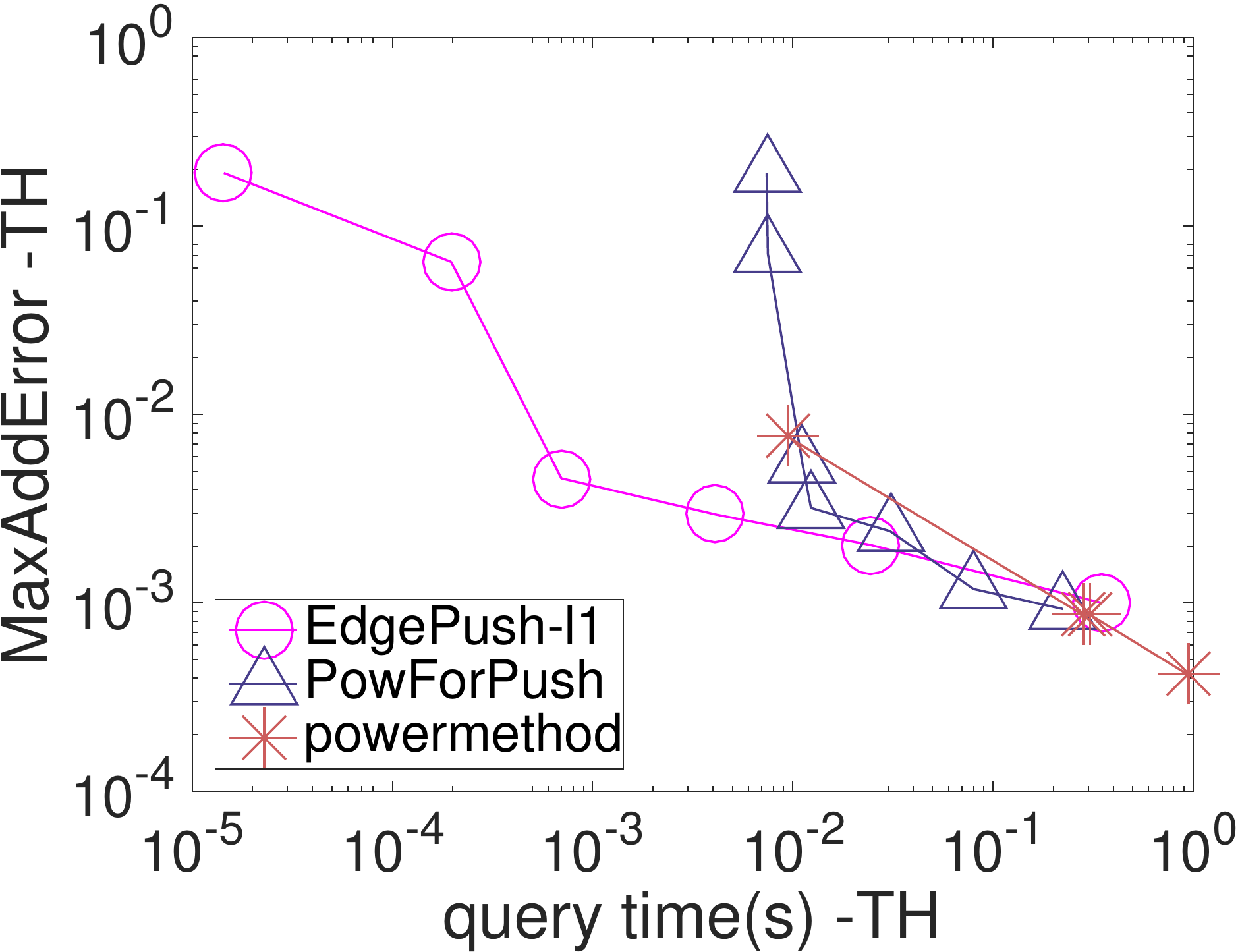} &
			\hspace{-3mm} \includegraphics[width=43mm]{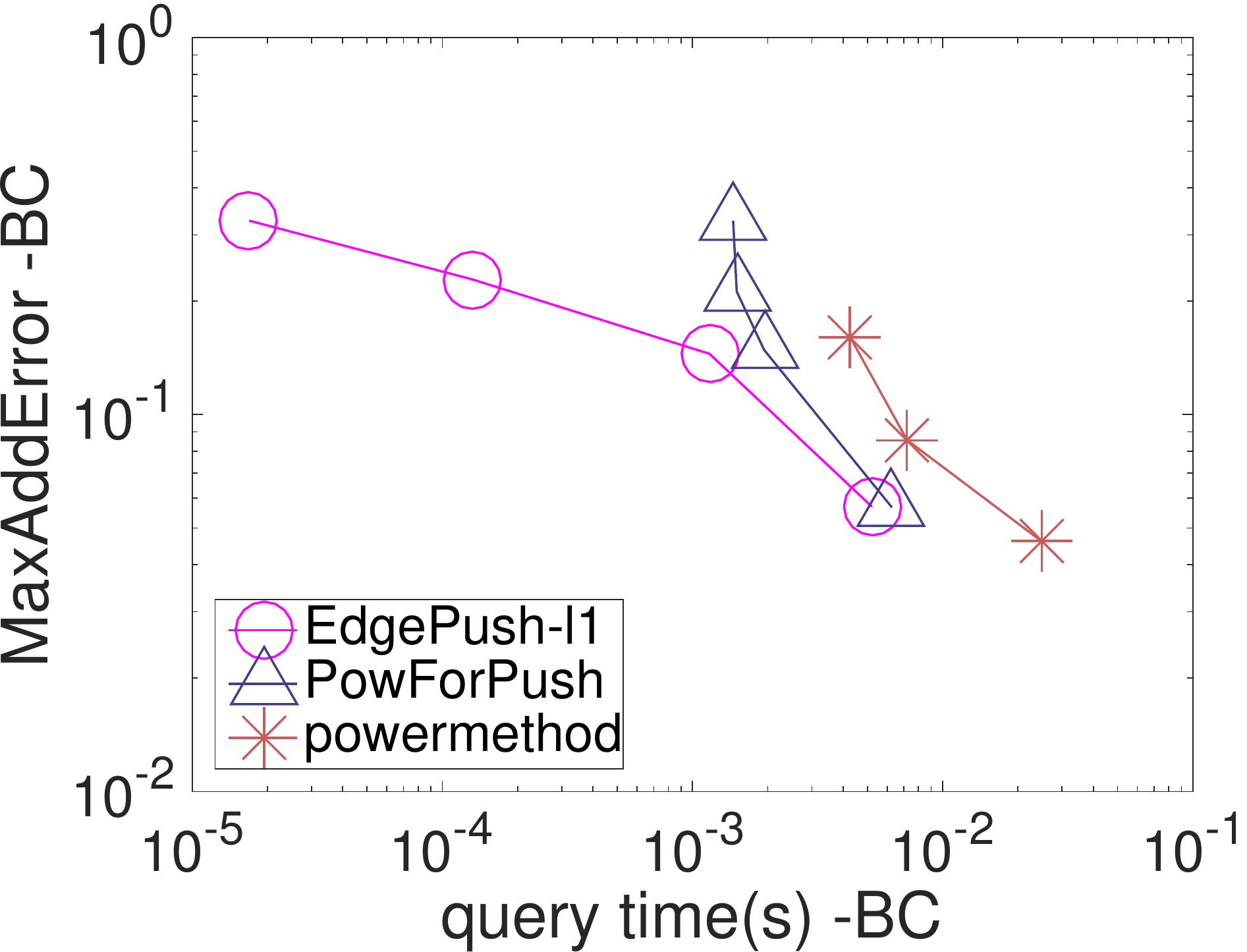} &
			\hspace{-3mm} \includegraphics[width=43mm]{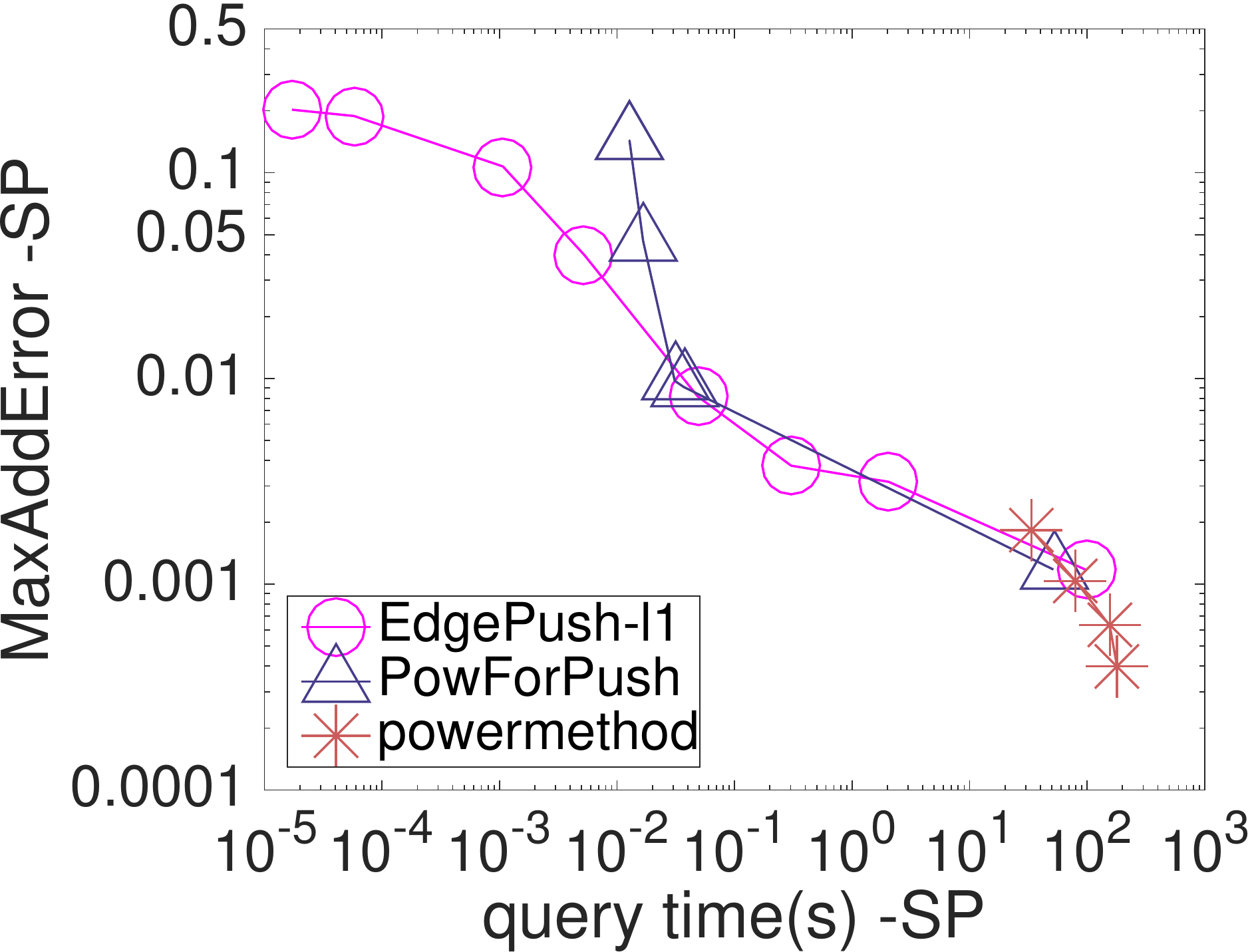} \\
		\end{tabular}
		\vspace{-5mm}
		\caption{{\em MaxAddErr} v.s. query time on real weighted graphs.}
		\label{fig:maxerror-query-real_l1-real}
    \end{minipage}

    \begin{minipage}[t]{1\textwidth}
		\centering
		\vspace{+1mm}
		\begin{tabular}{cccc}
			\hspace{-4mm} \includegraphics[width=43mm]{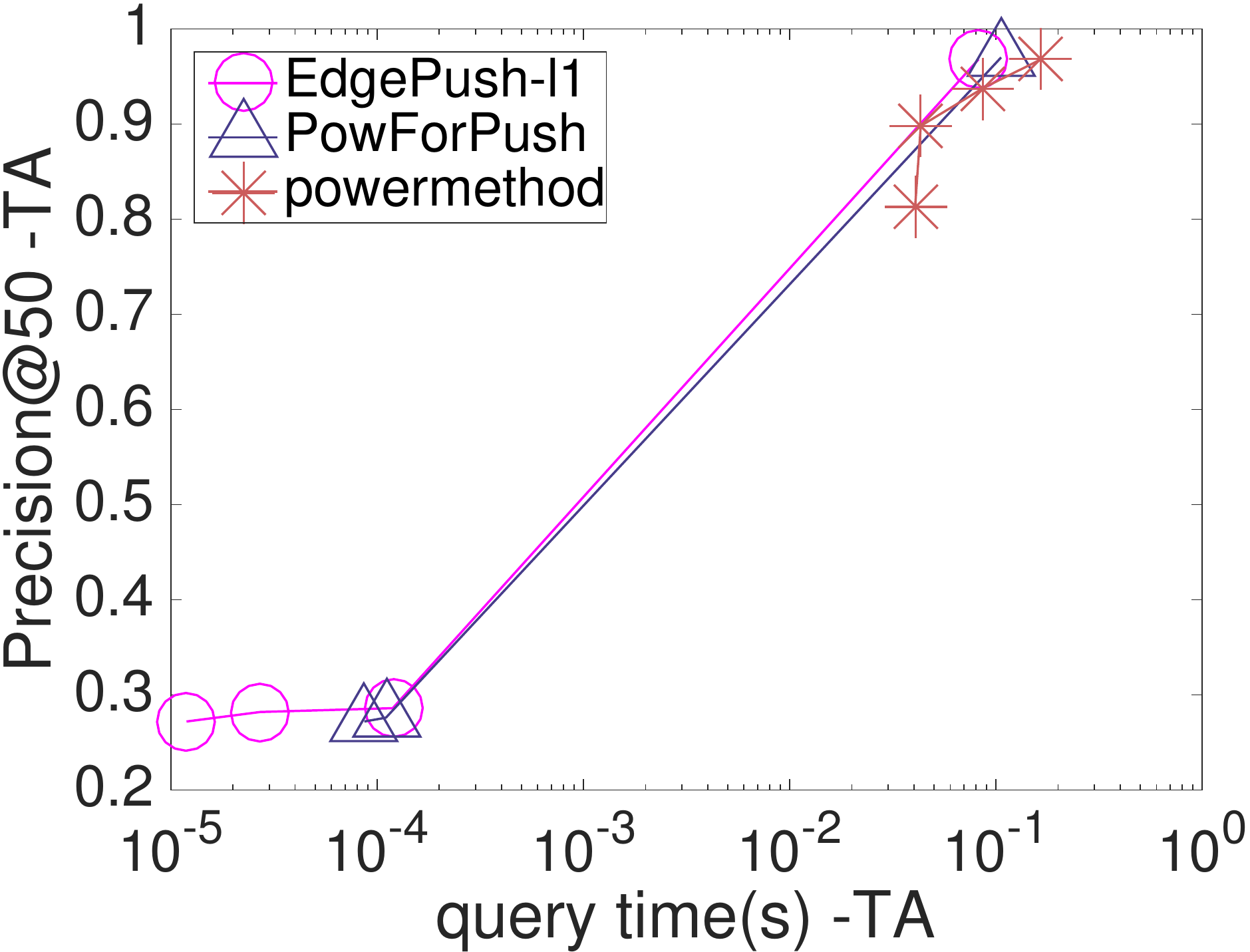} &
			\hspace{-3mm} \includegraphics[width=43mm]{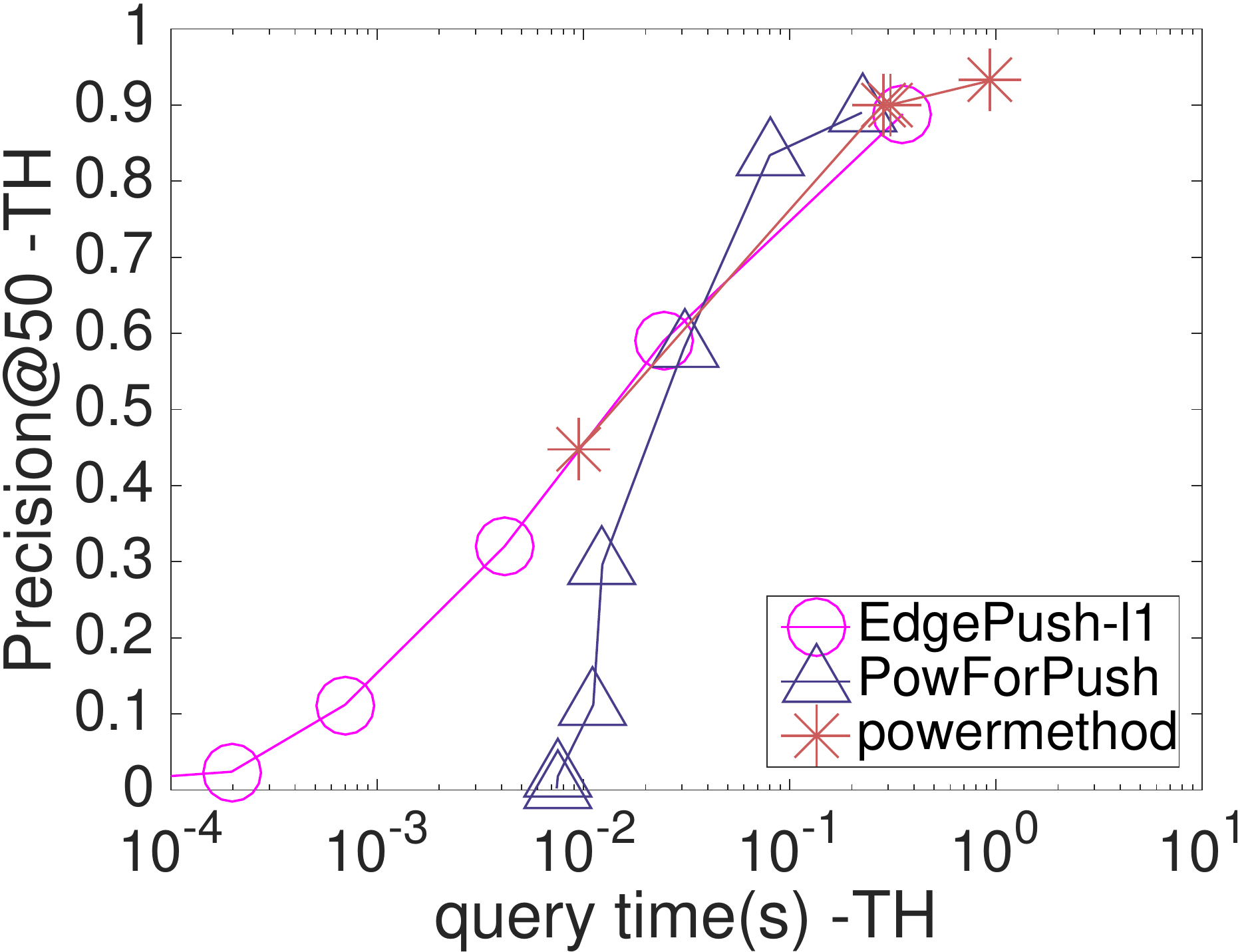} &
			\hspace{-3mm} \includegraphics[width=43mm]{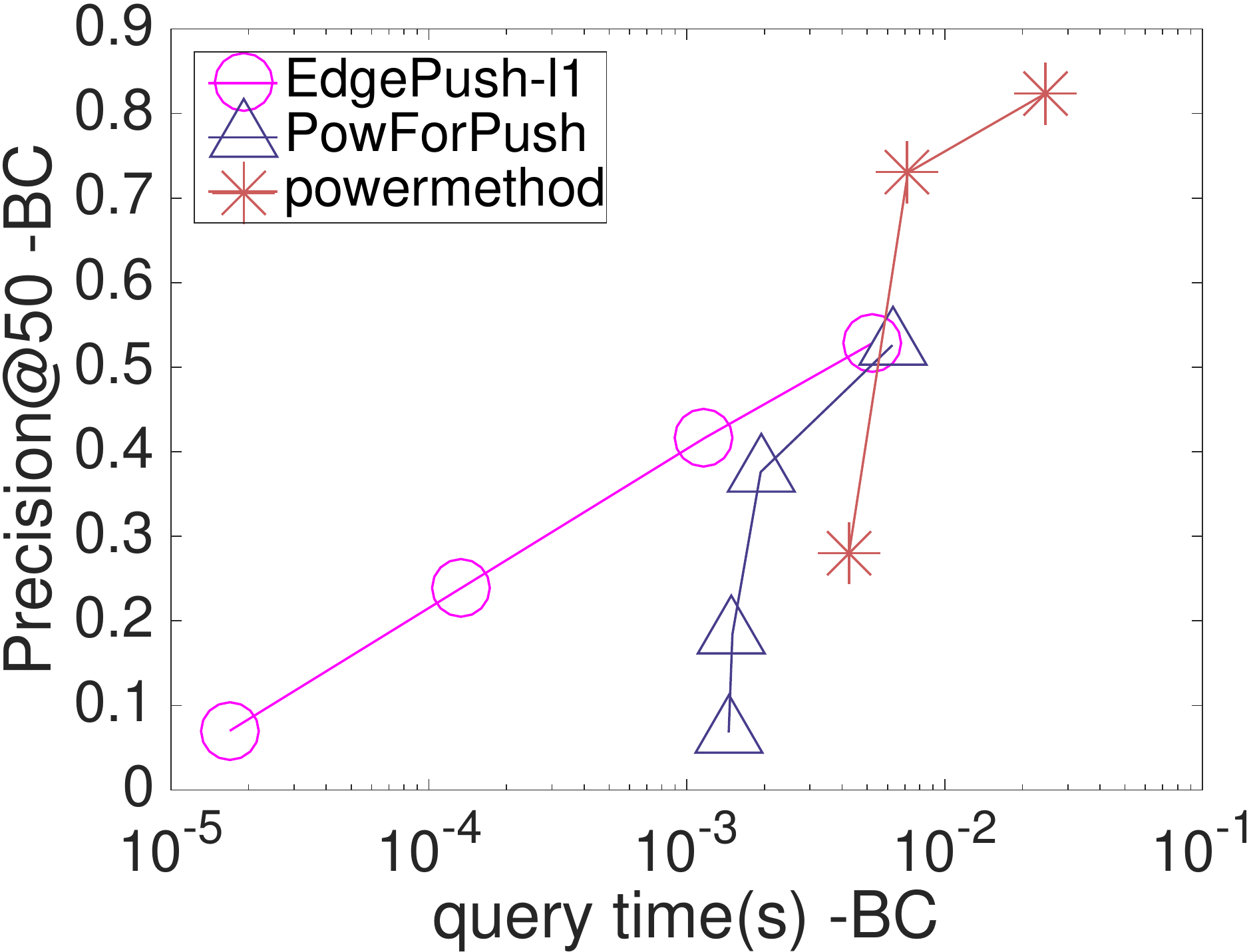} &
			\hspace{-3mm} \includegraphics[width=43mm]{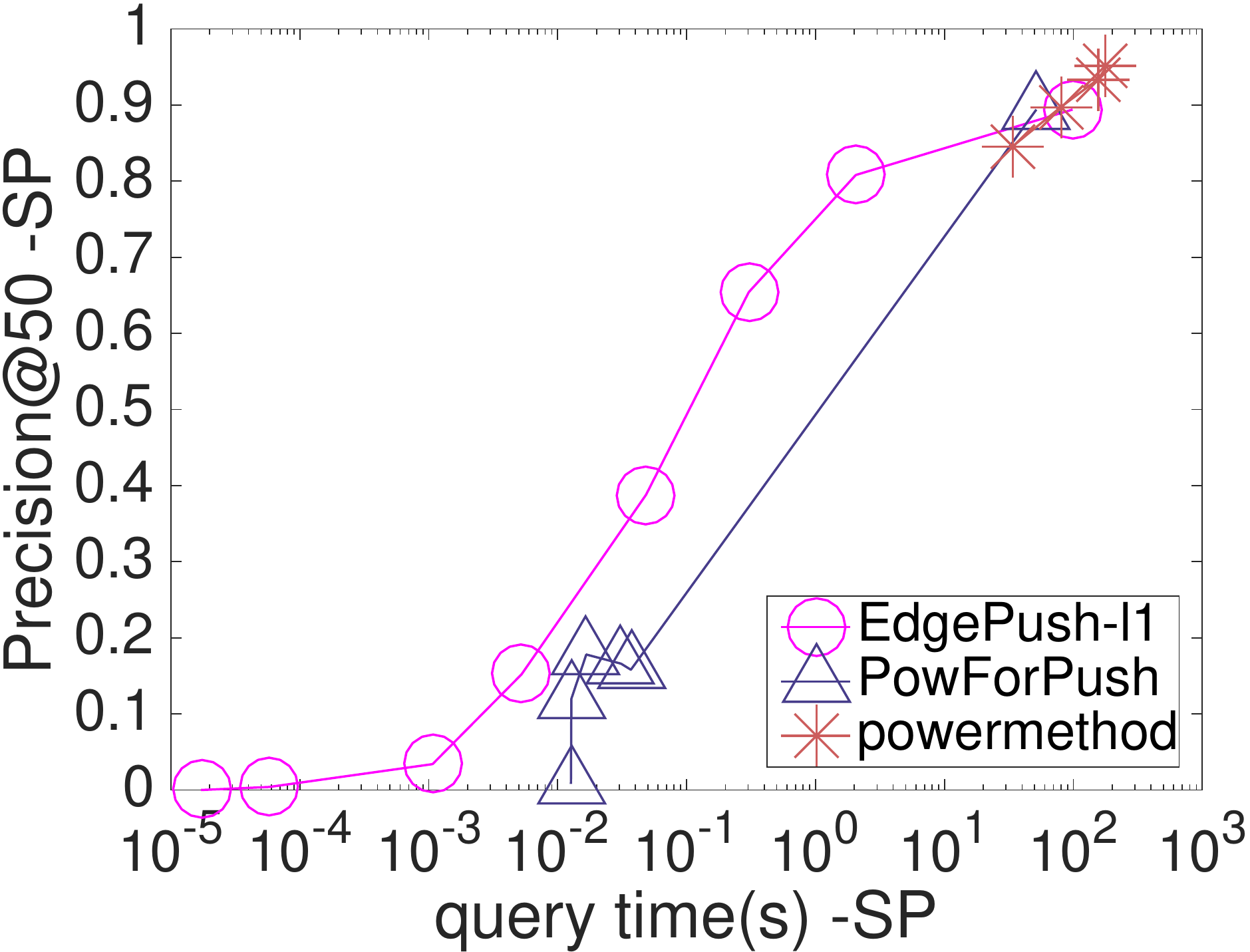}\\ 
		\end{tabular}
		\vspace{-2mm}
		\caption{{\em precision@50} v.s. query time on real weighted graphs.}
		\label{fig:precision-query-real_l1-real}
		\vspace{-2mm}
	\end{minipage}
\end{figure*}

\vspace{-2mm}
\subsection{SSPPR with Normalized Additive Error}\label{subsec:add}
In this subsection, we evaluate the effectiveness of \edgepush with normalized additive error. Furthermore, we apply \edgepush to the local clustering application to achieve better efficiency. 


\header{\bf Evaluation metrics.}
We adopt three metrics for evaluation.
\begin{itemize}[leftmargin = *]
	\item \underline{\nadd}: In the experiments, we calculate the maximum of the normalized additive error for each node to evaluate the approximation quality of the SSPPR queries. More precisely, we define \nadd as $\max_{u\in V} \left|\frac{\vpi(u)}{d(u)}-\frac{\epi(u)}{d(u)}\right|$,  
where $\vpi(u)$ and $\epi(u)$ denote the ground-truth and estimated PPR value of $u$ (w.r.t the source node $s$ by default), respectively. 
    \item \underline{{\em normalized precision@k}}: {\hz We use {\em normalized precision@k} to evaluate the performances for identifying top-$k$ results.} 
    Specifically, for an SSPPR vector $\vpi$, we define $V_k\left(\D^{-1}\vpi\right)$ (resp.  $V_k\left(\D^{-1}\epi\right)$) as the set of the top-$k$ nodes $u$ with the largest  $\frac{\vpi(u)}{d(u)}$ (resp. $\frac{\epi(u)}{d(u)}$) among all nodes in the graph. 
    The {\em normalized precision@$k$} is defined as the fraction of nodes in $V_k\left(\D^{-1}\vpi\right)$ that concurs with $V_k\left(\D^{-1}\epi\right)$. 
	We set $k = 50$ in the experiments. 
	\item \underline{Conductance}: 
	We employ {\em conductance} to measure the quality of the clusters, which is defined in Section~\ref{sec:pre}. 
\end{itemize}
\vspace{-2mm}
\header{\bf Methods.} For the SSPPR queries with normalized additive error, we compare our \edgepush (dubbed as \edgepush-Add) against four competitor methods: (i) MAPPR~\cite{yin2017MAPPR}: a version of the \localpush algorithm on weighted graphs; (ii) Monte-Carlo sampling~\cite{jeh2003scaling, fogaras2005MC}; (iii) FORA~\cite{wang2017fora}: an approximate SSPPR algorithm which combines the strength of \localpush and Monte-Carlo smapling; (iv) SpeedPPR~\cite{wu2021SpeedPPR}: an approximate SSPPR algorithm which combines PowForPush (an optimized version of Power Method) and Monte-Carlo sampling. 

{\hz
According to Algorithm~\ref{alg:APPR}, \localpush only has one parameter: the termination threshold $\theta$. {\hz We vary $\theta$ in $[10^{-3}, 10^{-9}]$ on both motif-based and real-world weighted graphs}. 
For Monte-Carlo sampling, FORA and SpeedPPR, they all have three parameters: the relative error threshold $\delta$, the relative error $\e_r$ and the failure probability $p_f$. Following~\cite{wang2017fora,wu2021SpeedPPR}, we fix $\e_r=0.5$ and $p_f=\frac{1}{n}$, where $n$ is the number of node in the graph. For Monte-Carlo sampling and FORA, we vary $\delta$ in $[10^{-1}, 10^{-5}]$. For SpeedPPR, we vary $\delta$ in $[5\times 10^{-1}, 5\times 10^{-5}]$ on motif-based weighted graphs, and in $[10^{-1}, 10^{-5}]$ on real weighted graphs. For our \edgepush, as shown in Algorithm~\ref{alg:edge-search}, each edge $\la u,v \ra \in \bar{E}$ has an individual termination threshold $\theta(u,v)$. According to Theorem~\ref{thm:edge-efficiency-add}, we set $\theta(u,v)=\frac{r_{\max}\cdot \sqrt{\A_{uv}}}{\sum_{\la x,y \ra \in \bar{E}}\sqrt{\A_{xy}}}$ for normalized additive error $r_{\max}$ and vary $r_{\max}$ from $10^{-3}$ to $10^{-9}$. All of the decay step is $0.1$. Additionally, we set the teleport probability $\alpha$ to $0.2$ in all the experiments. }


\header{\bf Results.} 
{\hz
In Figure~\ref{fig:maxerror-query-real_pro} and Figure~\ref{fig:maxerror-query-real_pro-real}, we draw the trade-off plots between the query time and the normalized maximum additive error (denoted as \nadd) on motif-based weighted graphs and real weighted graphs, respectively. Due to the out-of-memory problem, we omit the experimental results of FORA on SP. We observe that under the same \nadd, \edgepush costs the smallest query time among all these methods on all datasets. {\crc In particular, even on Threads (TH) whose $\cos^2(\p)=0.97$, \edgepush still outperforms all baselines in terms of query efficiency, which demonstrates the effectiveness of \edgepush. } Moreover, in Figure~\ref{fig:precision-query-real_pro} and Figure~\ref{fig:precision-query-real_pro-real}, we show the trade-offs between {\em normalized precision@50} and query time. For the eight datasets, an overall observation is that \edgepush outperforms all competitors by achieving higher precision results with less query time. Most notably, on the Orkut-Links (OL) dataset, \edgepush achieves a {\em normalized precision@50} of 0.8 using a query time of 0.0002 seconds, while the closest competitor, MAPPR, achieves a {\em normalized precision@50} of 0.6 using 0.026 seconds. Additionally, in Figure~\ref{fig:precision-query-real_pro-real}, we observe that compared to the performance of \edgepush on TH, the superiority of \edgepush over \localpush are more clear on TA, BC and SP. This concurs with the analysis that the superiority of \edgepush changes with the unbalancedness of edge weights. 

Furthermore, Figure~\ref{fig:conductance-query-real_pro} and Figure~\ref{fig:conductance-query-real_pro-real} show the trade-offs between {\em conductance} and the query time on motif-based and real weighted graphs. Again, our \edgepush outperforms other competitors by achieving smaller conductance values under the same query time. Additionally, we note that FORA and SpeedPPR gradually outperforms MAPPR in terms of the query efficiency for {\em conductance}. 
However, in the trade-off plots between query time and \nadd or \npre, MAPPR costs less query time compared to FORA or SpeedPPR under the same {\em normalized MaxAddErr} or {\em normalized precision@50}. Recall that FORA and SpeedPPR all combines \localpush with the Monte-Carlo sampling process. This suggests that the Monte-Carlo sampling method is in favor of the conductance criterion, while the \localpush process benefits from the \nadd and \npre. 
} 


\vspace{-6mm}
\subsection{SSPPR with $\ell_1$-Error}\label{subsec:l1}
Next, we demonstrate the effectiveness of \edgepush with $\ell_1$-error. 

\header{\bf Evaluation metrics.} To compare the query efficiency of \edgepush against other competitors, we adopt three metrics, \lerr, \add and \pre, for overall evaluation. 
\begin{itemize}[leftmargin = *]
\vspace{-1mm}
	\item \underline{\lerr}: As a classic evaluation metric, \lerr is defined as: $\|\epi - \vpi\|_1=\sum_{u\in V}|\epi_u-\vpi_u|$, where $\vpi$ and $\epi$ is the ground-truth and estimated SSPPR vectors, respectively. 
	\item \underline{\add}: To evaluate the maximum additive error of each SSPPR approximation, \add is defined as $\max_{u\in V} \left|\vpi_u-\epi_u\right|$.
	\item \underline{\pre}: 
	To evaluate the relative order of the estimated top-$k$ nodes with the highest SSPPR values, \pre is defined as the percentage of the nodes in $V_k(\epi)$ that coincides with the actual top-$k$ results $V_k(\vpi)$. Here $V_k(\vpi)$ and $V_k(\epi)$ denote the top-$k$ node sets for the ground-truth and estimated SSPPR values, respectively. Similarly, we set $k = 50$ in the experiments. 
\end{itemize}

\header{\bf Methods.} 
{\hz In this section, we compare the performance of \edgepush with $\ell_1$-error against two algorithms: Power Method~\cite{page1999pagerank} and PowForPush~\cite{wu2021SpeedPPR}. 
Recall that Power Method computes SSPPR queries by iteratively computing Equation~\eqref{eqn:powerdef}. 
Thus, we vary the number of iterations from $3$ to $15$ with an interval of $2$. }
PowForPush is the state-of-the-art algorithm for high-precision SSPPR queries, 
which gradually switches \localpush to Power Method with decreasing $\ell_1$-error. 
Specifically, in the first phase, PowForPush adopts \localpush to compute the SSPPR queries. 
When the current number of {\em active} nodes is greater than a specified {\em scanThreshold}, 
PowForPush switches to Power Method by performing a sequential scan technique to access active nodes for the \lpush operation. 
A node is called active if its residue is larger than the global termination threshold $\theta$.
The rationale behind this switching mechanism is that sequential scan is often more efficient  if the number of random access is relatively large.
{\hz In our experiments, we vary $\theta$ from $10^{-3}$ to $10^{-12}$ with 0.1 decay step. 
Inspired by PowForPush, we apply the same switching technique to our \edgepush for the fairness of comparison. 
Specifically, when the number of edges in the candidate set $\mathcal{C}$ is significantly great, 
we switch \edgepush to Power Method by performing sequential scanning to access active edges and stop maintaining the two-level structure for each node.
An edge $\la u,v \ra\in \bar{E}$ is active if its edge residue is larger than the termination threshold $\theta(u,v)$.
According to Theorem~\ref{thm:edge-efficiency-add}, by setting $\theta(u,v)=\frac{\e \cdot \sqrt{\A_{uv}}}{\sum_{\la x,y \ra \in \bar{E}}\sqrt{\A_{xy}}}$ for $\forall \la u,v \ra \in \bar{E}$, \edgepush achieves the minimum of the expected overall running time subjected to the $\ell_1$-error constraint. To align with the global termination threshold $\theta$ adopted in PowForPush, we vary $\theta(u,v)$ for $\forall \la u,v \ra \in \bar{E}$ from $\frac{10^{-3}\cdot\|\A\|_1 \cdot \sqrt{\A_{uv}}}{\sum_{\la  x,y\ra \in \bar{E}}\sqrt{\A_{xy}}}$ to $\frac{10^{-11}\cdot\|\A\|_1 \cdot \sqrt{\A_{uv}}}{\sum_{\la  x,y\ra \in \bar{E}}\sqrt{\A_{xy}}}$ with 0.1 decay step. }
{\hz To understand the variation interval of $\theta(u,v)$, note that on unweighted graphs where $\A_{uv}=1$, $\theta(u,v)$ is varied in $[10^{-3},10^{-12}]$, which concurs with the variation of termination threshold $\theta$ in PowForPush. }

\begin{figure*}[t]
	\begin{minipage}[t]{1\textwidth}
		\centering
		\begin{tabular}{cccc}
			\hspace{-6mm}
			\includegraphics[width=43mm]{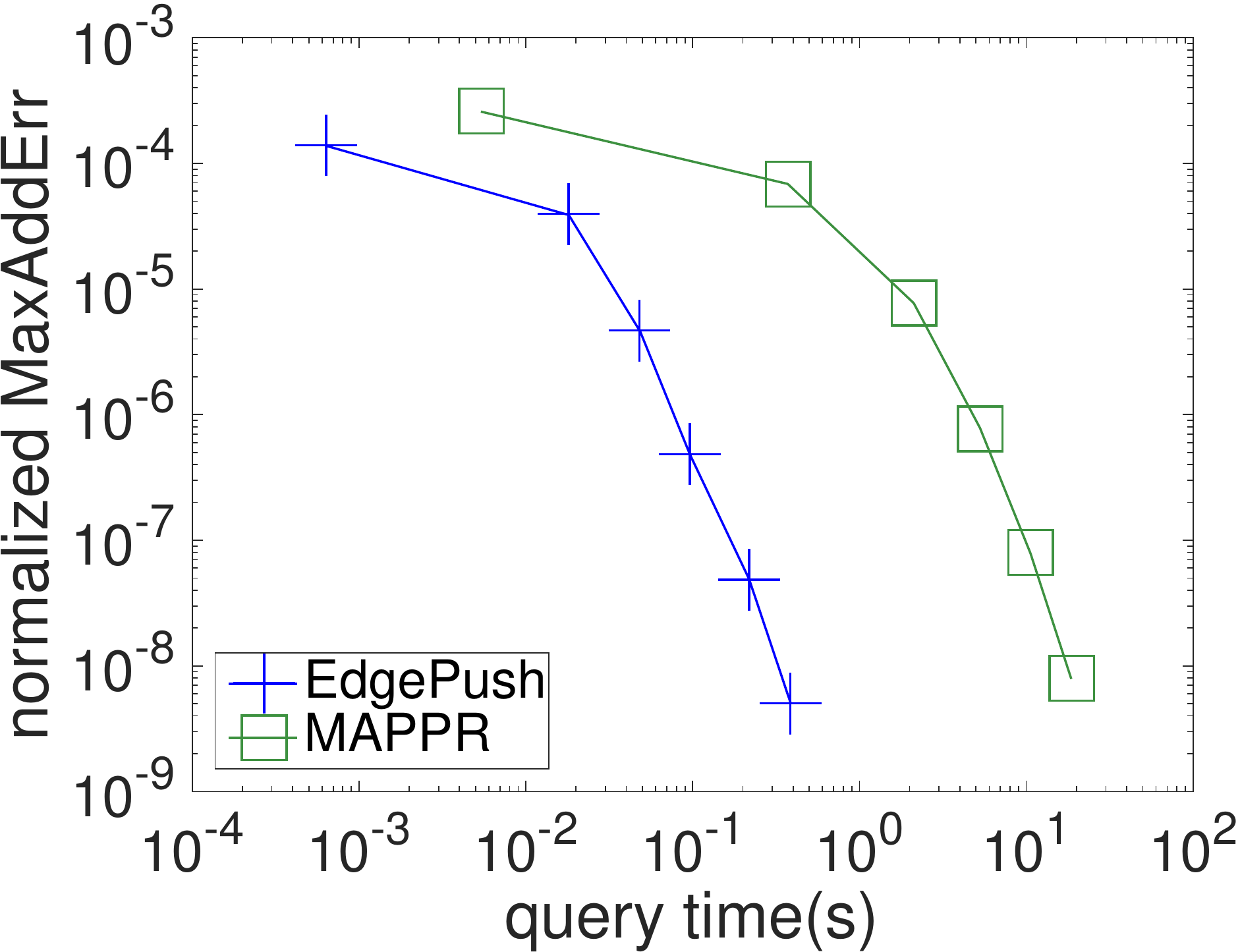} &
			\hspace{-3mm} \includegraphics[width=43mm]{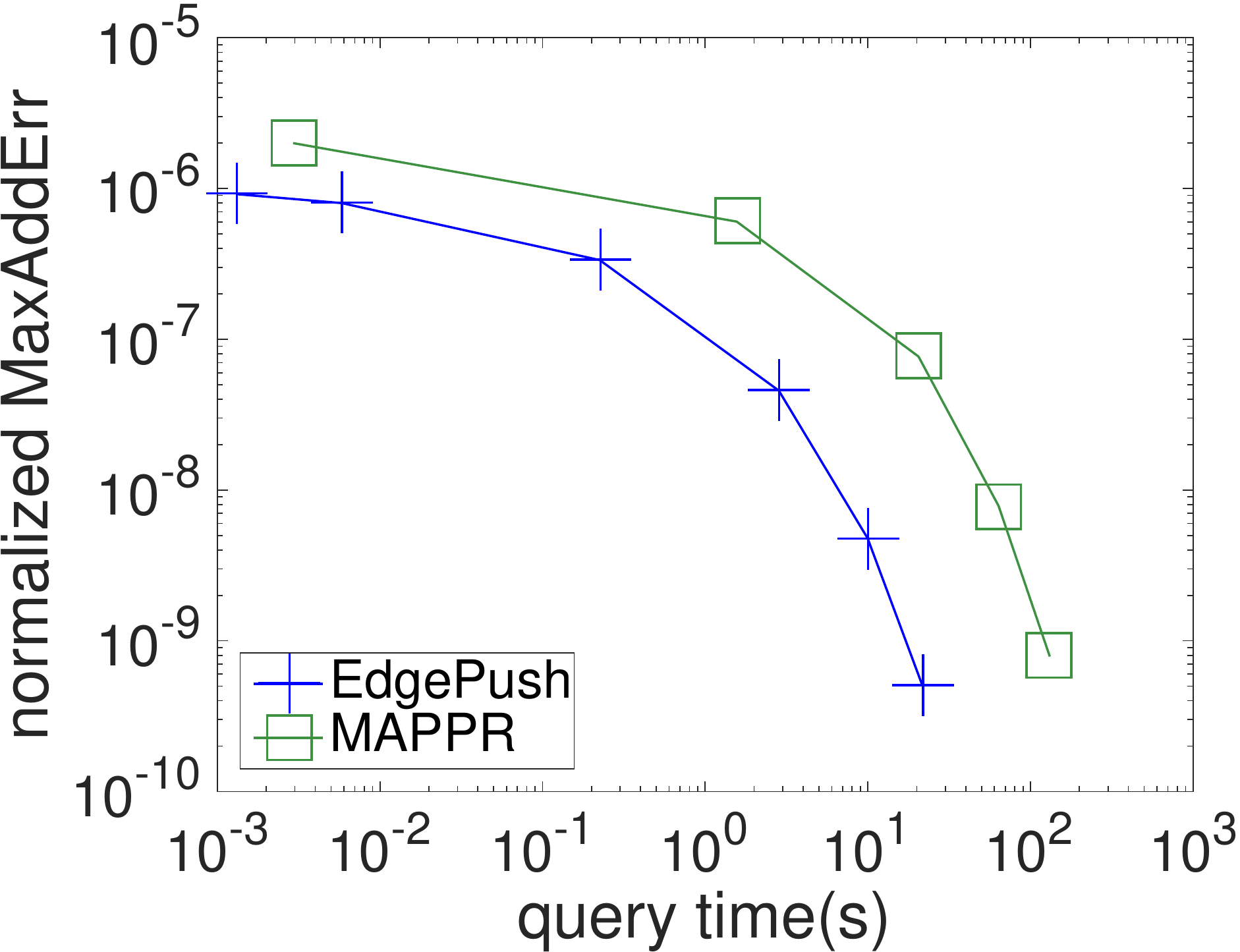} &
			\hspace{-3mm} \includegraphics[width=43mm]{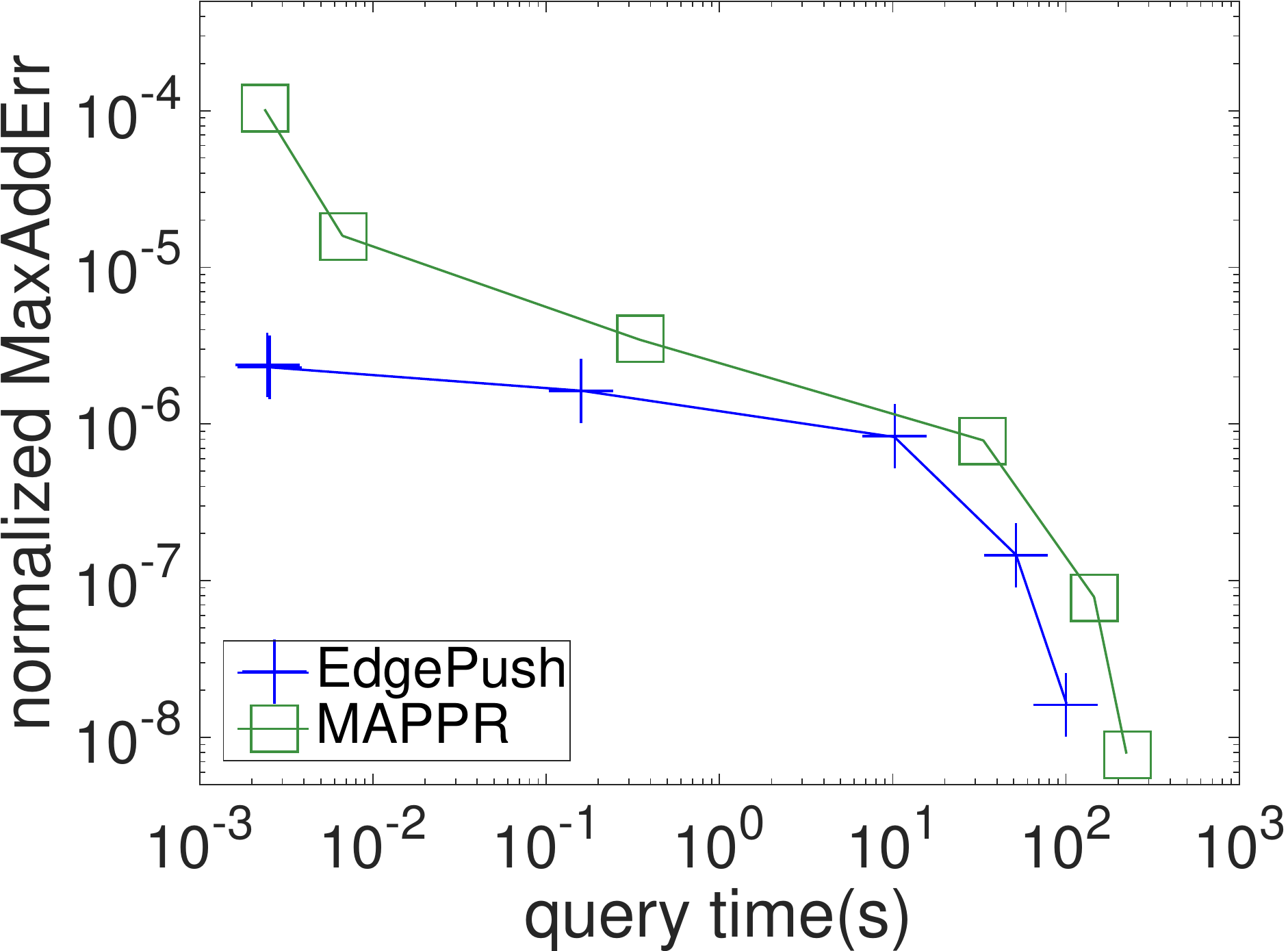} &
			\hspace{-3mm} \includegraphics[width=43mm]{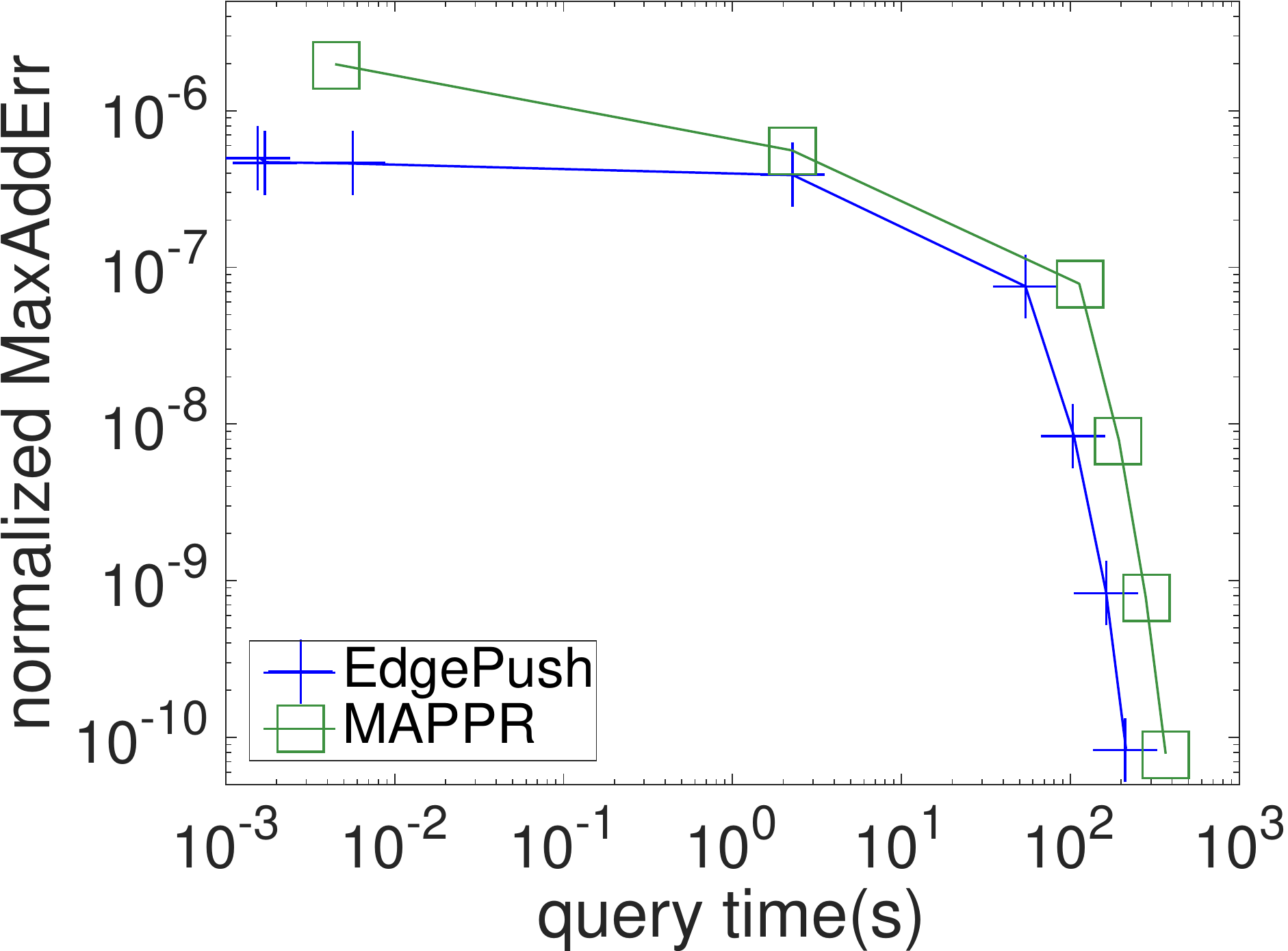} \\
			\hspace{-2mm} (a) {\small $\left(\sum_{v\in V}n_v\hspace{-0.5mm}\cdot \hspace{-0.5mm} \cos^2\p_v\right)/m=0.01$}  &
			\hspace{0mm} (b) {\small $\left(\sum_{v\in V}n_v\hspace{-0.5mm}\cdot \hspace{-0.5mm} \cos^2\p_v\right)/m=0.14$} &
			\hspace{0mm} (c) {\small $\left(\sum_{v\in V}n_v\hspace{-0.5mm}\cdot \hspace{-0.5mm} \cos^2\p_v\right)/m=0.41$} &
			\hspace{0mm} (d) {\small $\left(\sum_{v\in V}n_v\hspace{-0.5mm}\cdot \hspace{-0.5mm} \cos^2\p_v\right)/m=0.7$} \\
		\end{tabular}
		\vspace{-3mm}
		\caption{{\em normalized MaxAddErr} v.s. query time for unbalancedness analysis}
		\label{fig:maxerror-query-add-sensitivity}
		\vspace{-1mm}
	\end{minipage}
	
	\begin{minipage}[t]{1\textwidth}
		\centering
		\vspace{+1mm}
		\begin{tabular}{cccc}
			\hspace{-4mm} \includegraphics[width=43mm]{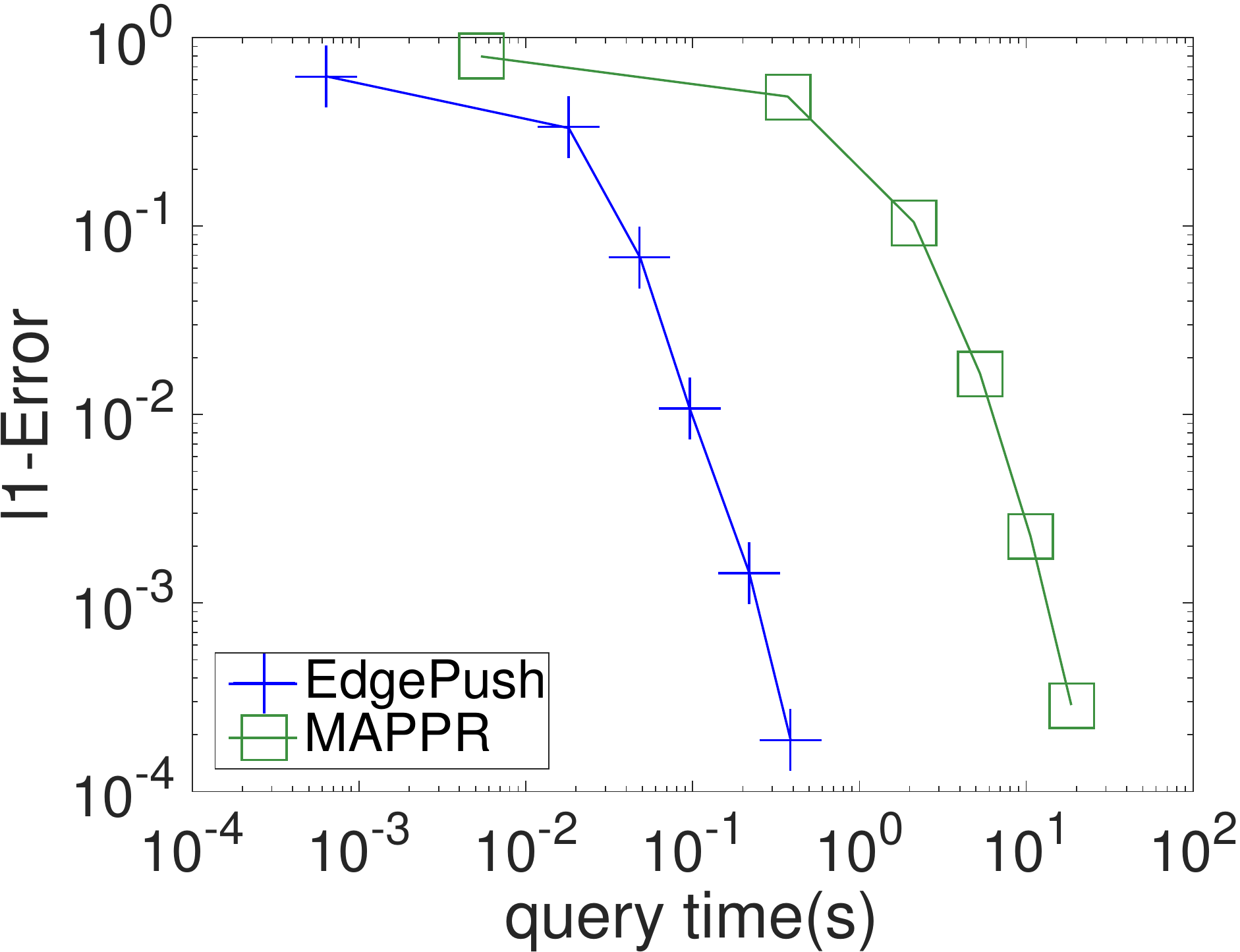} &
			\hspace{-3mm} \includegraphics[width=43mm]{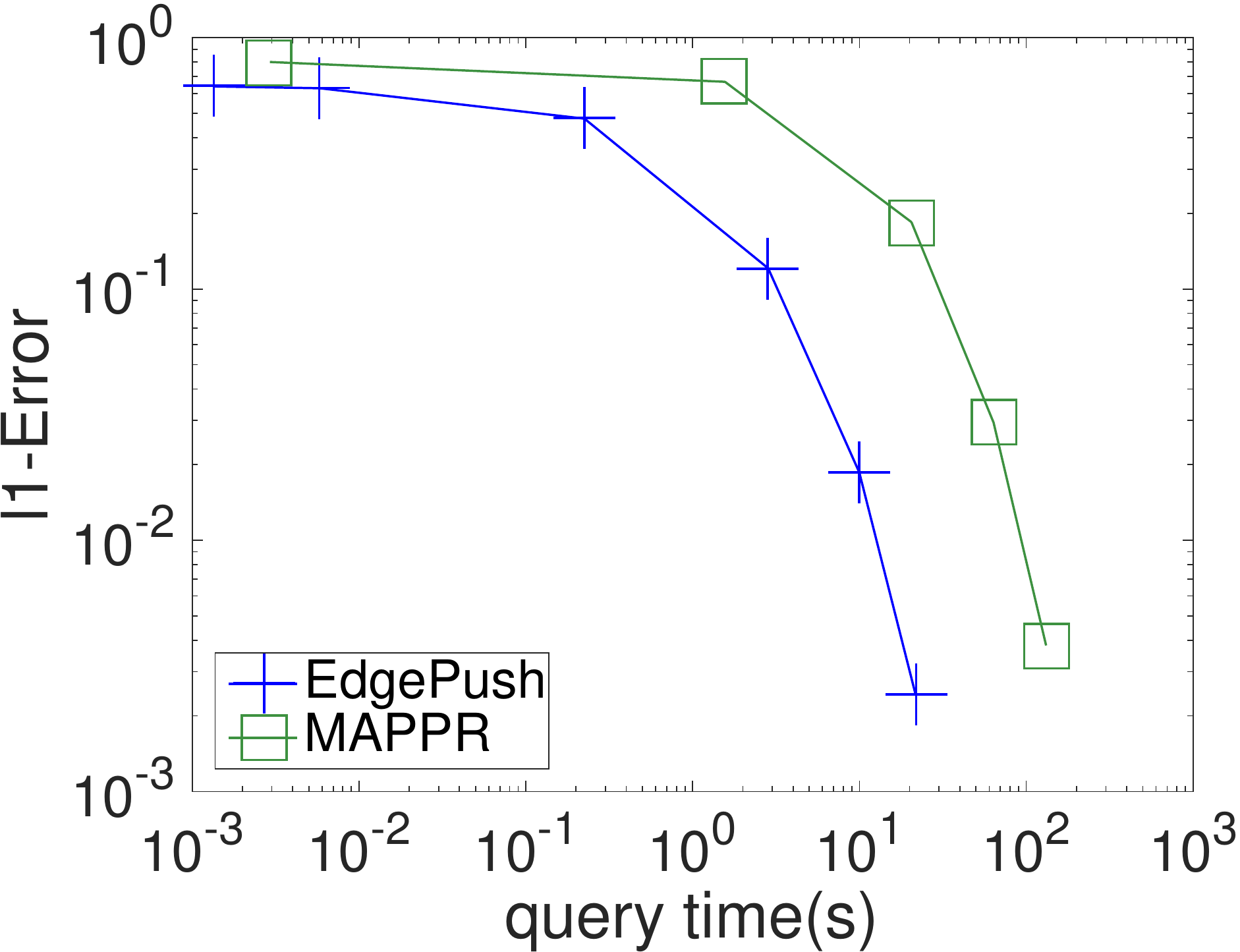} &
			\hspace{-3mm} \includegraphics[width=43mm]{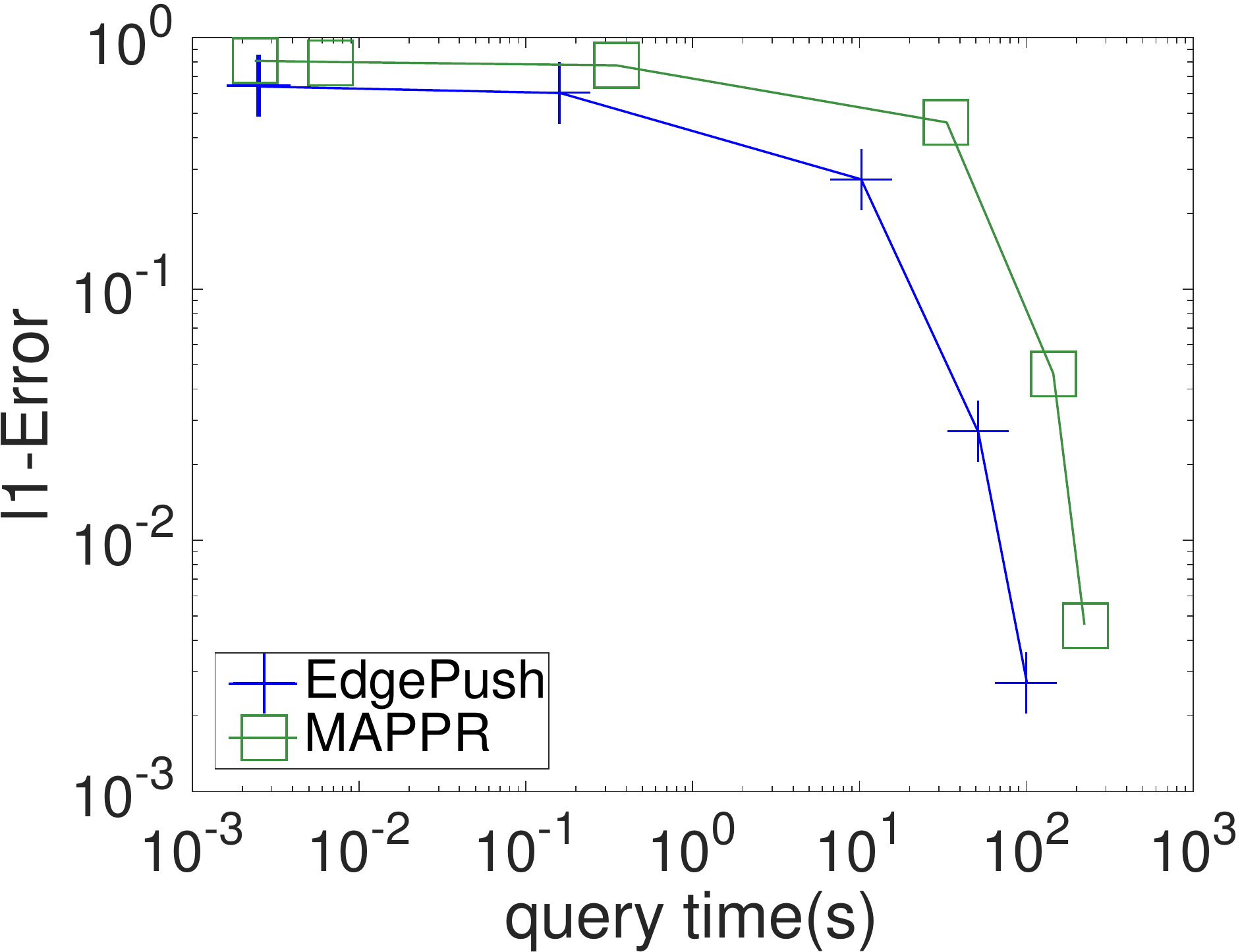} &
			\hspace{-3mm} \includegraphics[width=43mm]{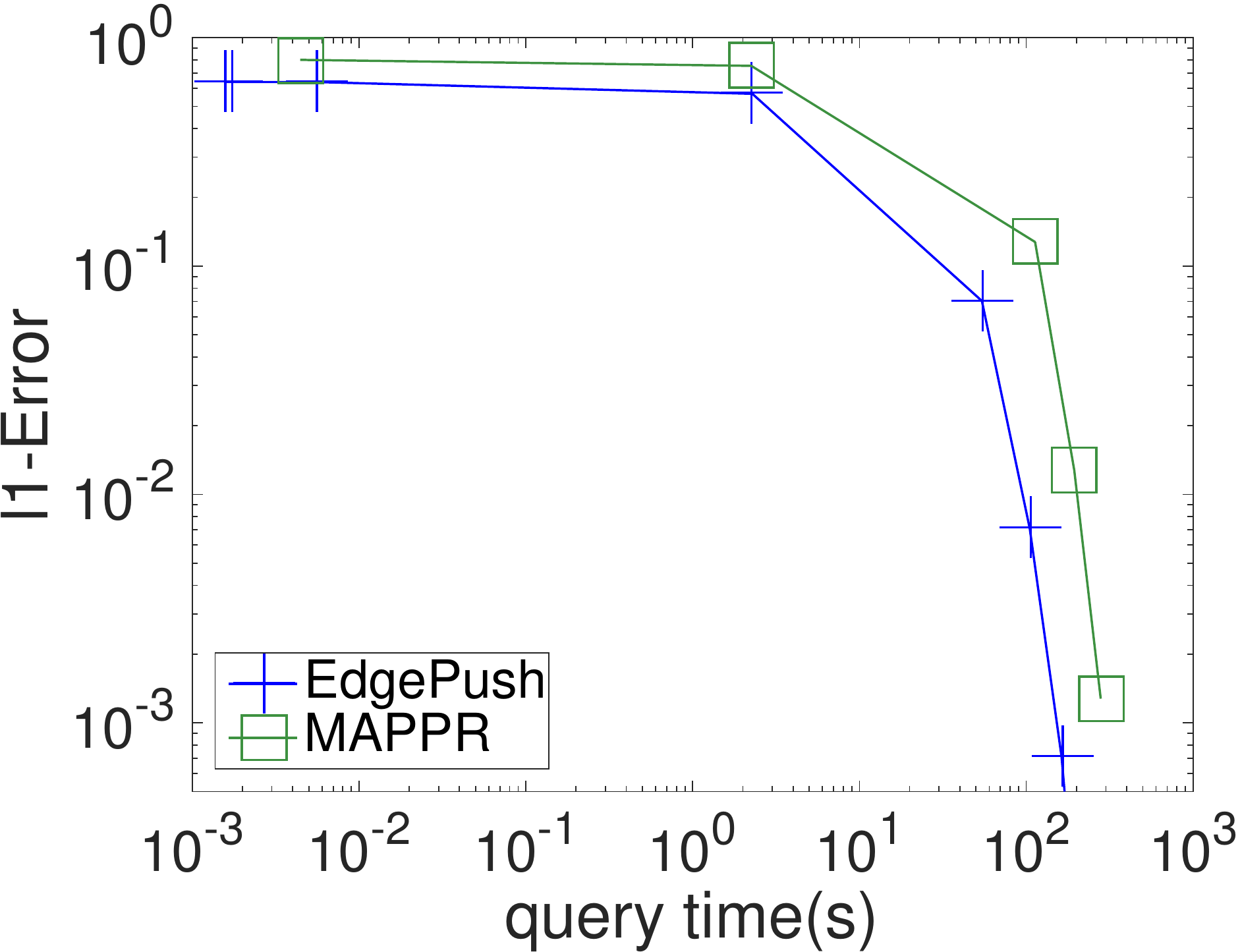} \\
			\hspace{+2mm} (a) { $\cos^2\p=0.01$}  &
			\hspace{+2mm} (b) { $\cos^2\p=0.14$} &
			\hspace{+2mm} (c) { $\cos^2\p=0.38$} &
			\hspace{+2mm} (d) { $\cos^2\p=0.66$} \\
		\end{tabular}
		\vspace{-3mm}
		\caption{{\em $\ell_1$-error} v.s. query time for unbalancedness analysis}
		\label{fig:l1error-query-sensitivity}
		\vspace{-2mm}
	\end{minipage}
\end{figure*}

\header
{\bf Results.} 
{\hz In Figure~\ref{fig:l1error-query-real_l1}, we plot the trade-offs between query time and actual $\ell_1$-error on motif-based weighted graphs. We observe that under relatively large $\ell_1$-error, \edgepush costs the smallest query time on all datasets. 
Additionally, we observe that the plot curves of the three methods gradually overlap with the decreasing of $\ell_1$-error. 
This is because with strict $\ell_1$-error constraints, \edgepush has to touch most of edges in the graph, which have actually become PowForPush. 
On the other hand, recall that the $\ell_1$-error is defined as the sum of additive error over all nodes in the graph (i.e. $\sum_{u\in V}|\epi(u)-\vpi(u)|$). Hence, even with some relatively large $\ell_1$-error, the (normalized) additive error can be small enough for real applications (e.g., local clustering~\cite{FOCS06_FS})
}
Moreover, in Figure~\ref{fig:l1error-query-real_l1-real}, we plot the trade-offs between query time and actual $\ell_1$-error on real-world weighed graphs. Analogous to the performances on motif-based weighted graphs, we observe that \edgepush costs the smallest query time under relatively large $\ell_1$-error. With the decreasing of $\ell_1$-error constraints, \edgepush has to touch most of egdes and the curves of \edgepush, PowForPush and Power Method gradually overlap. 
In Figure~\ref{fig:maxerror-query-real_l1-real} and ~\ref{fig:precision-query-real_l1-real}, we present the trade-off plots between query time and {\em MaxAddErr} and {\em conductance}, respectively. Likewise, we observe that to achieve relatively large $\ell_1$-error, \edgepush outperforms other baselines on all datasets. 

\vspace{-3mm}
\subsection{Unbalancedness Analysis}\label{subsec:sensitivity}
In this section, we evaluate the sensitivity of \edgepush to the unbalancedness of edge weight distribution. Specifically, we generate four fully-connected affinity graphs with $10^5$ nodes. Each node represent a data point in $\k$ dimensional space, where $\k$ is set as $(1, 1, 13, 20)$ to generate the four synthetic graphs. In particular, the coordinate of each node in each dimension is randomly sampled following the normal distribution $N(0,\sigma_N^2)$. And we set $\sigma_N^2$ as $(10^3, 50, 50, 50)$ for the four graphs. Additionally, in the four fully-connected affinity graphs, we assign the weight of edge $\forall \la x_i, x_j\ra \in \bar{E}$ as $\A_{ij}=\exp\left(-\|x_i-x_j\|^2/2\sigma^2\right)$, where $\sigma^2$ is the variance of all data points and $\|x_i-x_j\|$ is the Euclidean distance between nodes $x_i$ and $x_j$. More precisely, we set $\sigma^2=c\cdot d^2\cdot \sigma_N^2$, where $c$ is a tunable constant. We vary $c$ as $(0.1, 1, 1, 1)$ for the four affinity graphs. 

Consequently, as shown in Figure~\ref{fig:maxerror-query-add-sensitivity} and Figure~\ref{fig:l1error-query-sensitivity}, from left to right, the four affinity graphs vary $\cos^2 \p$ in (0.01,0.14,0.38, 0.66), and $\left(\sum_{v\in V}n_v\hspace{-0.5mm}\cdot \hspace{-0.5mm} \cos^2\p_v\right)/m$ in $(0.01, 0.14, 0.41, 0.77)$. Recall that in Lemma~\ref{lem:cos-l1} and ~\ref{lem:cos-add}, we introduce $\cos^2 \p$ and $\left(\sum_{v\in V}n_v\hspace{-0.5mm}\cdot \hspace{-0.5mm} \cos^2\p_v\right)/m$ to quantify the superiority of \edgepush over \localpush with $\ell_1$-error and normalized additive error. The smaller $\cos^2 \p$ or $\cos^2 \p_v$ is, the more unbalanced the graph is. That is, the four affinity graphs used in Figure~\ref{fig:maxerror-query-add-sensitivity} and Figure~\ref{fig:l1error-query-sensitivity}, from left to right, show decreasingly unbalancedness of the edge weight distribution. The rationale here is as the data dimension increases, the distance between pairwise data points becomes more and more similar, resulting in balanced edge weight distribution. 
Accordingly, as seen in Figure~\ref{fig:maxerror-query-add-sensitivity} and Figure~\ref{fig:l1error-query-sensitivity}, we observe that the gaps between the trade-off lines of \edgepush and \localpush gradually reduce from left to right, which is consistent with our analysis.

\vspace{-1mm}
\section{Conclusion} 
\label{sec:conclusion}
In this paper, we propose a novel edge-based local push method \edgepush for approximating the SSPPR vector on weighted graphs. 
\edgepush decomposes the \lpush operation in \localpush into separate \epush operations, each of which can be performed in $O(1)$ amortized time. 
We show that when the source node is randomly chosen according to the node degree distribution, the expected running time complexity of \edgepush is never worse than that of \localpush within certain $\ell_1$-error and normalized additive error. 
In particular, when the graph is dense and the edge weights are unbalanced, \edgepush can achieve a time complexity sub-linear to $m$, and can outperform \localpush by up to a $O(n)$ factor, where $n$ and $m$ are the numbers of nodes and edges in the graph.
{\hz Our experimental results show that when achieving the same approximation error, 
\edgepush outperforms \localpush on large-scale real-world datasets by orders of magnitude in terms of efficiency. }



\section{Acknowledgements} 
\label{sec:ack}
{\crc Zhewei Wei works at Gaoling School of Artificial Intelligence, Beijing Key Laboratory of Big Data Management and Analysis Methods, MOE Key Lab DEKE, Renmin University of China and Peng Cheng Laboratory.} 
This research was supported in part by National Natural Science Foundation of China (No. 61972401, No. 61932001, No. 61832017, No. 62072458 and No. 61932004), by Beijing Outstanding Young Scientist Program No. BJJWZYJH012019100020098, by Alibaba Group through Alibaba Innovative Research Programm, by CCF-Baidu Open Fund (No.2021PP15002000), by China Unicom Innovation Ecological Cooperation Plan and by the Huawei-Renmin University joint program on Information Retrieval. Additionally, Junhao Gan was supported in part by Australian Research Council (ARC) Discovery Early Career Researcher Award (DECRA) DE190101118. Hanzhi Wang was supported by the Outstanding Innovative Talents Cultivation Funded Programs 2020 of Renmin Univertity of China. We also wish to acknowledge the support provided by Intelligent Social Governance Interdisciplinary Platform, Major Innovation \& Planning Interdisciplinary Platform for the ``Double-First Class" Initiative, Public Policy and Decision-making Research Lab, Public Computing Cloud, Renmin University of China.


\bibliographystyle{ACM-Reference-Format}
\bibliography{paper}

\appendix



\vspace{-2mm}
\section{Proofs} 
\subsection{Proof of Lemma~\ref{lem:invariant_localpush}[Invariant by \localpush]}
The Invariant~\eqref{eqn:invariant-nodepush} can be formally proved by induction. According to Algorithm~\ref{alg:APPR}, initially, the node residue $\r$ is set as $\r=\bm{e}_{s}$ that $\r(s)=1$ and $\r(u)=0$ for $\forall u\neq s$. Thus, we have $\vpi(t)=0+1\cdot \vpi_s(t)$. Recall that we denote $\vpi_s(t)=\vpi(t)$ by default. Thus, the invariant holds at the initial stage. 

Now we assume the Invariant~\eqref{eqn:invariant-nodepush} holds before a \lpush operation on node $u$. We will show the correctness of  Invariant~\eqref{eqn:invariant-nodepush} after this \lpush. Let $\epi$ and $\r$ denote the reserve and residue vector before this \lpush operation. After this \lpush, according to the procedure of one \lpush operation (shown in Section~\ref{subsec: localpush_alg}), the reserve of node $u$ increases by $\alpha\cdot \r(u)$. In the meantime, the residue of $u$'s each neighbor increases by $\frac{(1-\alpha)\A_{uv}\r(u)}{d(u)}$, while the residue of node $u$ decreases to $0$. Let $\Delta \text{R.H.S}$ denote the total change on the right side of Invariant~\eqref{eqn:invariant-nodepush}. Accordingly, we have: 
\begin{align*}
\vspace{-2mm}
\Delta \text{R.H.S}=\hspace{-0.5mm}\alpha \r(u)\hspace{-0.5mm}\cdot \hspace{-0.5mm} I\{u\hspace{-0.5mm}=\hspace{-0.5mm}t\}\hspace{-0.5mm}-\hspace{-0.5mm}\r(u)\hspace{-0.5mm}\cdot \hspace{-0.5mm}\vpi_u(t)+\hspace{-1mm}\sum_{v\in N_u}\hspace{-1.5mm}\frac{(1\hspace{-0.5mm}-\hspace{-0.5mm}\alpha)\A_{uv}\r(u)}{d(u)}\cdot \vpi_v(t), 
\vspace{-2mm}
\end{align*}
where $\text{R.H.S}$ denotes the right hand side of Invariant~\eqref{eqn:invariant-nodepush}. $I\{u=t\}$ is an indicator variable that ${I}\{u=t\}=1$ if $u=t$ and $0$ otherwise. By Equation (5) in~\cite{FOCS06_FS}, $\vpi_u(t)$ can be rewritten as below by definition.  
\begin{align}\label{eqn:property-nodepush}
\vspace{-2mm}
\vpi_u(t)=\alpha \cdot I\{u=t\}+\sum_{v\in N_u}\hspace{-1mm}\frac{(1-\alpha)\A_{uv}}{d(u)}\cdot \vpi_v(t). 
\vspace{-2mm}
\end{align}
Plugging into $\Delta \text{R.H.S}$, we can derive $\Delta \text{R.H.S}=\r(u)\cdot 0=0$. Consequently, Invariant~\eqref{eqn:invariant-nodepush} still holds after the \lpush on $u$, which follows the lemma by induction. 

\header{\bf Remark [Intuition on Equation~\eqref{eqn:property-nodepush}]. } To understand Equation~\eqref{eqn:property-nodepush} intuitively, recall that $\vpi_u(t)$ equals to the probability that an $\alpha$-random walk starting from $u$ terminates at $t$. According to the definition of $\alpha$-random walk, at the first step, the walk either stops at $u$, or moves to one of $u$'s neighbor $v$ with $\frac{(1-\alpha)\A_{uv}}{d(u)}$ probability. Thus, the probability that an $\alpha$-random walk from $u$ stops at $t$ can be divided into two parts: i) if $u=t$, terminates at $u$ (i.e. $t$) with $\alpha$ probability at the first step; ii) or walks to one of $u$'s neighbors $v$ with $\frac{(1-\alpha)\A_{uv}}{d(u)}$ probability, then walks to $t$ from $v$ (i.e. $\vpi_v(t)$). 

\vspace{-2mm}
\subsection{Proof of Fact~\ref{thm:bound-local}}
We first show that the $\ell_1$-error of \localpush can be bounded by $\e$. By Invariant~\eqref{eqn:invariant-nodepush}, we have
\begin{align*}
\vspace{-2mm}
    \left|\vpi(t)-\epi(t)\right|=\vpi(t)-\epi(t)=\sum_{u\in V}\r(u)\cdot \vpi_u(t). 
\vspace{-2mm}
\end{align*}
According to Algorithm~\ref{alg:APPR}, after the \localpush process, the residue of each node $u\in V$ satisfies $\r(u)\le d(u)\cdot \theta$, following $\left|\vpi(t)-\epi(t)\right| \le \theta \cdot \sum_{u\in V}d(u)\cdot \vpi_u(t)$. 
To sum up $\vpi(t)-\epi(t)$ for $\forall t\in V$, the $\ell_1$-error can be bounded as
\begin{align*}
    \sum_{t\in V}\left|\vpi(t)-\epi(t)\right|\le \theta \cdot \hspace{-1mm} \sum_{u\in V}\hspace{-1mm} d(u)\cdot \sum_{t\in V}\vpi_u(t)=\theta \cdot \hspace{-1mm} \sum_{u\in V}\hspace{-1mm} d(u)=\theta \cdot \|\A\|_1. 
\end{align*}
In the last equality, we apply the fact $\sum_{t\in V}\vpi_u(t)=1$. By setting $\theta=\frac{\e}{\|\A\|_1}$, we can further bound the $\ell_1$-error as 
\begin{align*}
\vspace{-2mm}
\|\vpi-\epi\|=\sum_{t\in V}\left|\vpi(t)-\epi(t)\right|\le \e, 
\vspace{-2mm}
\end{align*}
following the $\ell_1$-error bound shown in Fact~\ref{thm:bound-local}. 

Next, we present the proof of the expected time cost of \localpush. Before analyzing the cost of \localpush, we first introduce Fact~\ref{fact:pi}: 
\begin{fact}\label{fact:pi}
When the source node $s$ is randomly chosen from the degree distribution of the nodes, i.e., $\bm{e}_s \sim \frac{\D \bm{1}}{\|\A\|_1}$, where $\bm{1}$ is the {\em all-one} vector, the expected SSPPR vector $\mathbb{E}[\vpi] = \frac{\D \bm{1}}{\|\A\|_1}$. Specifically, for each node $u\in V$, the expected $\mathbb{E}[\vpi(u)] = \frac{d(u)}{\|\A\|_1}$.
\end{fact}

\vspace{-1mm}
\begin{proof}
Recall that in Equation~\eqref{eqn:ppr_expansion}, we present a power series expansion to calculate the SSPPR vector $\vpi$. 
That is, $\vpi$ can be expressed as:  
$\vpi=\sum_{\ell=0}^\infty \alpha (1-\alpha)^\ell \P^\ell \bm{e}_s$, where $\P=\A \D^{-1}$ is the transition matrix. We note that matrix $\P$ has an eigenvector $\frac{\D \bm{1}}{\|\A\|_1}$ with $1$ as the corresponding eigenvalue. To see this, note that $\P\cdot \frac{\D \bm{1}}{\|\A\|_1}=\A \D^{-1}\cdot \frac{\D \bm{1}}{\|\A\|_1}= \frac{\A \bm{1}}{\|\A\|_1}=\frac{\D \bm{1}}{\|\A\|_1}$.  
Hence, when the source node $s$ is chosen from the node degree distribution, i.e., $\bm{e}_s \sim \frac{\D \bm{1}}{\|\A\|_1}$,
the expected SSPPR vector $\vpi$ can be rewritten as 
\begin{align}\nonumber
\vspace{-2mm}
	\mathbb{E}[\vpi]=\sum_{\ell=0}^\infty \alpha (1-\alpha)^\ell \P^\ell \cdot \frac{\D \bm{1}}{\|\A\|_1}=\sum_{\ell=0}^\infty \alpha (1-\alpha)^\ell \cdot \frac{\D \bm{1}}{\|\A\|_1}=\frac{\D \bm{1}}{\|\A\|_1}\,. 
\vspace{-2mm}
\end{align}
For each node $u\in V$, we have $\mathbb{E}[\vpi_u] =\frac{d_u}{\|\A\|_1}$, following the fact. 
\end{proof}

Based on Fact~\ref{fact:pi}, we can analyze the expected time cost of \localpush when the source node is chosen according to the degree distribution, shown in Lemma~\ref{lem:cost-localpush}. 

\begin{lemma}\label{lem:cost-localpush}
The overall running time of Localpush is bounded by $O\left(\sum_{u\in V}\frac{n(u)\cdot \vpi(u)}{\alpha  \theta \cdot d(u)}\right)$. In particular, when the source node is randomly chosen according to the degree distribution, the expected time cost of LocalPush is bounded by $O\left(\frac{m}{\alpha \theta \cdot \|\A\|_1}\right)$. 
\end{lemma}

\begin{proof}
According to Algorithm~\ref{alg:APPR}, for each node $u\in V$ that $\r(u)\ge d(u) \cdot \theta$, we perform a \lpush operation on $u$. Specifically, we convert $\alpha \r(u)$ to $\epi(u)$ and distribute the rest $(1-\alpha)\r(u)$ to the residues of $u$'s neighbors. Hence, after one \lpush operation on $u$, the reserve $\epi(u)$ is increased by $\alpha \r(u)\ge \alpha d(u)\cdot \theta$. Thus, when Algorithm~\ref{alg:APPR} terminates, the total number of \lpush operations on node $u$ is upper bounded by $\frac{\epi(u)}{\alpha d(u)\cdot \theta}$. Note that according to Invariant~\eqref{eqn:invariant-nodepush}, the reserve $\epi(u)$ is an underestimate of SSPPR value $\vpi(u)$. Thus, the total number of \lpush operations on node $u$ can be further bounded by $\frac{\vpi(u)}{\alpha d(u) \cdot \theta}$, following the total number of \lpush in the whole \localpush process bounded by $\sum_{u\in V}\frac{\vpi(u)}{\alpha d(u)\cdot \theta}$. On the other hand, for $\forall u\in V$, one \lpush operation on $u$ costs $O(n(u))$ to update the residues of all $u$'s neighbors, and $O(1)$ to increase $\epi(u)$. Consequently, the overall running time of \localpush is bounded by $O \left( \sum_{u\in V}\frac{n(u)\cdot \vpi(u)}{\alpha d(u) \cdot \theta}\right)$. In particular, when the source node $s$ is chosen according to the degree distribution that $\bm{e}_s \sim  \frac{\D \bm{1}}{\|\A\|_1}$, according to Fact~\ref{fact:pi} that $\mathbb{E}\left[\vpi(u)\right]=\frac{d(u)}{\|\A\|_1}$, the expected time cost of Algorithm~\ref{alg:APPR} is bounded by $O\left(\frac{m}{\alpha \theta \cdot \|\A\|_1}\right)$. 
\end{proof}

By Lemma~\ref{lem:cost-localpush}, when $\theta=\frac{\e}{\|\A\|_1}$, the upper bound of the expected time cost of \localpush is $O\left(\frac{m}{\alpha \e}\right)$, which follows this fact.

\subsection{Proof of Fact~\ref{thm:bound-local-addErr}}
Analogous to the proof of Fact~\ref{thm:bound-local}, we can bound the normalized additive error of each $t\in V$ based on Invariant~\eqref{eqn:invariant-nodepush}. More precisely, Invariant~\eqref{eqn:invariant-nodepush} indicates 
\begin{align}\label{eqn:adderr-local-1st}
\left|\vpi(t)-\epi(t)\right|=\vpi(t)-\epi(t)=\sum_{u\in V}\r(u)\cdot \vpi_u(t). 
\end{align}
After the \localpush process, the residue of each node $u$ satisfies $\r(u)\le d(u)\cdot \theta$. Plugging into Equation~\eqref{eqn:adderr-local-1st}, we have 
\begin{align}\label{eqn:adderr-local-2nd}
\vspace{-2mm}
\vpi(t)\hspace{-0.5mm}-\hspace{-0.5mm}\epi(t)\hspace{-1mm}\le \theta\cdot \hspace{-2mm}\sum_{u\in V}\hspace{-1mm}d(u)\cdot \vpi_u(t)=\theta\cdot \hspace{-1mm}\sum_{u\in V}\hspace{-1mm} d(t)\hspace{-0.5mm} \cdot \hspace{-0.5mm} \vpi_t(u)\hspace{-0.5mm}=\hspace{-0.5mm}\theta \hspace{-0.5mm}\cdot \hspace{-0.5mm} d(t). 
\vspace{-2mm}
\end{align}
In the first equality, we apply a property of PPR on undirected graphs that $d(u)\cdot \vpi_u(t)=d(t)\cdot \vpi_t(u)$. In addition, in the last equality, we use the fact that $\sum_{u\in V}\vpi_u(t)=1$. 
Furthermore, based on Equation~\eqref{eqn:adderr-local-2nd}, we have $\frac{\vpi(t)}{d(t)}-\frac{\epi(t)}{d(t)} \le \theta$. By setting $\theta=r_{\max}$, we can further bound the normalized additive error  for $\forall t\in V$ that $\frac{\vpi(t)}{d(t)}-\frac{\epi(t)}{d(t)} \le r_{\max}$. 

On the other hand, according to Lemma~\ref{lem:cost-localpush}, the time cost of \localpush with normalized additive error $r_{\max}$ can be bounded by $O\left(\sum_{u\in V}\frac{n(u)\cdot \vpi(u)}{\alpha \theta \cdot d(u)}\right)$. By setting $\theta=r_{\max}$. In particular, when the source node is randomly chosen according to the degree distribution, the expected time cost of \localpush with normalized additive error is bounded by $O\left(\frac{m}{\alpha r_{\max}\|\A\|_1}\right)$, which follows the fact.

\subsection{Proof of Lemma~\ref{lem:further-inv} [Invariant by \edgepush]}
Analogous to the proof of Lemma~\ref{lem:invariant_localpush}, we prove this invariant by mathematical induction. Initially, according to Algorithm~\ref{alg:edge-search}, the node income vector $\q$ is set as $\q = \bm{e}_s$ and the edge residue matrix $\Q=\bm{0}_{n\times n}$. Thus, according to Equation~\eqref{eqn:edge-relation}, the edge residue $\R_{sx}=\frac{(1-\alpha)\A_{sx}}{d(s)}$, where $\la s,x\ra \in \bar{E}$. On the other hand, for the other edges $\forall \la u,v\ra \in \bar{E}$ that $u\neq s$, the edge residue $\R_{uv}=0$. 
By plugging into the right hand side (dubbed as R.H.S.) of Equation~\eqref{eqn:invariant_edgepush}, we can derive
\vspace{-3mm}
\begin{align}\label{eqn:onestepFP}
\vspace{-2mm}
	\text{R.H.S. } = \alpha\cdot {I}\{t=s\}+\sum_{\la s,x \ra \in \bar{E}}\frac{(1-\alpha)\A_{sx}}{d(s)} \cdot \vpi_x(t), 
\vspace{-3mm}
\end{align}
which always holds according to the Equation (5) in~\cite{FOCS06_FS}. 
The intuition behind is that 
the $\alpha$-random walk from $s$ either stops at node $s$ with probability $\alpha$, or moves to one of its neighbors $x\in N(s)$ with probability $(1-\alpha)$. Up to now, we have proved that Invariant~\eqref{eqn:invariant_edgepush} holds at the initial stage.

Next, we assume the invariant holds before an \epush on edge $\la u, v \ra \in \bar{E}$. We want to show that the invariant still holds after this \epush. Specifically, let $y$ denote the edge residue $\R_{uv}=(1-\alpha)\q(u) \cdot \frac{\A_{uv}}{d(u)}-\Q_{uv}$ before this \epush operation. Because of this \epush, we have $\q(v) \gets \q(v)+y$ and $\Q_{uv}\gets \Q_{uv}+y$. In the meantime, according to Equation~\eqref{eqn:edge-relation}, the changes of edge residues satisfy: $\R_{uv}=\R_{uv}-y$ and $\R_{vw}=\R_{vw}+(1-\alpha)y\cdot \frac{\A_{vw}}{d(v)}$ for each $w\in N(v)$. Let $\Delta \text{R.H.S.}$ denote the total change on the right hand side of Invariant~\eqref{eqn:invariant_edgepush}. We have:  
\vspace{-2mm}
\begin{align}\label{eqn:deltaRHS}
\vspace{-2mm}
\Delta \text{R.H.S.}
\hspace{-0.5mm}=\hspace{-0.5mm} \alpha \hspace{-0.5mm}\cdot \hspace{-0.5mm}{I}\{t = v\}\hspace{-0.5mm}\cdot \hspace{-0.5mm}y \hspace{-0.5mm}- \hspace{-0.5mm}y \cdot \vpi_v(t)\hspace{-0.5mm}+ \hspace{-3mm}\sum_{w\in N(v)}\hspace{-3.5mm}\frac{(1\hspace{-0.5mm}-\hspace{-0.5mm}\alpha)y \hspace{-0.5mm}\cdot\A_{vw}}{d(v)} \hspace{-0.5mm}\cdot \hspace{-0.5mm}\vpi_w(t),
\vspace{-2mm}
\end{align}
\vspace{-1mm}
By Equation~\eqref{eqn:onestepFP} with $v$ being the source node, we have:
\begin{align}\nonumber
\vspace{-2mm}
    \vpi_v(t)=\alpha \cdot {I}\{t=v\}+\sum_{w\in N(v)} \frac{(1-\alpha)\A_{vw}}{d(v)}\cdot \vpi_w(t). 
\vspace{-1mm}
\end{align}
Plugging into Equation~\eqref{eqn:deltaRHS}, we have $\Delta\text{R.H.S.}=y\cdot 0=0$, which follows the invariant by induction. 

\vspace{-1mm}
\subsection{Proof of Lemma~\ref{lem:EdgePushCost}}
According to Algorithm~\ref{alg:edge-search}, for each edge $\la u,v \ra\in \bar{E}$ in the candidate set $\mathbf{C}$, 
we perform a \epush operation on edge $\la u,v\ra$. Specifically, we 
increase the edge expense $\Q_{uv}$ and node income $\q_v$ by at least $\theta(u,v)$. Thus, when Algorithm~\ref{alg:edge-search} terminates, we can use $\frac{\Q_{uv}}{\theta(u,v)}$ to bound the the number of \epush on edge $\la u,v \ra$. Furthermore, by Equation~\eqref{eqn:edge-relation}, $\Q_{uv}$ can be upper bounded as 
$\Q_{uv} \le (1-\alpha)\q(u) \cdot \frac{\A_{uv}}{d(u)} $. 
Note that by Invariant~\eqref{eqn:invariant_edgepush}, $\alpha \q(u)$ is an underestimate of the PPR value $\vpi(u)$ that $\q(u) \le \frac{\vpi(u)}{\alpha}$. Therefore, the number of \epush on edge $\la u,v\ra$ can be further bounded by $(1-\alpha)\cdot \frac{\vpi(u) \cdot \A_{uv}}{\alpha d(u) \theta(u,v)}$, following the upper bound of the total number of \epush in the whole \edgepush as $\sum_{\la u,v \ra\in \bar{E}}\frac{(1-\alpha)\cdot\vpi(u) \cdot \A_{uv}}{\alpha d(u) \theta(u,v)}$. According to Theorem~\ref{thm:cost-per-push}, the time cost of each \epush is bounded by $O(1)$ amortized. Thus, the overall running time of \edgepush is bounded by $O\left(\sum_{\la u,v \ra\in \bar{E}}\frac{(1-\alpha)\vpi(u) \A_{uv}}{\alpha \cdot d(u) \cdot \theta(u,v)}\right)$. When the source node is chosen according to the degree distribution, according to Fact~\ref{fact:pi}, 
the expected PPR value of node $u\in V$ satisfies
$\mathbb{E}[\vpi(u)] =\frac{d(u)}{\|\A\|_1}$. Thus, the expected overall running time of \edgepush can be bounded by $O\left(\sum_{\la u,v\ra \in \bar{E}}\frac{(1-\alpha)\A_{uv}}{\alpha \|\A\|_1 \cdot \theta(u,v)}\right)$, and the lemma follows.  

\subsection{Proof of Lemma~\ref{lem:EdgePushErr}} 
Note that in Algorithm~\ref{alg:edge-search}, we use $\alpha \q$ as the estimator of the SSPPR vector $\vpi$. Thus, the $\ell_1$-error can be computed as $\|\alpha \q -\vpi\|_1=\sum_{t\in V} \left|\alpha \q(t)-\vpi(t)\right|$. According to Invariant~\eqref{eqn:invariant_edgepush}, for each node $t\in V$, we have $\alpha \q(t)-\vpi(t)=\sum_{\la u,v \ra \in \bar{E}} \R_{uv}\cdot \vpi_v(t)$. 
Thus, the $\ell_1$-error can be further expressed as:
\begin{equation}
\begin{aligned}\nonumber
\vspace{-2mm}
    &\|\alpha \q -\vpi\|_1=\sum_{t\in V}\sum_{\la u,v \ra \in \bar{E}} \R_{uv}\cdot \vpi_v(t)=\sum_{\la u,v \ra \in \bar{E}} \R_{uv}\cdot \left(\sum_{t\in V}\vpi_v(t)\right). 
\vspace{-2mm}
\end{aligned}
\end{equation}
Note that $\sum_{t\in V}\vpi_v(t)=1$, 
following $\|\alpha \q -\vpi\|_1=\sum_{\la u,v \ra \in \bar{E}} \R_{uv}$. According to Algorithm~\ref{alg:edge-search}, after the \edgepush process, the edge residue of every edge $\la u,v\ra$ satisfies $\R_{uv}\ge \theta(u,v)$. Thus, we can derive $\|\alpha \q -\vpi\|_1\le \sum_{\la u,v \ra \in \bar{E}} \theta(u,v)$, and the lemma follows. 

\subsection{Proof of Lemma~\ref{lem:EdgePushErr-ad}}
Recall that in the proof of Lemma~\ref{lem:EdgePushErr}, we have $\alpha \q(t)-\vpi(t)=\sum_{\la u,v \ra \in \bar{E}} \R_{uv}\cdot \vpi_v(t)$ for $\forall t\in V$. Also, we note that the edge residue $\R_{uv}$ can be bounded by $\theta(u,v)$ after the \edgepush process, following $\alpha \q(t)-\vpi(t)=\sum_{\la u,v \ra \in \bar{E}} \theta(u,v)\cdot \vpi_v(t)$. Thus, for $\forall t\in V$, the normalized additive error can be bounded as 
\begin{align*}
\frac{1}{d(t)}\cdot \left(\alpha \q(t)-\vpi(t) \right)\le \frac{1}{d(t)}\cdot \sum_{\la u,v \ra \in \bar{E}} \theta(u,v)\cdot \vpi_v(t), 
\end{align*}
which follows the lemma.

\vspace{-1mm}
\subsection{Proof of Theorem~\ref{thm:edge-efficiency-l1}}
First, recall that in Lemma~\ref{lem:EdgePushErr}, we bound the $\ell_1$-error of \edgepush by $\sum_{\la u, v \ra \in \bar{E}} \theta(u,v)$. Thus, the claim that the overall $\ell_1$-error is at most $\e$ follows from the calculations below:
\begin{align*}
	\sum_{\la u, v \ra \in \bar{E}} \theta(u,v) = \sum_{\la u, v \ra \in \bar{E}} \frac{\e \cdot \sqrt{\A_{uv}}}{\sum_{\la x,y \ra \in \bar{E}}\sqrt{\A_{xy}}} = \e\,. 
\end{align*}
Second, Let $Cost$ denote the expected overall running time of \edgepush with specified $\ell_1$-error $\e$. According to Lemma~\ref{lem:EdgePushCost}, we have $Cost=\sum_{\la u,v\ra \in \bar{E}}\hspace{-0.5mm}\frac{(1-\alpha)\A_{uv}}{\alpha\|\A\|_1  \cdot \theta(u,v)}$. By substituting the setting of $\theta(u,v)$'s to the quantity $Cost$, 
we have:
\begin{align*}
	Cost = \hspace{-3mm}\sum_{\la u,v\ra \in \bar{E}}\hspace{-0.5mm}\frac{(1-\alpha)\A_{uv}}{\alpha\|\A\|_1 \hspace{-0.5mm} \cdot \hspace{-0.5mm} \theta(u,v)}\hspace{-0.5mm} = \hspace{-0.5mm} \frac{(1-\alpha)}{\alpha \e \|\A\|_1} \cdot \left(\sum_{\la u,v \ra\in \bar{E}}\hspace{-3mm}\sqrt{\A_{uv}}\right) \cdot \left(\hspace{-0.5mm}\sum_{\la x, y \ra \in \bar{E}}\hspace{-3mm} \sqrt{\A_{xy}} \right)\,,
\end{align*}
which can be further rewritten as $\frac{(1-\alpha)}{\alpha \e \|\A\|_1}\cdot \left(\sum_{\la u,v \ra\in \bar{E}}\sqrt{\A_{uv}}\right)^2$. 
Moreover, we note that, with this setting of $\theta(u,v)$ as suggested by theorem~\ref{thm:edge-efficiency-l1}, $Cost$ is indeed minimized.
This is because, by Cauchy-Schwarz Inequality (shown in Fact~\ref{fact:cauchy}, it can be verified that
$\e \cdot Cost = \left(\sum_{\la u,v \ra \in \bar{E}} \theta(u,v) \right) \cdot \left(\sum_{\la u,v\ra \in \bar{E}}\frac{(1-\alpha)\A_{uv}}{\alpha\cdot \|\A\|_1 \cdot \theta(u,v)} \right) \hspace{-1mm}\geq \frac{(1-\alpha)}{\alpha \|\A\|_1}\cdot \left(\sum_{\la u,v \ra\in \bar{E}}\sqrt{\A_{uv}}\right)^2$.

\subsection{Proof of Theorem~\ref{thm:edge-efficiency-add}}
Before proving Theorem~\ref{thm:edge-efficiency-add}, we first present Lemma~\ref{lem:localerr}. 
\begin{lemma}\label{lem:localerr}
	If $\sum_{u\in N(v)}\theta(u,v) \le r_{\max}\cdot d(v)$ holds for all $v \in V$, the EdgePush algorithm (i.e., Algorithm~\ref{alg:edge-search}) achieves a normalized additive error $r_{\max}$ for every node. 
\end{lemma}
As shown in Lemma~\ref{lem:localerr}, to achieve a normalized additive error within $r_{\max}$ for $\forall v \in V$ for SSPPR queries, we can further decompose the original error constraint into a finer-grained local level. 

\begin{proof}[Proof of Lemma~\ref{lem:localerr}]
Recall that in Lemma~\ref{lem:EdgePushErr-ad}, we bound the normalized additive error of \edgepush by $\frac{1}{d(t)}\cdot \hspace{-1mm}\sum\limits_{v\in V}\sum\limits_{u\in N(v)}\hspace{-1mm}\theta(u,v)\cdot \vpi_v(t)$ for any node $t\in V$. By ensuring $\sum_{u\in N(v)}\theta(u,v) \le r_{\max}\cdot d(v)$, the normalized additive error for $t$ can be further bounded by
\begin{align*}
\frac{1}{d(t)}\cdot \hspace{-1mm}\sum_{v\in V}\hspace{-0.5mm}r_{\max}\cdot d(v)\cdot \vpi_v(t) = \frac{1}{d(t)}\cdot r_{\max}\cdot \hspace{-1mm} \sum_{v\in V} d(t)\cdot \hspace{-1mm} \vpi_t(v)=r_{\max}. 
\end{align*}
In the first equality, we apply the property of SSPPR on unweighted graphs that $d(v)\cdot \vpi_v(t)=d(t)\cdot \vpi_t(v)$, 
while the last equality is due to the fact of SSPPR that $\sum_{v\in V}\vpi_t(v)=1$. 
The lemma follows. 
\end{proof}


Note that Lemma~\ref{lem:localerr} actually presents a local constraint on the approximation error. Likewise, we can decompose the overall time cost into a local level. More precisely, let $Cost(v)$ denote the expected running time of all the \epush operations on $\forall \la u,v \ra \in \bar{E}$. As shown in the proof of Lemma~\ref{lem:EdgePushCost}, the expected number of \epush operations on edge $\la u,v \ra$ is at most $\frac{(1-\alpha)\cdot \A_{uv}}{\alpha \|\A\|_1 \cdot \theta(u,v)}$. Thus, $Cost(v)$ can be bounded as $Cost(v) \le \sum_{u \in N(v)} \frac{(1-\alpha)\cdot \A_{uv}}{\alpha \|\A\|_1 \cdot \theta(u,v)}$. 

Based on Lemma~\ref{lem:localerr} and the bound of $Cost(v)$, we can present the proof of Theorem~\ref{thm:edge-efficiency-add}. 

\begin{proof}[Proof of Theorem~\ref{thm:edge-efficiency-add}]
First, it can be verified that by setting $\theta(u,v)=\frac{r_{\max} \cdot d(v)\sqrt{\A_{uv}}}{\sum_{x\in N(v)}\sqrt{\A_{xv}}}$ for each $\la u,v \ra\in \bar{E}$ suggested by Theorem~\ref{thm:edge-efficiency-add}, the local error constraint required by Lemma~\ref{lem:localerr} (i.e. $\sum_{u\in N(v)}\theta(u,v) \le r_{\max}\cdot d(v)$) is satisfied. 
Second, by substituting the setting of $\theta(u,v)$ in the bound of $Cost(v)$ (i.e. $Cost(v) \le \sum_{u \in N(v)} \frac{(1-\alpha)\cdot \A_{uv}}{\alpha \|\A\|_1 \cdot \theta(u,v)}$), we have $Cost(v) = \frac{(1 - \alpha)}{\alpha \cdot r_{\max} \|A\|_1} \cdot \frac{\left( \sum_{x \in N(v)} \sqrt{A_{xv}} \right)^2}{d(v)}$, leading to the expected overal running time of \edgepush as 
$O(\sum_{v \in V} Cost(v)) = O\left( \frac{(1 - \alpha)}{\alpha \cdot r_{\max} \|A\|_1} \cdot \sum_{v \in V} \frac{\left( \sum_{x \in N(v)} \sqrt{A_{xv}} \right)^2}{d(v)} \right)$. And the theorem follows. 

Likewise, it can be verified with this setting of $\theta(u,v)$ suggested by Theorem~\ref{thm:edge-efficiency-add}, $\e_v \cdot Cost(v)$
is indeed minimized, where $\e_v=\sum_{u\in N(v)}\theta(u,v)$. Specifically, we have
\begin{align}\nonumber
\e_v \hspace{-0.5mm}\cdot \hspace{-0.5mm} Cost(v)\hspace{-1mm}=\hspace{-1mm}\left(\sum_{u\in N(v)}\hspace{-3mm}\theta(u,v)\hspace{-0.5mm} \right)  \cdot \left(\hspace{-0.5mm}\sum_{u \in N(v)} \hspace{-3mm}\frac{(1-\alpha)\hspace{-0.5mm}\cdot \hspace{-0.5mm}\A_{uv}}{\alpha \|\A\|_1 \hspace{-0.5mm} \cdot \hspace{-0.5mm}\theta(u,v)}\hspace{-0.5mm}\right) \hspace{-1mm}\ge \hspace{-1mm} \frac{(1\hspace{-0.5mm}-\hspace{-0.5mm}\alpha)}{\alpha \|\A\|_1}\hspace{-0.5mm}\cdot\hspace{-0.5mm} \left(\hspace{-0.5mm}\sum_{u\in N(v)}\hspace{-3.5mm}\sqrt{\A_{uv}}\right)^2. 
\end{align}

\end{proof}

\vspace{-2mm}
\subsection{Proof of Lemma~\ref{lem:cos-l1} 
}

To prove Lemma~\ref{lem:cos-l1}, we first introduce a new form of Cauchy-Schwarz Inequality, which offers a geometric explanation for the Cauchy-Schwarz Inequality.

\begin{fact}[Geometry Explanation of Cauchy-Schwarz Inequality~\cite{steele2004cauchy}]\label{fact:cauchy2}
	Given two vectors $\z=\{\z(1),\z(2),...,\z(m)\}\in \mathbb{R}^{m}$, $\x=\{\x(1),\x(2),...,\x(m)\}\in \mathbb{R}^{m}$, the Cauchy-Schwarz Inequality states that:
	\begin{align}\label{eqn:cauchy-vec-2}
		\la \z,\x\ra^2 = \|\z\|^2 \cdot \|\x\|^2 \cdot \cos^2 \p \le \|\z\|^2\cdot \|\x\|^2, 
	\end{align}
	where $\la \z,\x\ra$ denotes the inner product of vectors $\z$ and $\x$. The angle $\p$ is defined as the angle between $\z$ and $\x$. The equality holds if the vector $\z$ is in the same or opposite direction as the vector $\x$, or if one of them is the zero vector. 
\end{fact}

Based on Fact~\ref{fact:cauchy2}, we can present the proof of Lemma~\ref{lem:super-l1}. Note that Equation~\eqref{eqn:impro-l1} can be also expressed as 
\begin{align}\label{eqn:cauchy-3}
\vspace{-2mm}
  \left(\sum_{\la u,v \ra \in \bar{E}} \sqrt{\A_{uv}}\right)^2= 2m \cdot \|\A\|_1 \cdot \cos^2 \p.  
\vspace{-2mm}
\end{align}
To prove the correctness of Equation~\eqref{eqn:cauchy-3}, we can rewrite the right hand side of Equation~\eqref{eqn:cauchy-3} as: 
\begin{align*}
\vspace{-2mm}
    &2m \cdot \|\A\|_1 \cdot \cos^2 \p  =\left( \sum_{\la u,v \ra \in \bar{E}}1 \right)\cdot \left( \sum_{\la u,v \ra \in \bar{E}}\A_{uv} \right) \cdot \cos^2 \p \\
    &= \|\x\|^2 \cdot \|\z\|^2 \cdot \cos^2 \p = \la \x, \z \ra^2 =  \left(\sum_{\la u,v \ra \in \bar{E}} \sqrt{\A_{uv}}\right)^2, 
\vspace{-2mm}
\end{align*}
following Equation~\eqref{eqn:cauchy-3}. Note that in the first equality, we apply the fact that $2m=\sum_{\la u,v \ra \in \bar{E}}1$ and $\|\A\|_1=\sum_{\la u,v \ra \in \bar{E}}\A_{uv}$. In the second equality, we use the property of vectors $\z$ and $\x$ that $\|\x\|^2=\sum_{\la u,v \ra \in \bar{E}}1$ and $\|\z\|^2=\sum_{\la u,v \ra \in \bar{E}}\A_{uv}$. Note that we present the definitions of vectors $\x$ and $\z$ in Definition~\ref{def:vec}. And in the third equality, we apply Fact~\ref{fact:cauchy2} (i.e. Equation~\eqref{eqn:cauchy-vec-2}). Finally, in the last equality, we use the definition of the inner product between the vectors $\x$ and $\z$. Thus, the lemma follows. 

\subsection{Proof of Lemma~\ref{lem:cos-add} 
}
Analogous to the proof of Lemma~\ref{lem:cos-l1}, we can rewrite the equation in Lemma~\ref{lem:cos-add} as: \begin{align}\label{eqn:super-proof-rewrite}
\vspace{-2mm}
\sum_{v\in V}\frac{\left(\sum_{x\in N(v)}\sqrt{\A_{xv}}\right)^2 }{d(v)}= \sum_{v\in V} n(v) \cdot \cos^2 \p_v, 
\vspace{-2mm}
\end{align}
We note that if we can show 
\begin{align}\label{eqn:super-proof-rewrite2}
\vspace{-2mm}
\left(\sum_{x\in N(v)}\sqrt{\A_{xv}}\right)^2 = d(v)\cdot n(v) \cdot \cos^2 \p_v    
\vspace{-2mm}
\end{align}
holds for each node $v\in V$, then Equation~\eqref{eqn:super-proof-rewrite} follows. By rewriting the right hand side of Equation~\eqref{eqn:super-proof-rewrite2}, we can derive: 
\begin{align*}
\vspace{-2mm}
d(v)\cdot n(v)\cdot \cos^2 \p_v=\left(\sum_{x\in N(v)}\hspace{-2.5mm}\A_{xv}\right)\cdot \left(\sum_{x\in N(v)}\hspace{-2mm}1\right) \cdot \cos^2 \p_v=\la \z_v,\x_v\ra^2. 
\vspace{-2mm}
\end{align*}
In the first equality, we apply the fact that $d(v)=\left(\sum_{x\in N(v)}\A_{xv}\right)$ and $n(v)=\left(\sum_{x\in N(v)}1\right)$. In the last equality, we use the property of vectors $\z_v$ and $\x_v$ that $\left(\sum_{x\in N(v)}\A_{xv}\right)=\|\z_v\|^2$ and $\sum_{x\in N(v)}1=\|\x_v\|^2$, where vectors $\z_v$ and $\x_v$ are defined in Definition~\ref{def:vec}. Additionally, in the last equality, we also apply the equation shown in Fact~\ref{fact:cauchy2}. Note that $\la \z_v,\x_v\ra^2=\left(\sum_{x\in N(v)}\sqrt{\A_{xv}}\right)^2$ by definition. Thus, Equation~\eqref{eqn:super-proof-rewrite2} follows, as well as Lemma~\ref{lem:cos-add}. 

\subsection{Proof of Lemma~\ref{lmm:ab}}
Recall that a node $v$ is $(a,b)$-unbalanced if $a\cdot n(v)$ of its edges taking $b\cdot d(v)$ edge weights. Motivated by this definition, we can divide the neighbors of node $\forall v\in V$ into two sets: $N^L(v)$ and $N^S(v)$. Specifically, the set $N^L(v)$ contains $a\cdot n(v)$ neighbors of $v$ which take $b\cdot d(v)$ fraction of $v$'s edge weights. Analogously, the set $N^S(v)$ contains the other neighbors that $N^S(v)=N(v)\setminus N^L(v)$. 
Thus, to bound $\sum_{v\in N(u)}\sqrt{\A_{uv}}\,$, we have 
\begin{align}\label{eqn:ab-sum}
\sum_{u\in N_v}\sqrt{\A_{uv}} =\sum_{u\in N^L(v)}\sqrt{\A_{uv}}+\sum_{u\in N^S(v)}\sqrt{\A_{uv}}.    
\end{align}
Moreover, according to the AM-GM Inequality, we can derive: 
\begin{align*}
\sum_{u\in N^L(v)}\hspace{-2mm}\sqrt{\A_{uv}}\le \left(a \cdot n(v) \right)\cdot \sqrt{\frac{\sum_{u\in N^L(v)}\A_{uv}}{a\cdot n(v)}} = \sqrt{a \cdot n(v) \cdot b \cdot d(v)}. 
\end{align*}
Likewise, $\sum_{u\in N_v^S}\sqrt{\A_{uv}}\le \sqrt{(1-a) n(v) \cdot (1-b) d(v)}$. Plugging the two inequalities into Equation~\eqref{eqn:ab-sum}, we have: 
\begin{align}\nonumber
\vspace{-2mm}
\sum_{u\in N_v}\hspace{-1mm}\sqrt{\A_{uv}} 
\le \left(\sqrt{ab}+\sqrt{(1-a)(1-b)}\right)\cdot \sqrt{n(v)  d(v)}\,, 
\vspace{-2mm}
\end{align}
which follows the lemma. 

\subsection{Proof of Lemma~\ref{lem:super-l1}}
We observe that the inequality shown in Lemma~\ref{lmm:ab} holds for each node $v\in V$. Thus, we can derive: 
\begin{align}\label{eqn:ab-ineq1}
\vspace{-4mm}
\left(\sum_{v\in V}\sum_{u \in N(v)}\hspace{-2mm}\sqrt{\A_{uv}}\right)^2
\hspace{-2mm}\le \left(\sqrt{ab} +\sqrt{(1-a)(1-b)}\right)^2 \cdot \left(\sum_{v\in V}\sqrt{n(v) d(v)}\right)^2.     
\vspace{-4mm}
\end{align}
Moreover, we apply Cauchy-Schwarz Inequality shown in Fact~\ref{fact:cauchy} to further bound $\left(\sum_{v\in V}\sqrt{n(v) d(v)}\right)^2$. Specifically, we have
\begin{align*}
\vspace{-2mm}
    \left(\sum_{v\in V}\sqrt{n(v) d(v)}\right)^2 \le \left(\sum_{v\in V}n(v)\right)\cdot \left(\sum_{v\in V}d(v)\right)=(2m)\cdot \|\A\|_1. 
\vspace{-2mm}
\end{align*}
Plugging into inequality~\eqref{eqn:ab-ineq1}, we can derive
\begin{align*}
\vspace{-2mm}
\left( \sum_{\la u,v \ra \in \bar{E}\in V}\hspace{-2mm}\sqrt{\A_{uv}}\right)^2
\hspace{-2mm}\le \left(\sqrt{ab} +\sqrt{(1-a)(1-b)}\right)^2 \cdot 2m \cdot \|\A\|_1,      
\vspace{-2mm}
\end{align*}
or equivalently, 
\begin{align*}
\frac{(1-\alpha)}{\alpha \e \|\A\|_1}\cdot \left(\sum_{\la u,v \ra \in \bar{E}}\sqrt{\A_{uv}}\right)^2 
\hspace{-2mm}\leq \left(\sqrt{ab} +\sqrt{(1-a)(1-b)}\right)^2 \cdot \frac{2m}{\alpha \e}\,,
\end{align*}
which follows the lemma. 

\subsection{Proof of Lemma~\ref{lem:super-add}}
Analogously, the inequality shown in Lemma~\eqref{lmm:ab} can be rewritten as
\begin{align*}
\left(\sum_{u \in N(v)}\hspace{-2mm}\sqrt{\A_{uv}}\right)^2
\hspace{-2mm}\le \left(\sqrt{ab} +\sqrt{(1-a)(1-b)}\right)^2 \cdot \left(\sqrt{n(v) d(v)}\right)^2. 
\end{align*}
It follows:  
\begin{align*}
\vspace{-2mm}
\frac{1}{d(v)}\left(\sum_{u \in N(v)}\hspace{-2mm}\sqrt{\A_{uv}}\right)^2
\hspace{-2mm}\le \left(\sqrt{ab} +\sqrt{(1-a)(1-b)}\right)^2 \cdot n(v). 
\vspace{-2mm}
\end{align*}
Note that the above inequality holds for each $v\in V$. Thus, we have
\begin{align*}
\vspace{-2mm}
\sum_{v\in V}\frac{1}{d(v)}\left(\sum_{u \in N(v)}\hspace{-2mm}\sqrt{\A_{uv}}\right)^2
\hspace{-2mm}\le \left(\sqrt{ab} +\sqrt{(1-a)(1-b)}\right)^2 \cdot \sum_{v\in V}n(v), 
\vspace{-2mm}
\end{align*}
where $\sum_{v\in V}n(v)=2m$ by definition, which follows the lemma. 

\end{document}